\documentclass[11pt]{article}
\pdfoutput=1

\usepackage{./preamble}

\title{Quantum Sensing and Hamiltonian Learning under Stochastic Parameter Evolution}

\author[1]{Kelvin Koor\thanks{\tt kelvinkoor@u.nus.edu}}
\author[1,2]{Patrick Rebentrost\thanks{\tt cqtfpr@nus.edu.sg}}
\affil[1]{Centre for Quantum Technologies, NUS}
\affil[2]{School of Computing, NUS}
\date{\today}

\begin{document}
\maketitle

\begin{abstract} 
We introduce the notion of stochastically evolving Hamiltonians (SEHs), which are local Hamiltonians whose coefficients are governed by stochastic differential equations (SDEs). After developing some theoretical aspects of these objects, we consider the task of Hamiltonian learning for SEHs. We reanalyse and make necessary extensions to the Huang-Tong-Fang-Su protocol, which first achieved the Heisenberg limit for Hamiltonian coefficient learning, under these circumstances. We delineate theoretically the regimes where the Heisenberg limit can and cannot be retained. Along the way we develop a few tools in quantum sensing and Hamiltonian simulation that may be of independent interest under similar settings in quantum information processing. The broad aim is to encourage the exploration of various settings in quantum information processing in which the control parameters are governed by SDEs.
\end{abstract}

\newpage
\tableofcontents

\newpage
\section{Introduction}\label{section:introduction}

Stochastic differential equations (SDEs) \cite{baldi2017stochastic,bass2011stochastic,cohen2015stochastic,le2016brownian,evans2012introduction,oksendal2013stochastic} provide a mathematical means for describing dynamical systems subject to random perturbations. As such, they are ubiquitous in and beyond science and engineering. Example applications include finance \cite{shreve2004stochastic,jeanblanc2009mathematical,eberlein2019mathematical}, control and filtering \cite{bain2009fundamentals,xiong2008introduction,van2007stochastic,touzi2012optimal}, and more recently, generative AI \cite{song2020score,huang2021variational,dockhorn2021score,de2021diffusion,winther2026introduction}. It is natural then to consider the setting where such systems, which undergo \textit{stochastic evolution}, appear in quantum information processing. For instance, the Hamiltonian governing a quantum system depends on physical quantities such as 
coupling strengths and external fields, which may fluctuate during the execution of a protocol. How would this affect the various quantum information processing tasks involving such a quantum system?

Motivated by this question, in this work we investigate the task of Hamiltonian learning (from real-time dynamics), where the Hamiltonian involved is a local Hamiltonian whose parameters are governed by SDEs. Our choice of task is motivated by the central importance of Hamiltonian learning \cite{huangtong2023learning,hu2025ansatz,ma2024learning,bakshi2024structure,zhao2025learning,dutkiewicz2024advantage,flammia2026autonomous,shin2026heisenberg,stilck2024efficient,zubida2021optimal,bairey2019learning,francca2025learning,jimenez2026rigorous} in quantum technologies, such as quantum sensing/metrology and quantum computing. In quantum sensing \cite{degen2017quantum,giovannetti2011advances,toth2014quantum,szczykulska2016multi,gardner2025stochastic,gardner2026quantum,allen2025quantum,gong2026robust,wang2026entanglement,brady2026precision,ott2026rare,prabhu2026exponential}, the signals of interest are encoded in Hamiltonians. Learning these Hamiltonians efficiently enables us infer the signals faster, potentially improving the performance of a broad range of sensing applications. In quantum computing, accurately characterizing the unknown Hamiltonian is essential for device benchmarking, calibration and control \cite{carrasco2021theoretical,eisert2020quantum,innocenti2020supervised,shulman2014suppressing}, thereby informing the design and operation of quantum computers with lower error rates.

The central objects in this paper are Hamiltonians taking the form
\begin{equation}
\begin{aligned}
    H(t) = \sum_{i=1}^m \delta_i(t)H_i, \qquad d\bm{\delta}(t) = \bm{F}(\bm{\delta}(t),t)\,dt + \bm{G}(\bm{\delta}(t),t)\,d\bm{W}(t).
\end{aligned}
\end{equation}
Here $H(t)$ is a local Hamiltonian, i.e., the $H_i$'s have $\Theta(1)$-sized support. Informally, the $\bm{F}$ term characterizes the \textit{deterministic} evolution of $\bm{\delta}(t)$, and the $\bm{G}$ term characterizes the \textit{stochastic} evolution of $\bm{\delta}(t)$ driven by the Wiener process $\bm{W}(t)$. The time-ordered unitary propagator associated to $H(t)$ is $U(t) = \exp_\mathcal{T} \left( -i\int_0^t H(\tau)\,d\tau \right)$. Because of the randomness introduced by $d\bm{W}(t)$, $H(t)$ and $U(t)$ are random variables, or more precisely, (operator-valued) \textit{stochastic processes}. We shall call such Hamiltonians \textit{stochastically evolving Hamiltonians} (SEHs), and their propagators \textit{stochastically evolving unitaries} (SEUs). The (time-evolving) distribution of the parameters $\bm{\delta}(t)$ thus determines a distribution over $U(t)$ and through averaging, a quantum channel describing the dynamics of the quantum system. For non-commuting interaction terms, these dynamics generally depend on the entire parameter trajectory.

Having defined SEHs, we proceed to investigate a few quantum information processing tasks involving SEHs. We state the main results (on an informal level) in the next subsection.

\subsection{Setting and main results}

The Hamiltonians we consider in this paper are local Hamiltonians whose parameters are governed by stochastic differential equations. We formulate this notion precisely in \cref{definition:SEH/SEU}. For the convenience of the reader we also provide a concise but informative appendix on stochastic differential equations in \cref{appendix:stochastic_processes}. With SEHs defined, we first build up some theoretical aspects of these objects to aid in intuition and understanding. In \cref{proposition:channel_description_pdf_SEH/SEU} we describe analytically the quantum channel that describes the evolution of a quantum system governed by an SEH. The form is unwieldy in the general case where the interaction terms of the Hamiltonian are non-commuting, but simplifies considerably if they commute. We end the section with a full description of the baby case of a single-qubit SEH, given in \cref{example:channel_description_pdf_single_qubit_SEH/SEU}.

Having laid the groundwork, we proceed to investigate a few quantum information processing tasks involving SEHs. Our primary goal is SEH learning, i.e., how to learn the coefficients of an SEH. Our starting point is the paper of Huang et al.~\cite{huangtong2023learning}. We analyse how the presence of noise characterised by SDEs affects their protocol, and how to account for this type of stochasticity. The problem is stated formally in \cref{problem:SEH_Hamiltonian_learning}, which we reproduce here for the convenience of the reader.

\begin{tcolorbox}[
    colback=teal!5!white,
    colframe=teal!75!black,
]
\textbf{\cref{problem:SEH_Hamiltonian_learning}.} Given the low-intersection local SEH $H(t)= \sum_{i=1}^m \delta_i(t)H_i$, where the local parameters $\delta_i(t)$ are governed by the SDE
\begin{equation}\begin{aligned}
    d\bm{\delta}(t) &= \bm{G}\,d\bm{W}(t)\\
    \bm{\delta}(0) &= \bm{\delta}_0.
\end{aligned}\end{equation}
Equivalently, in component form this is
\begin{equation}\begin{aligned}
    d\delta_{i}(t) &= \sum_{j=1}^\chi G_{ij}\,dW_{j}(t) \quad \text{for } i=1,\dots,m\\
    \delta_{i}(0) &= \delta_{i0}.
\end{aligned}\end{equation}
Here, $\bm{\delta}_0$ and $\bm{G}$ are unknown, but it is given that $|\delta_{i0}| \leq \delta_{\text{max}}$ for all $i$, and $0 \leq G_{ij} \leq G_{\text{max}}$.
\\

\textbf{Goal.} Learn $\bm{\delta}_0$.
\end{tcolorbox}

For simplicity, we have considered the SDE where the drift matrix is $\bm{F}=\bm{0}$ and the diffusion matrix $\bm{G}$ is constant (does not depend on $\bm{\delta}(t)$ nor $t$). The protocol of \cite{huangtong2023learning} involves two key subroutines that are based on the tasks of single-qubit quantum sensing and Hamiltonian simulation. As such, they are investigated first. 

We begin with single-qubit quantum sensing for SEHs. The problem is described in \cref{problem:single_qubit_quantum_sensing} which we again reproduce here.

\begin{tcolorbox}[
    colback=teal!5!white,
    colframe=teal!75!black,
]
\textbf{\cref{problem:single_qubit_quantum_sensing}.} Given the single-qubit Hamiltonian $H(t)= \delta(t)\sigma_z$, where
\begin{equation}\begin{aligned}
    d\delta(t) &= G\,dW(t)\\
    \delta(0) &= \delta_0.
\end{aligned}\end{equation}
Here, $\delta_0$ and $G$ are unknown. It is known however that $|\delta_0| \leq \delta_{\text{max}}$ and $0 \leq G \leq G_{\text{max}}$.
\\

\textbf{Goal.} Learn $\delta_0$.
\end{tcolorbox}

Again for simplicity we let the drift term $F$ be zero and the diffusion term $G$ be constant. We have the following results:
\begin{enumerate}
    \item \cref{algorithm:extended_ramsey_protocol} extends the classic Ramsey protocol, allowing still the standard-quantum-limited estimation of $\delta_0$ when it is governed by SDEs.

    \item \cref{algorithm:extended_rfe} extends the Robust Frequency Estimation protocol (which is a key enabler of the Heisenberg-limited protocol in \cite{huangtong2023learning}), allowing for the estimation of $\delta_0$ when it is governed by SDEs. The Heisenberg limit is retained when the diffusion strength is sufficiently small, $G \lesssim \varepsilon^{3/2}$. 
    
    As a byproduct, we further note that in the case of \textit{static disorder} (with Gaussian pdf), \cref{algorithm:extended_rfe} gives us a \textit{Heisenberg-limited} estimation protocol even in the presence of noise of strength $G = O(1)$ (where $G^2$ is the variance of the Gaussian pdf, which characterises the strength of the disorder).
    
    \item \cref{algorithm:strong_extended_rfe} is a stronger variant of \cref{algorithm:extended_rfe}, which does not make the assumption of knowing $G_{\text{max}}$ a priori. As in \cref{algorithm:extended_rfe}, the Heisenberg limit is retained in \cref{algorithm:strong_extended_rfe} when the diffusion strength is sufficiently small, $G \lesssim \varepsilon^{3/2}$, and for $G=O(1)$ in the case of static disorder.
    
    \item \cref{table:quantum_sensing_algorithms} collects these three protocols and compares their resource complexities in detail, namely the maximum sample runtime, number of samples, and total runtime.
\end{enumerate}

Next we look at the task of Hamiltonian simulation. The relevant protocol to extend/analyse here is qDRIFT, which is used to implement `Hamiltonian reshaping', the second key enabler of the HTFS protocol. The problem is as follows:
\begin{tcolorbox}[
    colback=teal!5!white,
    colframe=teal!75!black,
]
\textbf{Problem.} Given the Hamiltonian $H(t)$ in \cref{problem:SEH_Hamiltonian_learning}, and the target Hamiltonian
\begin{equation}\begin{aligned}
    \widetilde{H}(t) = \sum_j w_j U_jH(t)U_j^\dag,
\end{aligned}\end{equation}
where the $U_j$ are unitaries. Assume the primitives are the ability to implement the time-dependent evolutions $\exp_\mathcal{T} \left( -i\int H(\tau)\,d\tau \right)$, and the unitaries $U_j$.
\\

\textbf{Goal.} Simulate $\widetilde{H}(t)$ based on the given primitives.
\end{tcolorbox}

We build things up one at a time, first incorporating time-dependence in $H(t)$, then stochasticity. The relevant results here are: 
\begin{enumerate}
    \setcounter{enumi}{4}
    \item \cref{algorithm:time-dependent_qDRIFT}, \cref{result:time-dependent_qDRIFT} extends qDRIFT to $H(t)$ which are time-dependent (but still deterministic). Note that this differs from existing protocols in the literature due to different allowed primitives. Operationally \cref{algorithm:time-dependent_qDRIFT} is similar to the original qDRIFT protocol, but now with a different primitive the analysis is more involved. 

    \item \cref{algorithm:SEH_qDRIFT}, \cref{result:SEH_qDRIFT} extends qDRIFT to SEHs. As with \cref{algorithm:time-dependent_qDRIFT}, it is operationally similar to the original qDRIFT but with SDEs in play the analysis is considerably more involved. It turns out that due to the diffusive term $\bm{G}$ driving the SDE, for a desired simulation precision $\varepsilon_{\text{qDRIFT}}>0$ the number of sampling steps required in qDRIFT (in this case, also the number of layers of interleaving unitaries $U_k$), $N_{\text{qDRIFT}}$ accrues an overhead factor of $O(t\chi^2)$.

    \item \cref{result:pointwise_SEH_qDRIFT}. \cref{result:SEH_qDRIFT} analyses the error bound for the simulation of SEHs, when \textit{viewed as a channel}, i.e. as an \textit{expectation} over many runs. Here we further provide a \textit{pointwise} error bound (with high probability), i.e. for a single run of the SEH. This result is not necessary for the SEH learning task however.
\end{enumerate}

Having made the necessary extensions to single-qubit sensing and Hamiltonian simulation (which may be of independent interest), we return to the main goal of SEH learning, \cref{problem:SEH_Hamiltonian_learning}. The high-level strategy of \cite{huangtong2023learning} remains unchanged, and our main task is to integrate the modified subroutines above to account for stochastic evolution.
\begin{enumerate}
    \setcounter{enumi}{7}
    \item \cref{result:SEH_Hamiltonian_learning} gives the complexity of this modified protocol, showing in particular that the Heisenberg limit is retained if $G_{\text{max}} \lesssim \left(\frac{\varepsilon^3}{2^{k}\chi}\right)^{1/2}$.
\end{enumerate}

\subsection{Discussion and outlook}

The purpose of this work is manifold. While we have focused much effort on the specific task of SEH learning (with SEH sensing and simulation as necessary byproducts), our broader aim is to encourage the exploration of other settings in quantum information processing in which the control parameters are governed by SDEs. That is the raison d'\^{e}tre of \cref{section:SEHs} and \cref{appendix:stochastic_processes}. We also hope that the insights and techniques developed in \cref{section:quantum_sensing,section:Hamiltonian_simulation,section:Hamiltonian_learning} would provide a foundation for the study of these other settings as well. To begin, we point out a few concrete directions for future work:
\begin{enumerate}
    \item \textbf{Fermionic/Bosonic SEHs.} Above our SEHs were defined for qubit Hamiltonians. An analogous definition could naturally be made for fermionic/bosonic Hamiltonians as well. In what settings do such fermionic/bosonic SEHs appear, and what are some natural tasks related to them?

    \item \textbf{Structure learning of SEHs.} The \textit{structure} learning of Hamiltonians \cite{hu2025ansatz,bakshi2024structure,zhao2025learning} has been an active area of research, with a few protocols successfully attaining the Heisenberg limit. How would the stochastic evolution of the Hamiltonian parameters affect these protocols?

    \item \textbf{Stochastically evolving Lindbladians.} Lindbladian learning \cite{arad2026near,chen2026learning,romanov2026learning,ivashkov2026ansatz,mobus2026robust,lewis2026learning} is currently also an extremely active area of research. It would be interesting to see how stochastic evolution in the dissipative part of the Lindbladian affects current protocols. The stochastic evolution could pertain to the Lindblad operators, dissipation rates, or both.

    \item \textbf{AC quantum sensing under stochastic evolution.} AC sensing, the task of detecting oscillating fields, is ubiquitous in science and engineering, with applications ranging from nuclear spectroscopy to gravitational wave detection to the searching for axionic dark matter. If the oscillating signal is governed by an SDE, how does this affect the quantum sensing protocols \cite{allen2025quantum,iemini2024floquet,mishra2022integrable,gribben2024boundary} attempting to learn the signal?

    \item \textbf{SEHs as a resource.} In this paper, we have largely adopted the view that stochasticity is an impediment in quantum information processing to be overcome. However, could this type of stochasticity be harnessed as a \textit{resource} in quantum information processing? In particular, \cite{onorati2017mixing} exploited the stochastic time-dependence (formulated in a different way from our SEHs) in their Hamiltonian to construct approximate unitary designs. For what other tasks can SEHs be utilised as a resource?
    
    \item \textbf{Other driving noise processes.} In this paper our SDEs are all driven by Wiener processes. There exists many other types of SDES, depending on the underlying driving processes. Examples of such processes include the L\'{e}vy process, Poisson process, and more adventurously, the \textit{fractional} Wiener process/Brownian motion (fBm). When do these show up in quantum information processing? In particular, the fBm is a (relatively analytically tractable) \textit{non-Markovian} process and may be relevant in certain settings.

\end{enumerate}

\section{Stochastically Evolving Hamiltonians and Unitaries}\label{section:SEHs}
In this section we first make precise the notion of stochastically evolving (local) Hamiltonians, and the corresponding unitaries they generate. Since these are the central objects in this article, we shall develop some aspects of the attendant framework. Not all the results stated (in their most general form) will be relevant to the Hamiltonian learning task below -- the primary purpose of the theory introduced here is to build up intuition and understanding in this paper, and to lay the groundwork for future investigations of quantum information processing tasks involving such objects.

First, recall that an $n$-qubit $k$-\textbf{local Hamiltonian} takes the form
\begin{equation}\label{equation:local_Hamiltonian}
\begin{aligned}
    H = \sum_{i\in E} \delta_iH_i
\end{aligned}
\end{equation}
where $E$ indexes the set of local interaction terms, each $H_i = H_{S_i} \otimes I_{\backslash {S_i}}$ acts nontrivially on at most $|S_i|=k$ qubits, and $m \coloneq |E| \leq \binom{n}{k} \leq n^k = O(\poly n)$ if $k=O(1)$. We can assume wlog that each $H_i$ is a ($k$-local) Pauli string since this incurs only an $O(1)$ local term overhead.

\begin{definition}\label{definition:SEH/SEU}
    A \textbf{stochastically evolving Hamiltonian (SEH)} takes the form
    \begin{equation}
    \begin{aligned}
        H(t) = \sum_{i=1}^m \delta_i(t)H_i
    \end{aligned}
    \end{equation}
    where the local parameters $\delta_i(t)$ are governed by SDEs (\cref{subsection:SDEs}). In the most general form, this is given by
    \begin{equation}\begin{aligned}
        d\bm{\delta}(t) &= \bm{F}(\bm{\delta}(t),t)\,dt + \bm{G}(\bm{\delta}(t),t)\,d\bm{W}(t)\\
        \bm{\delta}(0) &= \bm{\delta}_0.
    \end{aligned}\end{equation}
    Note that the parameters $\delta_i(t)$ are in general correlated, since they could be driven by overlapping Wiener processes $W_j(t)$'s.

    The \textbf{stochastically evolving unitary (SEU)} generated by $H(t)$ is the unitary satisfying the Schr\"{o}dinger equation $\frac{dU(t)}{dt} = -iH(t)U(t)$, i.e.
    \begin{equation}\label{equation:time-ordered_SEU}
    \begin{aligned}
        U(t) = \exp_\mathcal{T} \left( -i\int_0^t H(\tau)\,d\tau \right).
    \end{aligned}
    \end{equation}
    Because of the stochasticity of the $\delta_i(t)$'s, $H(t)$ and $U(t)$ are random variables. More precisely, they are \textit{operator-valued} stochastic processes.
\end{definition}

\begin{remark}\label{remark:potential_technicality}
Before discussing further let us first clear a potential technicality. We note that even with randomness the time-ordered unitary formalism still holds. That is, for a fixed $\omega \in \Omega$ we have that the solution to
\begin{equation}
\begin{aligned}
    \frac{dU(t,\omega)}{dt} = -iH(t,\omega)U(t,\omega).
\end{aligned}
\end{equation}
is still given by
\begin{equation}\label{equation:time-ordered_SEU_pointwise}
\begin{aligned}
    U(t,\omega) = \exp_\mathcal{T} \left( -i\int_0^t H(\tau,\omega)\,d\tau \right).
\end{aligned}
\end{equation}
Indeed, \cref{equation:time-ordered_SEU} is an equality of random variables\footnote{In practice we often write random variables without explicit dependence on $\omega$. When $\omega$ is explicitly written, it is usually during arguments where the pointwise properties of a random variable are being emphasized. We refer the reader to \cref{appendix:stochastic_processes} for the clarification of such subtleties.\label{footnote:on_rv}}, and its meaning is that \cref{equation:time-ordered_SEU_pointwise} holds for all $\omega$.

The reason this works without any further complications one might expect from stochastic calculus, such as It\^{o} correction terms (\cref{theorem:Ito_lemma}) if the chain rule for \textit{differentiation} is ever invoked, is because to arrive at this analytical form we apply Picard iteration, which ultimately entails repeated \textit{integrations} of $H(t,\omega)$ with respect to time. While $\delta_i(t,\omega)$ and $H(t,\omega)$ are not well-behaved in the differentiability sense -- indeed they are nowhere differentiable and have unbounded total variation (\cref{subsection:Wiener_process}) -- they are nonetheless continuous and have no singularities. Thus their time-integrals exist and are well-defined, the same is true of the time-integrals of the time-integrals, ad infinitum.
\end{remark}

\begin{remark}\label{remark:realizations_of_quantum_objects}
    As mentioned, for each $t$ the quantum objects $H(t), U(t)$ and $\mathcal{L}(t)$ (this is the Liouvillian superoperator we shall encounter very soon) are random variables. Let us now introduce the following notation for the \textit{realizations} of $H(t)/\mathcal{L}(t)/U(t)$: if $\bm{\delta}(t,\omega) =\bm{x} \in \mathbb{R}^m$, i.e. $\bm{x}$ is a realization for $\bm{\delta}(t)$ corresponding to the outcome $\omega$, then we write $H(t,\omega) = H(\bm{x})$ and similarly for the Liouvillians and unitaries. That is, $H(\bm{x}) = \sum_i x_iH_i$ is no longer a random variable. Note that for Hamiltonians and Liouvillians the realizations are $H(\bm{x}), \mathcal{L}(\bm{x})$, but for a general time-ordered unitary its realization is $U(\gamma)$, where $\gamma$ is an $\mathbb{R}^m$-valued \textit{path}. This is because while $H(t)$ and $\mathcal{L}(t)$ depend on $\omega$ through the single random variable $\bm{\delta}(t,\omega)$ (with fixed $t$), $U(t)$ generally depends on $\omega$ through the entire \textit{path} $\bm{\delta}(\omega) = (\delta(\tau,\omega))_{\tau \in [0,t]}$.
\end{remark}

For a quantum state $\rho$, its evolution under a Hamiltonian $H$ is given by $\rho(t) = U(t)\rho U(t)^\dag$, where $U(t)$ is the unitary generated by $H$. In the setting of this paper, the Hamiltonian of interest is an SEH, i.e., it is both time-dependent and random. Thus, for a realization $\omega$, we have $\rho(t,\omega) = U(t,\omega)\rho U(t,\omega)^\dag$, which is a random quantity and generally differs for different realizations. This motivates the description of the evolution of $\rho$ under SEHs by means of quantum channels. We denote this channel by $\mathcal{S}_t(\cdot)$ (omitting for brevity its dependence on the specific local Hamiltonian interaction structure and SDE structure $\bm{F}, \bm{G}$):
\begin{equation}\label{equation:channel_description_SEH/SEU}
\begin{aligned}
    \mathcal{S}_t(\rho) \coloneq \E[U(t)\rho U(t)^\dag].
\end{aligned}
\end{equation}
A priori, this expectation is to be taken over $\Omega$, the underlying source of randomness in our quantum system. We can recast the expectation over state space instead, in the spirit of \cref{equation:expectations_over_state_space}. Recall that the random variables of interest $(\delta_i(\tau))_{\tau \in [0,t]}$ are stochastic processes, i.e., path-valued. We have
\begin{equation}\label{equation:channel_description_state_space}
\begin{aligned}
    \mathcal{S}_t(\rho) &= \E_\Omega[U(t,\omega)\rho U(t,\omega)^\dag]\\
    &= \E_{C([0,t],\mathbb{R}^m)}[U(\gamma)\rho U(\gamma)^\dag]\\
    &= \int_{C([0,t],\mathbb{R}^m)} U(\gamma)\rho U(\gamma)^\dag\,\mathbb{P}(d\gamma)
\end{aligned}
\end{equation}
where $\mathbb{P}$ here is the distribution of the stochastic process $\bm{\delta} = (\delta(\tau))_{\tau \in [0,t]}$ on $C([0,t],\mathbb{R}^m)$, which depends on the SDE governing $\bm{\delta}$, see \cref{equation:path-valued_random_variable}. While \cref{equation:channel_description_state_space} sheds some light conceptually on $\mathcal{S}_t$, it is hardly useful both in theory and computation.

Instead of taking expectations with respect to path measures as in \cref{equation:channel_description_state_space}, we can recast it with respect to the multi-time joint pdf of $\bm{\delta}$ instead, which can be obtained from the Fokker-Planck equation, \cref{equation:Fokker-Planck}. Furthermore for commuting Hamiltonians only the single-time pdf is required. We have the following result.
\begin{proposition}\label{proposition:channel_description_pdf_SEH/SEU}
    Let $\mathcal{S}_t$ be the quantum channel corresponding to the SEH $H(t)=\sum_{i=1}^m \delta_i(t)H_i$. The action of $\mathcal{S}_t$ on a state $\rho$ can be written as
    \begin{equation}
    \begin{aligned}
        \mathcal{S}_t(\rho) = \sum_{n=0}^\infty \int_{0 \leq \tau_n \leq \dots \leq \tau_1 \leq t} d\tau_1 \dots d\tau_n \left[ \int_{\mathbb{R}^{m \times n}} p(\bm{x_1},\tau_1; \dots ; \bm{x_n},\tau_n) \mathcal{L}(\bm{x_1})\dots \mathcal{L}(\bm{x_n}) \;d\bm{x_1} \dots d\bm{x_n} \right](\rho),
    \end{aligned}
    \end{equation}
    where 
    \begin{equation}
    \begin{aligned}
        p(\bm{x_1},\tau_1; \dots ; \bm{x_n},\tau_n) = \prod_{i=1}^{n-1} p(\bm{x_i},\tau_i|\bm{x_{i+1}},\tau_{i+1})\cdot p(\bm{x_n},\tau_n)
    \end{aligned}
    \end{equation}
    is the multi-time joint pdf for the random variables $(\bm{\delta}(\tau_1),\dots,\bm{\delta}(\tau_n))$ and the
    \begin{equation}
    \begin{aligned}
        \mathcal{L}(\bm{x}) = -i\sum_{j=1}^m x_j[H_j,\,\,\cdot\,\,]
    \end{aligned}
    \end{equation}
    are realizations of the Liouvillian superoperator (\cref{remark:realizations_of_quantum_objects}).

    Furthermore, if $H(t)$ is commuting, i.e., $[H_i,H_j]=0$ for all $i,j$, we have
    \begin{equation}
    \begin{aligned}
        \mathcal{S}_t(\rho) = \int_{\mathbb{R}^m} p_{\int}(\bm{x},t) \exp\left(-i \sum_{i=1}^m x_iH_i \right) \rho \exp\left(i \sum_{i=1}^m x_iH_i \right) d\bm{x}
    \end{aligned}
    \end{equation}
    where $p_{\int}(\bm{x},t) = p_{\int}(x_1,\dots, x_m,t)$ is the pdf of the \textit{integrated} process $\int_0^t \bm{\delta}(\tau)\,d\tau$.
\end{proposition}

\begin{proof}[Derivation.]
Consider the Dyson series expansion \cref{equation:dyson_series_unitary_channel} for the unitary channel $\mathcal{U}(t)$ defined by $\mathcal{U}(t)(\rho) \coloneq U(t)\rho U(t)^\dag$:
\begin{equation}
\begin{aligned}
    \mathcal{U}(t) &= \exp_\mathcal{T}\left( \int_0^t \mathcal{L}(\tau)\,d\tau \right)\\
    &= \sum_{n=0}^\infty \int_{0 \leq \tau_n \leq \dots \leq \tau_1 \leq t} d\tau_1 \dots d\tau_n\,\mathcal{L}(\tau_1)\dots \mathcal{L}(\tau_n)\\
    &= \sum_{n=0}^\infty (-i)^n \int_{0 \leq \tau_n \leq \dots \leq \tau_1 \leq t} d\tau_1 \dots d\tau_n\,
    [H(\tau_1),[H(\tau_2),\dots[H(\tau_n),\,\cdot\,]\dots]]
\end{aligned}
\end{equation}
where recall the Liouvillian superoperator $\mathcal{L}(t)$ at time $t$ is
\begin{equation}\begin{aligned}
    \mathcal{L}(t)(\rho) \coloneq -i[H(t),\rho] = -i\sum_{i=1}^m \delta_i(t)[H_i,\rho].
\end{aligned}\end{equation}
Substituting this into \cref{equation:channel_description_SEH/SEU}, we have
\begin{equation}
\begin{aligned}
    \mathcal{S}_t(\rho) &= \E[\mathcal{U}(t)\rho]\\
    &= \E \exp_\mathcal{T}\left( \int_0^t \mathcal{L}(\tau)\,d\tau \right)(\rho) \\
    &= \sum_{n=0}^\infty \int_{0 \leq \tau_n \leq \dots \leq \tau_1 \leq t} d\tau_1 \dots d\tau_n\,\E[\mathcal{L}(\tau_1)\dots \mathcal{L}(\tau_n)](\rho)\\
    &= \sum_{n=0}^\infty \int_{0 \leq \tau_n \leq \dots \leq \tau_1 \leq t} d\tau_1 \dots d\tau_n \left[ \int_{\mathbb{R}^{m \times n}} p(\bm{x_1},\tau_1; \dots ; \bm{x_n},\tau_n) \mathcal{L}(\bm{x_1})\dots \mathcal{L}(\bm{x_n}) \;d\bm{x_1} \dots d\bm{x_n} \right](\rho).
\end{aligned}
\end{equation}
In the last equality we have expanded
\begin{equation}
\begin{aligned}
    \E[\mathcal{L}(\tau_1)\dots \mathcal{L}(\tau_n)] = \int_{\mathbb{R}^{m \times n}} p(\bm{x_1},\tau_1; \dots ; \bm{x_n},\tau_n) \mathcal{L}(\bm{x_1})\dots \mathcal{L}(\bm{x_n}) \,d\bm{x_1} \dots d\bm{x_n}.
\end{aligned}
\end{equation}
This somewhat lengthy expression is simply a manifestation of the equation $\E \!f(X) = \int p(x)f(x)\,dx$ we use all the time. Since stochastic processes governed by the given SDEs are Markovian (\cref{fact:properties_SDE_processes}), the joint pdf 
\begin{equation}
\begin{aligned}
    p(\bm{x_1},\tau_1; \dots ; \bm{x_n},\tau_n) = \prod_{i=1}^{n-1} p(\bm{x_i},\tau_i|\bm{x_{i+1}},\tau_{i+1})\cdot p(\bm{x_n},\tau_n)
\end{aligned}
\end{equation}
can be obtained (in theory at least) from the Fokker-Planck equation for $\bm{\delta}$, see \cref{equation:Fokker-Planck}. Note that the chronological ordering here is $\tau_n \leq \dots \leq \tau_1$ instead of the usual $\tau_1 \leq \dots \leq \tau_n$ in Markov process theory; this is a consequence of the convention used in time-ordered unitaries. 

For commuting Hamiltonians this expression simplifies considerably since time-ordering no longer matters. We have
\begin{equation}
\begin{aligned}
    \mathcal{S}_t(\rho) &= \E \exp\left( \int_0^t \mathcal{L}(\tau)\,d\tau \right)(\rho)\\
    &= \E \exp\left(-i\int_0^t H(\tau)\,d\tau \right) \rho \exp\left(i\int_0^t H(\tau)\,d\tau \right)\\
    &= \E \exp\left(-i \sum_{i=1}^m \left(\int_0^t \delta_i(\tau)\,d\tau\right) H_i \right) \rho \exp\left(i \sum_{i=1}^m \left(\int_0^t \delta_i(\tau)\,d\tau\right) H_i \right)\\
    &= \int_{\mathbb{R}^m} p_{\int}(\bm{x},t) \exp\left(-i \sum_{i=1}^m x_iH_i \right) \rho \exp\left(i \sum_{i=1}^m x_iH_i \right) d\bm{x}.
\end{aligned}
\end{equation}
Here the $\bm{x}$'s are realizations of the random variable $\int_0^t \bm{\delta}(\tau)\,d\tau$, and $p_{\int}(\bm{x},t)$ is the pdf of $\int_0^t \bm{\delta}(\tau)\,d\tau$.
\end{proof}

We close off this section with a simple example.
\begin{example}\label{example:channel_description_pdf_single_qubit_SEH/SEU}
    Consider the single-qubit SEH $H(t)= \delta(t)\sigma_z$, where
    \begin{equation}\begin{aligned}
        d\delta(t) &= G\,dW(t)\\
        \delta(0) &= \delta_0 \in \mathbb{R}.
    \end{aligned}\end{equation}
    That is, $\delta(t)$ is a pure Brownian motion, with a scaled diffusion factor $G$, and no drift. In this case, \cref{proposition:channel_description_pdf_SEH/SEU} says the evolution of a single qubit state can be described as
    \begin{equation}
    \begin{aligned}
        \rho(t) = \mathcal{S}_t(\rho(0)) = \int_{-\infty}^{\infty} p(x,t)e^{-ix\sigma_z}\rho(0)e^{ix\sigma_z}\,dx.
    \end{aligned}
    \end{equation}
    Here,
\begin{equation}\begin{aligned}
        p(x,t) = \frac{1}{\sqrt{2\pi\sigma^2(t)}}e^{-\frac{(x-\mu(t))^2}{2\sigma(t)^2}}
    \end{aligned}\end{equation}
    is the pdf of $\int_0^t \delta(\tau)\, d\tau$ where
    \begin{equation}\begin{aligned}
        \mu(t) = \delta_0t, \qquad \sigma^2(t) = \frac{G^2t^3}{3}.
    \end{aligned}\end{equation}
    In this almost trivial example we do not need the heavy machinery of stochastic calculus: since $\delta(t)=\delta_0 + GW(t)$, we simply integrate to obtain
    \begin{equation}\begin{aligned}
        \int_0^t \delta(\tau)\, d\tau = \delta_0t + G \int_0^t W(\tau)\, d\tau \sim \mathcal{N}(\delta_0t,\frac{G^2 t^3}{3}).
    \end{aligned}\end{equation}
    We can further evaluate the integrand $e^{-ix\sigma_z}\rho(0)e^{ix\sigma_z}$:
    \begin{equation}
    \begin{aligned}
        e^{-ix\sigma_z}\rho(0)e^{ix\sigma_z} = 
        \begin{pmatrix}
            \rho(0)_{00} & e^{-2ix}\rho(0)_{01}\\
            e^{2ix}\rho(0)_{10} & \rho(0)_{11}
        \end{pmatrix}.
    \end{aligned}
    \end{equation}
    Therefore we see that while the diagonal elements of $\rho(t)$ remain unchanged, the off-diagonal elements are 
    \begin{equation}
    \begin{aligned}
        \rho(t)_{01} &= \int p(x,t)e^{-2ix}\,dx \cdot \rho(t)_{01}\\
        &= e^{-i2\delta_0t}e^{-2G^2t^3/3}\rho(t)_{01}.
    \end{aligned}
    \end{equation}
It is well-known that $\sigma_z$ noise induces dephasing in the computational basis. Here, the specific dephasing factor $e^{-2G^2t^3/3}$ arises from the random walk behaviour of the noise.
    
\end{example}

\begin{remark} Above, we have defined SEHs only for qubit Hamiltonians, and with respect to SDEs driven by Wiener processes. There are many ways to generalize this notion. For instance, one could very well define fermionic/bosonic SEHs, or use a different driving process in the SDE.

Furthermore, we note that the notion of Hamiltonians with stochastic time-dependence is not new. However, existing works incorporate this aspect in different ways:
\begin{enumerate}
    \item In \cite{yu2025average}, the authors incorporated stochastic time-dependence by writing $H(t) = H(x(t))$, where $x(t)$ is a `classical stochastic quantity'. They also considered special cases in which $x(t)$ are Gaussian processes, in particular a single-qubit Hamiltonian similar to ours in \cref{example:channel_description_pdf_single_qubit_SEH/SEU}. \cite{yu2025average} also developed a master-equation framework for $H(x(t))$. In our case, we explicitly incorporate stochastic time-dependence through the local interaction strengths of local Hamiltonians, and explicitly describe this stochasticity by means of SDEs. The $e^{-2G^2t^3/3}$ dephasing factor in \cref{example:channel_description_pdf_single_qubit_SEH/SEU} is a consequence of this SDE description, and is absent from \cite{yu2025average}. 

    \item In \cite{lashkari2013towards}, the authors incorporated stochastic time-dependence through the notion of a `Brownian quantum circuit'. This is a sequence of unitaries 
    \begin{equation}
    \begin{aligned}
        \exp(-iH_r\Delta t)\exp(-iH_{r-1}\Delta t)\dots \exp(-iH_1\Delta t),
    \end{aligned}
    \end{equation}
    where each $H_i$ is a \textit{time-independent} two-local Hamiltonian given by
    \begin{equation}
    \begin{aligned}
        H_i = \sum_{j<k} \sum_{\alpha_j,\alpha_k=0}^3 \sigma_j^{\alpha_j} \otimes \sigma_k^{\alpha_k} \Delta B_{i,j,k,\alpha_j,\alpha_k}. 
    \end{aligned}
    \end{equation}
    Here the $\Delta B_{i,j,k,\alpha_j,\alpha_k}$'s are i.i.d. drawn from $\mathcal{N}(0,\epsilon^2)$, with $\epsilon^2 \propto 1/\Delta t$. This immediately brings to mind the quadratic variation of the Wiener process. Thus, while this model shares characteristics similar to our SEHs, they are distinct.

    \item In \cite{onorati2017mixing}, the authors incorporated stochastic time-dependence in a similar way to \cite{lashkari2013towards}. They also went further and described their unitary $U(t)$ as a Brownian motion/Wiener process \textit{on the unitary group itself}. While the authors stopped short of using any further SDE formalism, it can be deduced that the generators of $U(t)$ bear a very similar form to our SEHs in \cref{definition:SEH/SEU}, in the \textit{special case} where $\bm{F}=\bm{0}$ and $\bm{G}$ is constant and diagonal, i.e. there are no drift terms and no stochastic correlations among the local interaction coefficients. 

    \item In \cite{de2026noisy}, the stochastic time-dependence $Z(t)$ in their Hamiltonian takes the form of \textit{white noise}, i.e. $Z(t)=\frac{dW(t)}{dt}$. Thus, their Hamiltonians are `on the level' of $\frac{dW(t)}{dt}$. In contrast, the $\delta(t)$'s in our SEHs evolve stochastically due to $\int G\, dW(t) \sim \int G \frac{dW(t)}{dt}\,dt \sim GW(t)$, i.e. the integral over $\frac{dW(t)}{dt}$. Thus, the SEHs themselves are `on the level' of $W(t)$. Here we have abused notation and used $\sim$ to simply mean `on the level of' -- it is not to be interpreted in any strict mathematical sense. 
    
    The authors then developed SDEs for the \textit{quantum states} $\psi_t$, i.e. equations of the form $d\psi_t = (\dots)\,dt + (\dots)\,dW_t$, and further established connections to the stochastic Schr\"{o}dinger equation and stochastic Liouville equation, making heavy use of SDE machinery along the way. To conclude, the stochastic time-dependence in \cite{de2026noisy} is distinct from ours, although we do recommend it for a complementary perspective.
    
\end{enumerate}
\end{remark}

\section{Single-qubit Quantum Sensing under Stochastic Evolution}\label{section:quantum_sensing}

The first quantum information processing task we consider is the parameter estimation of a single-qubit SEH. The setting we consider is that in \cref{example:channel_description_pdf_single_qubit_SEH/SEU}. Concretely, we consider
\begin{problem}\label{problem:single_qubit_quantum_sensing}
    Given the single-qubit Hamiltonian $H(t)= \delta(t)\sigma_z$, where
    \begin{equation}\begin{aligned}
        d\delta(t) &= G\,dW(t)\\
        \delta(0) &= \delta_0.
    \end{aligned}\end{equation}
    Here $\delta_0$ and $G$ are unknown. It is known however that $|\delta_0| \leq \delta_{\text{max}}$ and $0 \leq G \leq G_{\text{max}}$.
    
    \textbf{Goal.} Learn $\delta_0$.
\end{problem}
Before proceeding we first review the well-understood case when $G=0$. We discuss two estimation protocols, namely the classic Ramsey protocol \cite{ramsey1950molecular,degen2017quantum} and the more recent Robust Frequency Estimation (RFE) protocol developed in \cite{kimmel2015robust,ma2024learning}, which achieves the Heisenberg limit. Then we discuss the extension of these protocols in the setting of \cref{problem:single_qubit_quantum_sensing} and their respective performances and limits. The extensions/modifications are highlighted in \textcolor{DeepPink3}{pink}.

We note that there are broadly two approaches to Heisenberg-limited estimation for the simple case of single-qubit quantum sensing. The first approach utilises massive entanglement (namely, in the NOON/GHZ states) as a resource \cite{lee2002quantum,leibfried2004toward}, and the second approach is based on long-time evolution \cite{kimmel2015robust,de2005quantum,higgins2007entanglement}. Following \cite{huangtong2023learning}, we pursue the (extension of the) second approach.

\subsection{Protocols for noiseless quantum sensing}
When $G=0$ we have the simple procedure \cref{algorithm:ramsey_protocol}, taught in introductory quantum mechanics classes. We call it the Ramsey protocol in this paper.

\begin{algorithm}
\caption{Ramsey Protocol}
\label{algorithm:ramsey_protocol}
\begin{algorithmic}[1]

\Require 
\Statex - Black-box access to the unitary $U(t)= e^{-iHt}$, where $H=\delta_0 \sigma_z$. $\delta_0$ is unknown, but satisfies $|\delta_0| \leq \delta_{\text{max}}$.
\Statex - Desired precision $\varepsilon>0$ for the estimate $\hat{\delta}_0$.
\Statex - Success probability $1-\eta$, $\eta>0$.

\Ensure 
\Statex - Estimate $\hat{\delta}_0$ satisfying $|\hat{\delta}_0-\delta_0| < \varepsilon$, with success probability $1-\eta$.

\Statex
\Pseudocode
\State Fix a time $t\leq \pi/4\delta_{\text{max}}$.
\For{$i=1,\dots,N=O\left(\frac{\log(1/\eta)}{\varepsilon^2}\right)$}
    \State Initialize the qubit state $\ket{+} = \frac{1}{\sqrt{2}}(\ket{0}+\ket{1})$.
    \State Evolve $\ket{+}$ under $U(t)$.
    \State Measure the observable $\sigma_y$ and record the measurement outcome $Y_i \in \{\pm1\}$.
\EndFor
\State Compute $\hat{Y} \coloneq \frac{1}{N}\sum_{i=1}^N Y_i$.
\State Compute $\hat{\delta}_0 \coloneq \frac{\sin^{-1} \hat{Y}}{2t}$.
\State \Return $\hat{\delta}_0$

\end{algorithmic}
\end{algorithm}

\paragraph{Analysis of \cref{algorithm:ramsey_protocol}.} 
For a single run, evolving under $U(t)$ gives the state $\ket{\psi_t} = e^{-i\delta_0 \sigma_z t}\ket{+} = \frac{1}{\sqrt{2}}(e^{-i\delta_0 t}\ket{0}+e^{i\delta_0 t}\ket{1})$. Measurements of $\sigma_y$ yield outcomes $Y_i = \pm 1$, corresponding to the eigenstates $\ket{y,\pm1}$ (check out the notation in \cref{appendix:quantum_info}). Note that $\Pr(Y_i = \pm 1) = \frac{1}{2}(1\pm \sin(2\delta_0 t))$, so $\E \!Y_i = \sin(2\delta_0 t)$. By Hoeffding's inequality (\cref{theorem:Hoeffding}), to ensure $|\hat{Y}-\E \!Y| < \varepsilon$ with success probability $1-\eta$, we need 
\begin{equation}
\begin{aligned}
    N=O\left( \frac{\log(1/\eta)}{\varepsilon^2} \right)
\end{aligned}
\end{equation}
runs. Finally, the Lipschitz continuity of $\sin^{-1} x$ (away from the endpoints $x=\pm 1$) ensures that the final error for $\delta_0$ is $|\hat{\delta}_0-\delta_0| < O(\varepsilon)$.

We need the sample runtime $t\leq \pi/4\delta_{\text{max}}$ because the trigonometric function sine is not uniquely invertible. We choose $t$ such that $2\delta_0t \leq \pi/2$: since sine is one-to-one on $[-\pi/2,\pi/2]$, we then have a unique candidate for the estimate $\hat{\delta}_0$. Since we need $N=O(1/\varepsilon^2)$ runs with $t=O(1)$ per run, the total runtime scales as $T=Nt=O(1/\varepsilon^2)$, attaining the standard quantum limit. 

Lastly, note that we could have measured the observable $\sigma_x$ instead, with the ensuing analysis being completely similar, but for cosine instead of sine. \hfill $\square$

Next, we introduce the Robust Frequency Estimation (RFE) protocol in \cite{ma2024learning,hu2025ansatz}, which has its origins in \cite{kimmel2015robust}. This protocol achieves Heisenberg scaling in the estimation of $\delta_0$. Notably, unlike previous approaches which utilises massive entanglement as an additional resource, RFE works with just a single qubit, at the expense of requiring long sample runtime.

\begin{algorithm}
\caption{Robust Frequency Estimation (RFE)}
\label{algorithm:rfe}
\begin{algorithmic}[1]

\Require 
\Statex - Black-box access to the unitary $U(t)= e^{-iHt}$, where $H=\delta_0 \sigma_z$. $\delta_0$ is unknown, but satisfies $|\delta_0| \leq \delta_{\text{max}}$.
\Statex - Desired precision $\varepsilon>0$ for the estimate $\hat{\delta}_0$.
\Statex - Success probability $1-\eta$, $\eta>0$.

\Ensure 
\Statex - Estimate $\hat{\delta}_0$ satisfying $|\hat{\delta}_0-\delta_0| < \varepsilon$, with success probability $1-\eta$.

\Statex
\Pseudocode
\State Initialise $a=-\delta_{\text{max}}, b=\delta_{\text{max}}$.
\For{$i=1,\dots,N=O\left(\log \frac{\delta_{\text{max}}}{\varepsilon}\right)$}
    \State Set $t_i = O\left((\frac{1}{\delta_{\text{max}}})(\frac{3}{2})^{i-1}\right)$.
    \For{$j=1,\dots,m=O\left(\log \frac{N}{\eta}\right)$} \Comment{Ramsey subroutine}
        \State Initialize the qubit state $\ket{+} = \frac{1}{\sqrt{2}}(\ket{0}+\ket{1})$.
        \State Evolve $\ket{+}$ under $U(t_i)$.
        \State Measure the observable $\sigma_y$ and record the measurement outcome $Y_j(t_i) \in \{\pm1\}$.
    \EndFor
    \For{$j=1,\dots,m=O\left(\log \frac{N}{\eta}\right)$} \Comment{Ramsey subroutine}
        \State Initialize the qubit state $\ket{+} = \frac{1}{\sqrt{2}}(\ket{0}+\ket{1})$.
        \State Evolve $\ket{+}$ under $U(t_i)$.
        \State Measure the observable $\sigma_x$ and record the measurement outcome $X_j(t_i) \in \{\pm1\}$.
    \EndFor
    \State Compute $\hat{X}(t_i) = \frac{1}{m}\sum_{j=1}^m X_j(t_i)$, $\hat{Y}(t_i) = \frac{1}{m}\sum_{j=1}^m Y_j(t_i)$. Let $S(t_i)\coloneq \hat{X}(t_i) + i\hat{Y}(t_i)$.
    \State Compute $\im\left( \exp\left(-i\frac{(a+b)\pi}{2(b-a)}\right)S(t_i) \right)$.
    \If{$\im \leq 0$} \Comment{Invoke \cref{lemma:lemma8_hu2025}}
        \State $(a,b) \gets (a,\frac{a+2b}{3})$
    \ElsIf{$\im > 0$}
        \State $(a,b) \gets (\frac{2a+b}{3},b)$
    \EndIf
\EndFor
\State \textbf{print} $\hat{\delta}_0 \in [a,b]$

\end{algorithmic}
\end{algorithm}

\paragraph{Analysis of \cref{algorithm:rfe} (Theorem 7 in \cite{hu2025ansatz}).} 
A complete analysis of \cref{algorithm:rfe} can be found in \cite{hu2025ansatz}; for convenience of the reader we include a reasonably detailed sketch here.

The essence of RFE is as follows: instead of naively relying on Hoeffding to achieve a precise estimate for $\delta_0$ via independent experimental runs, of which $N=\Omega(1/\varepsilon^2)$ is required for precision $\varepsilon$, RFE iteratively whittles down the size of the interval containing $\delta_0$ by a constant factor for each measure-and-compute run. The key enabler of this is the computation of the quantity $\im(\dots)$ and \cref{lemma:lemma8_hu2025}, which tells us how to shave off the interval based on the value of $\im(\dots)$. Since the interval size decreases by a constant factor, only a logarithmic number of experimental runs is required, at the expense of an exponentially growing sample runtime. Nonetheless, the total runtime attains the Heisenberg limit $T=O(1/\varepsilon)$.

To invoke the bisection \cref{lemma:lemma8_hu2025}, some preprocessing is required to prepare the $Z(\frac{\pi}{b-a})$ in \cref{lemma:lemma8_hu2025}. For the $i$th run we first perform the core subroutine in the Ramsey protocol \cref{algorithm:ramsey_protocol}, which gives us the raw, $\pm1$-valued data $X_j(t_i), Y_j(t_i)$. From these measurement data, we construct the empirical means $\hat{X}(t_i), \hat{Y}(t_i)$, and put them together as $S(t_i)$. This quantity $S(t_i)$ is to take the role of $Z(\frac{\pi}{b-a})$ in \cref{lemma:lemma8_hu2025}, for the $i$th iteration. 

We note that this is actually a simplified version of the actual RFE protocol in \cite{hu2025ansatz}. There, a median-of-means (MoM) approach was used: they first partitioned the $\{X_j(t_i)\}_{j=1}^m$, $\{Y_j(t_i)\}_{j=1}^m$ into $k=\Theta(1)$-sized blocks. For each block they computed $\hat{X}_l(t_i) = \frac{1}{k}\sum_{j \in \text{block }l} X_j(t_i)$ and $\hat{Y}_l(t_i) = \frac{1}{k}\sum_{j \in \text{block }l} Y_j(t_i)$, then defined $X_{\text{MoM}}(t_i) \coloneq \text{Median}\{\hat{X}_l(t_i)\}_{l=1}^{m/k}, Y_{\text{MoM}}(t_i) \coloneq \text{Median}\{\hat{Y}_l(t_i)\}_{l=1}^{m/k}$. Their $S(t_i)$ was then taken to be $S(t_i)\coloneq X_{\text{MoM}}(t_i) + iY_{\text{MoM}}(t_i)$. The reason \cite{hu2025ansatz} used MoM was because they wanted to make their Hamiltonian learning protocol robust to $O(1)$-sized SPAM errors -- the `robust' in RFE refers to this aspect. For simplicity we shall not address this concern here, so while the original MoM subroutine indeed works in this protocol, it is somewhat overkill for our purposes. We simply construct the empirical mean from the full $m$ samples and apply Hoeffding directly.

The bisection \cref{lemma:lemma8_hu2025} then tells us whether to retain the first two-thirds of the original interval $[a,b]$ or the last two-thirds, depending on the sign of $\im(\dots)$. Either way, the size of the interval containing $\delta_0$ becomes $b-a \to \frac{2}{3}(b-a)$. Thus after $N$ runs we have $b-a \to (\frac{2}{3})^N(b-a) \approx \varepsilon$, therefore we need $N=O(\log_{3/2} (1/\varepsilon))$ runs. Note that for the $i$th run \cref{lemma:lemma8_hu2025} necessitates that $t_i \sim \frac{1}{b_i-a_i} = (\frac{3}{2})^i \frac{1}{b-a}$, this is why the sample runtime increases exponentially.

Of course, each run involves a small failure probability which has to be accounted for. If the overall failure probability of RFE is set to $\eta$, we have from Hoeffding and the union bound that $m=O(\log \frac{N}{\eta})$, where recall $m$ is the number of runs in the preprocessing part (i.e. the Ramsey subroutine). All in all, we have
\begin{equation}\begin{aligned}
    \text{max sample runtime, } &\max_i t_i = O(1/\varepsilon)\\
    \text{number of samples, } &K = \sum_{i=1}^N m =  O\left(\left(\log \frac{\delta_{\text{max}}}{\varepsilon}\right)\left(\log \frac{1}{\eta}+\log \log \frac{\delta_{\text{max}}}{\varepsilon}\right)\right)\\
    \text{total runtime, } &T = \sum_{i=1}^N mt_i = O\left(\frac{1}{\varepsilon}\left(\log \frac{1}{\eta}+\log \log \frac{\delta_{\text{max}}}{\varepsilon}\right)\right).
\end{aligned}\end{equation} \hfill $\square$

\begin{lemma}[Lemma 8 in \cite{hu2025ansatz}]\label{lemma:lemma8_hu2025}
    Let $\theta \in [a,b]$. Let $Z(t)$ be a random variable such that
    \begin{equation}\begin{aligned}
        |Z(t)-e^{i\theta t}| < 1/2
    \end{aligned}\end{equation}
    for all $t$. Then with a single sample of $Z(\frac{\pi}{b-a})$, we can distinguish between the two (overlapping) cases:
    \begin{enumerate}[i.]
        \item $\im\left( \exp\left(-i\frac{(a+b)\pi}{2(b-a)}\right)Z(\frac{\pi}{b-a}) \right) \leq 0 \implies \theta \in [a,\frac{a+2b}{3}]$.
        \item $\im\left( \exp\left(-i\frac{(a+b)\pi}{2(b-a)}\right)Z(\frac{\pi}{b-a}) \right) > 0 \implies \theta \in [\frac{2a+b}{3},b]$.
    \end{enumerate}
\end{lemma}

\subsection{Extension of protocols to accommodate SDE-governed parameters}
When $G>0$ we have the following \cref{algorithm:extended_ramsey_protocol}.

\begin{algorithm}
\caption{Extended Ramsey Protocol}
\label{algorithm:extended_ramsey_protocol}
\begin{algorithmic}[1]

\Require 
\Statex - Black-box access to the channel $\mathcal{S}_t$ (see \cref{equation:channel_description_SEH/SEU}), where $H(t)$ is given in \cref{problem:single_qubit_quantum_sensing}.
\Statex - Desired precision $\varepsilon>0$ for the estimate $\hat{\delta}_0$.
\Statex - Success probability $1-\eta$, $\eta>0$.

\Ensure 
\Statex - Estimate $\hat{\delta}_0$ satisfying $|\hat{\delta}_0-\delta_0| < \varepsilon$, with success probability $1-\eta$.

\Statex
\Pseudocode
\State Fix a time $t\leq \pi/4\delta_{\text{max}}$.
\For{$i=1,\dots,N=O\left(\frac{\log(1/\eta)}{\varepsilon^2}\right)$}
    \State Initialize the qubit state $\ket{+} = \frac{1}{\sqrt{2}}(\ket{0}+\ket{1})$.
    \State Evolve $\ket{+}$ under $\mathcal{S}_t$.
    \State Measure the observable $\sigma_y$ and record the measurement outcome $Y_i \in \{\pm1\}$.
\EndFor
\textcolor{DeepPink3}{
\For{$i=1,\dots,N=O\left(\frac{\log(1/\eta)}{\varepsilon^2}\right)$}
    \State Initialize the qubit state $\ket{+} = \frac{1}{\sqrt{2}}(\ket{0}+\ket{1})$.
    \State Evolve $\ket{+}$ under $\mathcal{S}_t$.
    \State Measure the observable $\sigma_x$ and record the measurement outcome $X_i \in \{\pm1\}$.
\EndFor
}
\State Compute $\hat{Y} \coloneq \frac{1}{N}\sum_{i=1}^N Y_i$, \textcolor{DeepPink3}{$\hat{X} \coloneq \frac{1}{N}\sum_{i=1}^N X_i$}.
\State \textcolor{DeepPink3}{Compute $\hat{\delta}_0 \coloneq \frac{\tan^{-1} (\hat{Y}/\hat{X})}{2t}$}.
\State \Return $\hat{\delta}_0$

\end{algorithmic}
\end{algorithm}

\paragraph{Analysis of \cref{algorithm:extended_ramsey_protocol}.} 
With the randomness introduced by SDEs, the evolution of quantum states is now described by quantum channels instead of mere unitaries. For the simple SDE setting in \cref{problem:single_qubit_quantum_sensing}, the channel $\mathcal{S}_t \coloneq \E[U(t)\rho U(t)^\dag]$ has an especially tractable description, which we gave in \cref{example:channel_description_pdf_single_qubit_SEH/SEU}. Namely,
\begin{equation}
\begin{aligned}
    \mathcal{S}_t(\rho(0)) = \int_{-\infty}^{\infty} p(x,t)e^{-ix\sigma_z}\rho(0)e^{ix\sigma_z}\,dx
\end{aligned}
\end{equation}
with
\begin{equation}\begin{aligned}
    p(x,t) = \frac{1}{\sqrt{2\pi\sigma^2(t)}}e^{-\frac{(x-\mu(t))^2}{2\sigma(t)^2}},
\end{aligned}\end{equation}
where
\begin{equation}\begin{aligned}
    \mu(t) = \delta_0t, \qquad \sigma^2(t) = \frac{G^2t^3}{3}.
\end{aligned}\end{equation}
Initializing with $\rho(0) = \ketbra{+}{+}$, evolving under $\mathcal{S}_t$, then measuring $\sigma_y$ as before yields 
\begin{equation}\label{equation:evaluating_EY}
\begin{aligned}
    \E\hat{Y} =\E \!Y_i &= \tr(\rho(t)(\ketbra{y,1}{y,1}-\ketbra{y,-1}{y,-1}))\\
    &= \int_{-\infty}^{\infty} p(x)|\bra{y,1}e^{-i\sigma_zx}\ket{+}|^2\,dx - \int_{-\infty}^{\infty} p(x)|\bra{y,-1}e^{-i\sigma_zx}\ket{+}|^2\,dx\\
    &= \int_{-\infty}^{\infty} p(x) \cdot \frac{1}{2}(1+\sin 2x)\,dx - \int_{-\infty}^{\infty} p(x) \cdot \frac{1}{2}(1-\sin 2x)\,dx\\
    &= \int_{-\infty}^{\infty} p(x) \sin 2x\,dx\\
    &= e^{-2\sigma^2(t)}\sin 2\mu(t)\\
    &= e^{-2G^2t^3/3}\sin 2\delta_0t.
\end{aligned}
\end{equation}
Here the integral was evaluated using the well-known result for computing characteristic functions of Gaussians: $\E[e^{ikX}] = e^{ik\mu-k^2\sigma^2/2}$. At this point, if we try to proceed with \cref{algorithm:ramsey_protocol}, i.e. attempt to extract $\hat{\delta}_0$ from $\hat{Y}$, this would fail due to the damping factor $e^{-2G^2t^3/3}$. The idea is to observe that by also measuring with respect to the observable $\sigma_x$, we obtain the estimate $\hat{X}$, where $\E \hat{X} = e^{-2G^2t^3/3}\cos 2\delta_0t$ by a similar calculation as in \cref{equation:evaluating_EY}. Taking the quotient $\hat{Y}/\hat{X}$ then cancels off the noise factor, thus allowing the extraction of $\hat{\delta}_0$.

The error analysis here is slightly more involved. We essentially have the following task: let $Z=(X,Y)=(C\cos \theta, C\sin \theta)$ and its estimator be $\hat{Z} = (\hat{X},\hat{Y}) = (\hat{C}\cos \hat{\theta}, \hat{C}\sin \hat{\theta})$. It is clear that if $|\hat{X}-X|<\varepsilon'$ and $|\hat{Y}-Y|<\varepsilon'$, then $|\hat{Z}-Z|< \sqrt{2}\varepsilon'$. Furthermore, by the geometric \cref{lemma:geometric_lemma} below, if $|\hat{Z}-Z|< \sqrt{2}\varepsilon'$, then $|\hat{\theta}-\theta| < \sin^{-1}(\sqrt{2}\varepsilon'/C)$. Reverse-engineering from the target $\hat{\theta}$-precision $\varepsilon$ gives $\varepsilon'= \frac{C}{\sqrt{2}}\sin \varepsilon$. Thus, running the experiments producing $\hat{Y}, \hat{X}$ for $N=O(1/\varepsilon'^2) = O(1/C^2\sin^2 \varepsilon)$ times each, gives an $\varepsilon$-precise estimate of $\hat{\theta}$ with high probability.

In our case, $C=e^{-2G^2t^3/3}$ and $\theta=2\delta_0t$, so 
\begin{equation}
\begin{aligned}
    N = O\left( \frac{e^{4G^2t^3/3}}{\sin^2 2t\varepsilon} \right) = O\left( \frac{1}{\varepsilon^2} \right)
\end{aligned}
\end{equation}
if $t=O(1)$ and $t\varepsilon \ll 1$, still attaining the standard quantum limit, although we do note the considerable dependence on $G$.

Finally, we should mention that the $\tan^{-1} (y/x)$ function used to extract $\hat{\delta}_0$ is technically the $\text{atan2}(y,x)$ function. Strictly speaking, $\tan^{-1} (y/x)$ is insufficient because it does not distinguish between $(x,y)$ and $(-x,-y)$. The \href{https://en.wikipedia.org/wiki/Atan2}{atan2} function resolves this unambiguously, simply giving the polar coordinate $\theta \in (-\pi,\pi]$ of the vector $(x,y)$. \hfill $\square$

\begin{lemma}\label{lemma:geometric_lemma}
    Let $z=(r \cos \theta, r \sin \theta)$ and $z'=(r' \cos \theta', r' \sin \theta')$ be such that $|z'-z|<\varepsilon$, where $0<\varepsilon<r$. Then $|\theta'-\theta|<\sin^{-1}(\varepsilon/r)$.
\end{lemma}
\begin{proof}
    Proof by diagram: see \cref{figure:visualization_geometric_lemma}. Wlog we set $z$ along the $x$-axis, so $\theta=0$. Any candidate $z'$ must lie within the $\varepsilon$-sized circle. Clearly the largest $\theta'$ is attained when the corresponding $z'$ is tangent to the circle. For this $\theta'$ it is easy to see that $\sin \theta' = \varepsilon/r$.
\begin{figure}[ht]
\centering

\begin{tikzpicture}[scale=1.2, >=stealth]
\def\r{4}         
\def\eps{0.7}     
\pgfmathsetmacro{\ang}{asin(\eps/\r)}           
\pgfmathsetmacro{\tlen}{sqrt(\r*\r-\eps*\eps)}  
\coordinate (O) at (0,0);
\coordinate (Z) at (\r,0);
\coordinate (T) at ({\tlen*cos(\ang)},{\tlen*sin(\ang)});
\draw[->] (-0.5,0) -- (5.3,0) node[right] {$x$};
\draw[->] (0,-0.5) -- (0,2.2) node[above] {$y$};
\draw[->, thick] (O) -- (Z);
\node[above] at ($(O)!0.72!(Z)$) {$z$};
\node[below] at ($(O)!0.5!(Z)$) {$r$};
\draw (Z) circle (\eps);
\fill (Z) circle (1.2pt);
\draw[->, thick] (O) -- (T);
\node[above left] at ($(O)!0.68!(T)$) {$z'$};
\draw[dashed] (Z) -- (T);
\node[right] at ($(Z)!0.5!(T)$) {$\varepsilon$};
\draw (1.5,0) arc[start angle=0,end angle=\ang,radius=1.5];
\node at ({1.95*cos(\ang/2)},{1.95*sin(\ang/2)}) {$\theta'$};
\end{tikzpicture}

\caption{A diagrammatic proof of the geometric \cref{lemma:geometric_lemma}.}
\label{figure:visualization_geometric_lemma}
\end{figure}
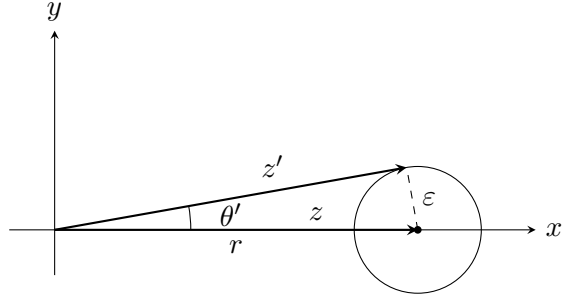
\end{proof}

\begin{algorithm}
\caption{Extended RFE}
\label{algorithm:extended_rfe}
\begin{algorithmic}[1]

\Require 
\Statex - Black-box access to the channel $\mathcal{S}_t$ (see \cref{equation:channel_description_SEH/SEU}), where $H(t)$ is given in \cref{problem:single_qubit_quantum_sensing}.
\Statex - Desired precision $\varepsilon>0$ for the estimate $\hat{\delta}_0$.
\Statex - Success probability $1-\eta$, $\eta>0$.

\Ensure 
\Statex - Estimate $\hat{\delta}_0$ satisfying $|\hat{\delta}_0-\delta_0| < \varepsilon$, with success probability $1-\eta$.

\Statex
\Pseudocode
\State Initialise $a=-\delta_{\text{max}}, b=\delta_{\text{max}}$.
\State \textcolor{DeepPink3}{Define $C(x)\coloneq e^{-2G_{\text{max}}^2x^3/3}$.}
\For{$i=1,\dots,N=O\left(\log \frac{\delta_{\text{max}}}{\varepsilon}\right)$}
    \State Set $t_i = O\left((\frac{1}{\delta_{\text{max}}})(\frac{3}{2})^{i-1}\right)$.
    \For{$j=1,\dots,\textcolor{DeepPink3}{m_i=O\left(\frac{1}{C(t_i)^2}\log \frac{N}{\eta}\right)}$} \Comment{Ramsey subroutine}
        \State Initialize the qubit state $\ket{+} = \frac{1}{\sqrt{2}}(\ket{0}+\ket{1})$.
        \State Evolve $\ket{+}$ under $\mathcal{S}_{t_i}$.
        \State Measure the observable $\sigma_y$ and record the measurement outcome $Y_j(t_i) \in \{\pm1\}$.
    \EndFor
    \For{$j=1,\dots,\textcolor{DeepPink3}{m_i=O\left(\frac{1}{C(t_i)^2}\log \frac{N}{\eta}\right)}$} \Comment{Ramsey subroutine}
        \State Initialize the qubit state $\ket{+} = \frac{1}{\sqrt{2}}(\ket{0}+\ket{1})$.
        \State Evolve $\ket{+}$ under $\mathcal{S}_{t_i}$.
        \State Measure the observable $\sigma_x$ and record the measurement outcome $X_j(t_i) \in \{\pm1\}$.
    \EndFor
    \State Compute $\hat{X}(t_i) = \frac{1}{\textcolor{DeepPink3}{m_i}}\sum_{j=1}^{\textcolor{DeepPink3}{m_i}} X_j(t_i)$, $\hat{Y}(t_i) = \frac{1}{\textcolor{DeepPink3}{m_i}}\sum_{j=1}^{\textcolor{DeepPink3}{m_i}} Y_j(t_i)$. Let $S(t_i)\coloneq \hat{X}(t_i) + i\hat{Y}(t_i)$.
    \State Compute $\im\left( \exp\left(-i\frac{(a+b)\pi}{2(b-a)}\right)S(t_i) \right)$.
    \If{$\im \leq 0$} \Comment{Invoke \cref{lemma:lemma8_hu2025}}
        \State $(a,b) \gets (a,\frac{a+2b}{3})$
    \ElsIf{$\im > 0$}
        \State $(a,b) \gets (\frac{2a+b}{3},b)$
    \EndIf
\EndFor
\State \textbf{print} $\hat{\delta}_0 \in [a,b]$

\end{algorithmic}
\end{algorithm}

\paragraph{Analysis of \cref{algorithm:extended_rfe}.}
The extension of RFE to the $G>0$ case is simple if we know $G_{\text{max}}$, the upper bound of the noise strength, a priori. In this case, we only need to tweak slightly the preprocessing protocol required to apply \cref{lemma:lemma8_hu2025}. The key is the (trivial) observation that
\begin{equation}
\begin{aligned}
    |Z(t)-Ce^{i\theta t}| < C/2 \iff |Z(t)/C-e^{i\theta t}| < 1/2
\end{aligned}
\end{equation}
for any $C>0$. Thus define $C(x)\coloneq e^{-2G_{\text{max}}^2x^3/3}$. Since $\E \!X_j(t_i)=e^{-2G^2t_i^3/3} \cos \theta t_i$ and $\E \!Y_j(t_i)=e^{-2G^2t_i^3/3} \sin \theta t_i$ now, see \cref{equation:evaluating_EY}, we have $\E \!S(t_i)=e^{-2G^2t_i^3/3}e^{i\theta t_i}$ ($\theta=2\delta_0$). By demanding precision $C(t_i)/4$, we get
\begin{equation}
\begin{aligned}
    |S(t_i)-e^{-2G^2t_i^3/3}e^{i\theta t}| < \frac{C(t_i)}{4} \leq \frac{e^{-2G^2t_i^3/3}}{4} < \frac{e^{-2G^2t_i^3/3}}{2}.
\end{aligned}
\end{equation}
We can now invoke \cref{lemma:lemma8_hu2025} with $S(t_i)/e^{-2G^2t_i^3/3}$ taking the role of $Z(\frac{\pi}{b-a})$. The ensuing analysis is similar to that of \cref{algorithm:rfe} above. We have $m_i=O\left(\frac{\log (N/\eta)}{C(t_i)^2}\right)$, which now depends on the iteration $i$.

The total number of samples is
\begin{equation}
\begin{aligned}
    K = \sum_{i=1}^N m_i = \log \frac{N}{\eta}  \sum_{i=1}^N e^{4G_{\text{max}}^2t_i^3/3}.
\end{aligned}
\end{equation}
Since $t_N \sim 1/\varepsilon$ and $t_i = t_N(2/3)^{N-i}$, applying \cref{lemma:convexity_lemma} below with $\lambda = 4G_{\text{max}}^2t_N^3$ and $r=(2/3)^3 = 8/27$, we get 
\begin{equation}
\begin{aligned}
    \sum_{i=1}^N e^{4G_{\text{max}}^2t_i^3/3} \sim N + e^{4G_{\text{max}}^2/3\varepsilon^3}.
\end{aligned}
\end{equation}

The total runtime is 
\begin{equation}
\begin{aligned}
    T = \sum_{i=1}^N m_it_i &= \log \frac{N}{\eta}  \sum_{i=1}^N e^{4G_{\text{max}}^2t_i^3/3} t_i\\
    &= \log \frac{N}{\eta}\cdot e^{4G_{\text{max}}^2/3\varepsilon^3} \cdot \frac{1}{\varepsilon}
\end{aligned}
\end{equation}
because the runtime at the last stage $t_N \sim 1/\varepsilon$ dominates.

All in all, we have
\begin{equation}\begin{aligned}
    \text{max sample runtime, } &\max_i t_i = O(1/\varepsilon)\\
    \text{number of samples, } &K = \sum_{i=1}^N m_i =  O\left(\left(\log \frac{\delta_{\text{max}}}{\varepsilon} + \exp(4G_{\text{max}}^2/3\varepsilon^3) \right)\left(\log \frac{1}{\eta}+\log \log \frac{\delta_{\text{max}}}{\varepsilon}\right)\right)\\
    \text{total runtime, } &T = \sum_{i=1}^N m_it_i = O\left(\left(\frac{\exp(4G_{\text{max}}^2/3\varepsilon^3)}{\varepsilon}\right)\left(\log \frac{1}{\eta}+\log \log \frac{\delta_{\text{max}}}{\varepsilon}\right)\right).
\end{aligned}\end{equation} 
Due to the $\exp(4G_{\text{max}}^2/3\varepsilon^3)$ factor, we see that \cref{algorithm:extended_ramsey_protocol} is largely useful in the regime where $G_{\text{max}}^2/\varepsilon^3 \sim \log \frac{1}{\varepsilon}$, or equivalently $G_{\text{max}} \sim \left(\varepsilon^3 \log \frac{1}{\varepsilon}\right)^{1/2}$ -- beyond that the complexity grows extremely rapidly. However, we also observe that this is ultimately an artefact of the presence of $t^3$ in the exponent. In the case of \textit{static} disorder (with Gaussian pdf), we retain the qualitative behaviour of \cref{example:channel_description_pdf_single_qubit_SEH/SEU}, but there is no longer any $t$-dependence in $\mu,\sigma^2$. In this situation \cref{algorithm:extended_ramsey_protocol} gives us a \textit{Heisenberg-limited} estimation protocol in the presence of noise of strength $\sigma = O(1)$ (recall that $\sigma^2$ is the variance of the Gaussian pdf, which characterises the strength of the disorder). This may be somewhat surprising, since the presence of noise generally degrade the performance of estimation protocols which attain the Heisenberg limit under noiseless settings \cite{huelga1997improvement,escher2011general,demkowicz2012elusive}. \hfill $\square$

\begin{lemma}\label{lemma:convexity_lemma}
    Let $\lambda>0$ and $0 < r \leq 1$. Then
    \begin{equation}
    \begin{aligned}
        n + (e^\lambda -1) \leq \sum_{i=0}^{n-1} e^{\lambda r^i} \leq n + \frac{1}{1-r}(e^\lambda -1).
    \end{aligned}
    \end{equation}
    In other words, fixing $r = \Theta(1)$ we have
    \begin{equation}
    \begin{aligned}
        \sum_{i=0}^{n-1} e^{\lambda r^i} \sim n + e^\lambda.
    \end{aligned}
    \end{equation}
\end{lemma}
\begin{proof}
    By the convexity of $e^x$ on $[0,1]$, Jensen gives $e^{\lambda x} = e^{(1-x)\cdot 0 + x\lambda} \leq (1-x)e^0 + xe^\lambda = 1+x(e^\lambda-1)$. Rewriting $\sum_{i=0}^{n-1} e^{\lambda r^i} = n + \sum_{i=0}^{n-1} (e^{\lambda r^i}-1)$, we have
    \begin{equation}
    \begin{aligned}
        n + (e^\lambda -1) &\leq n + \sum_{i=0}^{n-1} (e^{\lambda r^i}-1)\\
        &\leq n + \sum_{i=0}^{n-1} r^i (e^\lambda-1)\\
        &\leq n + \frac{1-r^n}{1-r}(e^\lambda-1)\\
        &\leq n + \frac{1}{1-r}(e^\lambda -1).
    \end{aligned}
    \end{equation}
\end{proof}

In the setting of \cref{problem:single_qubit_quantum_sensing} we assumed the upper bound $G_{\text{max}}$ was known. We now give a stronger extension of RFE, \cref{algorithm:strong_extended_rfe} in which $G_{\text{max}}$ is not assumed to be known.

\begin{algorithm}
\caption{Strong Extended RFE}
\label{algorithm:strong_extended_rfe}
\begin{algorithmic}[1]

\Require 
\Statex - Black-box access to the channel $\mathcal{S}_t$ (see \cref{equation:channel_description_SEH/SEU}), where $H(t)$ is given in \cref{problem:single_qubit_quantum_sensing}, but now $G_{\text{max}}$ is also unknown.
\Statex - Desired precision $\varepsilon>0$ for the estimate $\hat{\delta}_0$.
\Statex - Success probability $1-\eta$, $\eta>0$.

\Ensure 
\Statex - Estimate $\hat{\delta}_0$ satisfying $|\hat{\delta}_0-\delta_0| < \varepsilon$, with success probability $1-\eta$.

\Statex
\Pseudocode
\State Initialise $a=-\delta_{\text{max}}, b=\delta_{\text{max}}$.
\For{$i=1,\dots,N=O\left(\log \frac{\delta_{\text{max}}}{\varepsilon}\right)$}
    \State Set $t_i = O\left((\frac{1}{\delta_{\text{max}}})(\frac{3}{2})^{i-1}\right)$.
    \textcolor{DeepPink3}{\State Set $Z(t_i)=0$, $R=\frac{1}{2}$.}
    \algrenewtext{While}[1]{%
    \textcolor{DeepPink3}{\algorithmicwhile\ #1\ \algorithmicdo}%
    }
    \algrenewtext{EndWhile}{%
    \textcolor{DeepPink3}{\algorithmicend\ \algorithmicwhile}%
    }
    \While{$|Z(t_i)| \leq 3R$} \textcolor{DeepPink3}{\Comment{Ensure \cref{lemma:lemma8_hu2025}} can be used}
        \For{$j=1,\dots,\textcolor{DeepPink3}{m_{i,R}=O\left(\frac{1}{R^2}\left(\log \frac{N}{\eta} + \log \log \frac{1}{R} \right)\right)}$} \Comment{Ramsey subroutine}
            \State Initialize the qubit state $\ket{+} = \frac{1}{\sqrt{2}}(\ket{0}+\ket{1})$.
            \State Evolve $\ket{+}$ under $\mathcal{S}_{t_i}$.
            \State Measure the observable $\sigma_y$ and record the measurement outcome $Y_j(t_i) \in \{\pm1\}$.
        \EndFor
        \For{$j=1,\dots,\textcolor{DeepPink3}{m_{i,R}=O\left(\frac{1}{R^2}\left(\log \frac{N}{\eta} + \log \log \frac{1}{R} \right)\right)}$} \Comment{Ramsey subroutine}
            \State Initialize the qubit state $\ket{+} = \frac{1}{\sqrt{2}}(\ket{0}+\ket{1})$.
            \State Evolve $\ket{+}$ under $\mathcal{S}_{t_i}$.
            \State Measure the observable $\sigma_x$ and record the measurement outcome $X_j(t_i) \in \{\pm1\}$.
        \EndFor
        \State Compute $\hat{X}(t_i) = \frac{1}{\textcolor{DeepPink3}{m_{i,R}}}\sum_{j=1}^{\textcolor{DeepPink3}{m_{i,R}}} X_j(t_i)$, $\hat{Y}(t_i) = \frac{1}{\textcolor{DeepPink3}{m_{i,R}}}\sum_{j=1}^{\textcolor{DeepPink3}{m_{i,R}}} Y_j(t_i)$. Let $S(t_i)\coloneq \hat{X}(t_i) +$ 
        \Statex \hspace{\algorithmicindent}\hspace{\algorithmicindent}%
        $i\hat{Y}(t_i)$.
        \textcolor{DeepPink3}{\State $Z(t_i) \gets S(t_i)$.}
        \textcolor{DeepPink3}{\State $R \gets R/2$.}
    \EndWhile
    \State Compute $\im\left( \exp\left(-i\frac{(a+b)\pi}{2(b-a)}\right)\textcolor{DeepPink3}{Z(t_i)} \right)$.
    \If{$\im \leq 0$} \Comment{Invoke \cref{lemma:lemma8_hu2025}}
        \State $(a,b) \gets (a,\frac{a+2b}{3})$
    \ElsIf{$\im > 0$}
        \State $(a,b) \gets (\frac{2a+b}{3},b)$
    \EndIf
\EndFor
\State \textbf{print} $\hat{\delta}_0 \in [a,b]$

\end{algorithmic}
\end{algorithm}

\paragraph{Analysis of \cref{algorithm:strong_extended_rfe}.}
When $G_{\text{max}}$ was known, it was relatively simple to modify the RFE in \cref{algorithm:extended_rfe}. Let us recap what was done. \cref{lemma:lemma8_hu2025} involved a random variable $Z(t)$ satisfying $|Z(t)-e^{i\theta t}|< 1/2$. But in the setting of \cref{problem:single_qubit_quantum_sensing}, the quantity of interest is $|Z(t)-C(t)e^{i\theta t}|$, with $0<C(t)\leq 1$ in general. One can check easily that it is now insufficient to have $|Z(t)-C(t)e^{i\theta t}|<1/2$ -- \cref{lemma:lemma8_hu2025} would not work anymore. If however, we demand a precision of $C(t)/2$ instead of merely $1/2$, then $|Z(t)-C(t)e^{i\theta t}|<C(t)/2 \iff |Z(t)/C(t)-e^{i\theta t}|<1/2$, so \cref{lemma:lemma8_hu2025} can be used with $Z(t)/C(t)$ playing the role of $Z(t)$.

The problem now is we do not know the value of $C(t)$ a priori, and is thus unable to ascertain the number of raw samples required to construct $Z(t)$ satisfying $C(t)/2$-precision. However, if we know a lower bound to $C(t)$, say $C_{\text{min}}(t)$, then there is an easy workaround. By stipulating a precision of $C_{\text{min}}(t)/2$ we get
\begin{equation}
\begin{aligned}
    |Z(t)-C(t)e^{i\theta t}| < C_{\text{min}}(t)/2 &\implies |Z(t)-C(t)e^{i\theta t}| < C(t)/2\\
    &\iff |Z(t)/C(t)-e^{i\theta t}| < 1/2
\end{aligned}
\end{equation}
which is exactly what was done in \cref{algorithm:extended_rfe}, with $C_{\text{min}}(t) = e^{-2G_{\text{max}}^2t^3}/3$ (note that all along there is a time-dependence in the $R$-factor).

What happens now that we do not assume to know $G_{\text{max}}$ a priori? In this case no $C_{\text{min}}(t)$ exists, and the above workaround fails. However, here we show that this problem can be overcome. The idea is to start with an initial precision, say $R=1/2$, and construct $Z(t)$ such that $|Z(t)-C(t)e^{i\theta t}|<R$. A priori it need not hold that $|Z(t)-C(t)e^{i\theta t}|<C(t)/2$, which is the true objective if our intention is to apply \cref{lemma:lemma8_hu2025} (remember, not knowing $C(t)$ is the ultimate source of our troubles). It does hold however, if $R<C(t)/2$. So we need a way to check when this condition holds.

It turns out that there is a simple way to do so! Namely, we simply check if $|Z(t)| > 3R$. The reason is as follows: assuming $|Z(t)-C(t)e^{i\theta t}|<R$, then by the reverse triangle inequality we also have
\begin{equation}
\begin{aligned}
    &|Z(t)|-C(t) \leq ||Z(t)|-C(t)| < |Z(t)-C(t)e^{i\theta t}| < R\\
    \implies &|Z(t)|-R < C(t).
\end{aligned}
\end{equation}
If we want $R<C(t)/2$, then it is sufficient to ask that $R < \frac{|Z(t)|-R}{2}$. Rearranging, we get $|Z(t)| > 3R$.

The strategy now is as follows. First note that in \cref{problem:single_qubit_quantum_sensing}, $C(t)=e^{-2G^2t^3/3}$. For the $i$th run in the outermost for-loop (where we decide how to cut down the interval containing $\theta$) we start with an initial precision, say $R=1/2$. We collect a number of samples, given by Hoeffding, to ensure that $|Z(t_i)-C(t_i)e^{i\theta t_i}|<R$ with high probability. To ensure that $|Z(t_i)-C(t_i)e^{i\theta t_i}|<C(t_i)/2$, we check whether the condition $|Z(t_i)| > 3R$ holds or not. If yes then we are done; if not then we let $R \gets R/2$ and repeat. This is precisely what the while-loop in \cref{algorithm:strong_extended_rfe} is doing. Eventually this while-loop must terminate, because for a fixed $C(t_i)$, $\log \frac{1}{C(t_i)}$ iterations is sufficient for $R$ to become smaller than $C(t_i)/2$. Once this loop terminates, we would have the $Z(t_i)$ satisfying $|Z(t_i)-C(t_i)e^{i\theta t_i}|<C(t_i)/2$. Then proceed as before -- apply \cref{lemma:lemma8_hu2025} and cut down the interval containing $\theta$.
 
Now we analyse the resources required to ensure the entire protocol works below failure probability $\eta$. As before, since there are $N$ iterations of the shaving off process, within the $i$th iteration we need to guarantee a $\eta/N$ failure probability -- in particular, this means the while-loop must produce the $Z(t_i)$ with this failure probability. 

Now consider just this $i$th iteration. The while-loop within this iteration takes 
\begin{equation}
\begin{aligned}
    \log \frac{1}{C(t_i)}
\end{aligned}
\end{equation}
runs to terminate. For the $k$th iteration in the while-loop, where $k=1,\dots,\log \frac{1}{C(t_i)}$ , it suffices to let its failure probability be $\alpha_k = \frac{\eta}{N}\cdot \frac{6}{\pi^2k^2}$, since $\sum_{k=1}^\infty \alpha_k = \frac{\eta}{N} \cdot \frac{6}{\pi^2} \sum_{k=1}^\infty \frac{1}{k^2} = \frac{\eta}{N}$. The precision at the $k$th iteration is $R_k=1/2^k$, so by Hoeffding the number of samples $m_{i,R_k}$ required for the Ramsey subroutines is 
\begin{equation}
\begin{aligned}
    m_{i,R_k} &= O\left(\frac{1}{R_k^2} \log \frac{1}{\alpha_k}\right)\\
    &= O\left(\frac{1}{R_k^2} \log \frac{Nk^2}{\eta}\right)\\
    &= O\left(\frac{1}{R_k^2} \log \left(\frac{N}{\eta}\cdot \log^2 \frac{1}{R_k} \right)\right)\\
    &= O\left(\frac{1}{R_k^2} \left(\log \frac{N}{\eta} + \log \log \frac{1}{R_k} \right)\right)
\end{aligned}
\end{equation}
This is just for the $k$th iteration in the while-loop. For the whole while-loop, the number of samples is
\begin{equation}
\begin{aligned}
    m_i = \sum_{k=1}^{\log 1/C(t_i)} m_{i,R_k} = O\left(\frac{1}{C(t_i)^2} \left(\log \frac{N}{\eta} + \log \log \frac{1}{C(t_i)} \right)\right).
\end{aligned}
\end{equation}
To get the equality we have used $R_k=1/2^k$, and that the sum of a geometric series (which is what we have here modulo $\log \log$ terms) is a constant times the last term in the series, which in this case is $1/C(t_i)^2$.

Finally, summing over the outermost if-loop gives the total number of samples
\begin{equation}
\begin{aligned}
    K = \sum_{i=1}^N m_i = O\left(
    \left(\log \frac{\delta_{\text{max}}}{\varepsilon} + \exp(4G^2/3\varepsilon^3) \right)
    \left(\log \frac{1}{\eta}+\log \log \frac{\delta_{\text{max}}}{\varepsilon}\right) + \exp(4G^2/3\varepsilon^3) \log \frac{G^2}{\varepsilon^3}
    \right)
\end{aligned}
\end{equation}
and total runtime
\begin{equation}
\begin{aligned}
    T = \sum_{i=1}^N m_it_i = O\left(
    \left(\frac{\exp(4G^2/3\varepsilon^3)}{\varepsilon}\right)\left(\log \frac{1}{\eta}+\log \log \frac{\delta_{\text{max}}}{\varepsilon} + \log \frac{G^2}{\varepsilon^3} \right)
    \right)
\end{aligned}
\end{equation}
The analysis of $K$ and $T$ here is very similar to that for \cref{algorithm:extended_rfe} above, where we invoked \cref{lemma:convexity_lemma} for $K$ and domination by the last sample runtime $t_N$ for $T$. 

We emphasize again that all $1/\varepsilon^3$ factors here have origins in the presence of $t_N^3$ factors and the fact that $t_N \sim 1/\varepsilon$. Also again, as in \cref{algorithm:extended_rfe}, \cref{algorithm:strong_extended_rfe} is largely useful only in the regime $G \sim \varepsilon^{3/2} \log^{1/2} \frac{1}{\varepsilon}$, and in the case of static disorder (with Gaussian pdf) it gives us a \textit{Heisenberg-limited} estimation protocol in the presence of noise of strength $\sigma = O(1)$. \hfill $\square$

We summarize our results above in the following table:
\begin{table}[H]
\centering
\begin{tabular}{l|ccc} 

& 
\makecell[c]{Max sample runtime,\\$\max_i t_i$} & 
\makecell[c]{Number of samples,\\$K$} & 
\makecell[c]{Total runtime,\\$T$} \\
\hline

\makecell[c]{Ramsey\\ \cref{algorithm:ramsey_protocol}} & 
$O(1)$ & 
$O(1/\varepsilon^2)$ & 
$O(1/\varepsilon^2)$ \\

&&& \\

\makecell[c]{RFE\\ \cref{algorithm:rfe}} & 
$O(1/\varepsilon)$ &
$O\left(\log \frac{\delta_{\text{max}}}{\varepsilon}\right)$ &
$O(1/\varepsilon)$ \\

&&& \\

\makecell[c]{Extended Ramsey\\ \cref{algorithm:extended_ramsey_protocol}} &
$O(1)$ & 
$O(1/\varepsilon^2)$ & 
$O(1/\varepsilon^2)$ \\

&&& \\

\makecell[c]{Extended RFE\\ \cref{algorithm:extended_rfe}} &
$O(1/\varepsilon)$ &
$O\left(\log \frac{\delta_{\text{max}}}{\varepsilon} + \exp\left(\frac{4G_{\text{max}}^2}{3\varepsilon^3} \right)\right)$ &
$O\left(\frac{\exp(4G_{\text{max}}^2/3\varepsilon^3)}{\varepsilon}\right)$ \\

&&& \\

\makecell[c]{Strong Extended RFE\\ \cref{algorithm:strong_extended_rfe}} &
$O(1/\varepsilon)$ &
$O\left(\log \frac{\delta_{\text{max}}}{\varepsilon}\right) + \widetilde{O}\left(\exp\left(\frac{4G_{\text{max}}^2}{3\varepsilon^3} \right) \right)$ &
$O(1/\varepsilon) \cdot \widetilde{O}\left(\exp\left(\frac{4G_{\text{max}}^2}{3\varepsilon^3} \right) \right)$ \\

\end{tabular}

\par\bigskip

\begin{tabular}{l|cc} 

& 
Useful in the regime... & 
Comments \\
\hline

\makecell[c]{Ramsey\\ \cref{algorithm:ramsey_protocol}} & 
$G = 0$ & 
Standard quantum limit \\

&& \\

\makecell[c]{RFE\\ \cref{algorithm:rfe}} & 
$G = 0$ &
Heisenberg limit \\

&& \\

\makecell[c]{Extended Ramsey\\ \cref{algorithm:extended_ramsey_protocol}} &
$G \lesssim 1$ & 
Standard quantum limit \\

&& \\

\makecell[c]{Extended RFE\\ \cref{algorithm:extended_rfe}} &
$G \lesssim \varepsilon^{3/2}\log^{1/2}\frac{1}{\varepsilon}$ &
\makecell[c]{Heisenberg limit when $G \lesssim \varepsilon^{3/2}$;\\also Heisenberg limit for static disorder of noise strength $\sigma \lesssim 1$.} \\

&& \\

\makecell[c]{Strong Extended RFE\\ \cref{algorithm:strong_extended_rfe}} &
$G \lesssim \varepsilon^{3/2}\log^{1/2}\frac{1}{\varepsilon}$ &
\makecell[c]{Does not assume $G_{\text{max}}$ is known a priori;\\Heisenberg limit when $G \lesssim \varepsilon^{3/2}$;\\also Heisenberg limit for static disorder of noise strength $\sigma \lesssim 1$.} \\

\end{tabular}

\caption{
A comparison of all five protocols in this section. \cref{algorithm:ramsey_protocol,algorithm:rfe} are existing works in the literature; \cref{algorithm:extended_ramsey_protocol,algorithm:extended_rfe,algorithm:strong_extended_rfe} are designed to handle the stochastic evolution in the parameter to be estimated. The sample complexity and total runtime complexity in \cref{algorithm:rfe,algorithm:extended_rfe,algorithm:strong_extended_rfe} are stated modulo $\log \log$ factors and for simplicity, the $\log \frac{1}{\eta}$ success probability factor is also omitted throughout.
}

\label{table:quantum_sensing_algorithms}
\end{table}

\begin{remark}
    To end this section, we remark that we could consider different SDEs governing $\delta(t)$. For instance, under the Ornstein-Uhlenbeck process (see \cref{example:examples_of_SDEs}) we would have
    \begin{equation}\begin{aligned}
        d\delta(t) &= \gamma(\delta_0-\delta(t))\,dt + G\,dW(t)\\
        \delta(0) &= \delta_0.
    \end{aligned}\end{equation}
    The physical significance of Ornstein-Uhlenbeck is that this model incorporates `damage control', i.e. the mean-reverting term $\gamma(\delta_0-\delta(t))$ counterbalances the diffusive term $G$. A preliminary analysis shows that while the distributions of $\delta(t)$ and $\int_0^t \delta(\tau)\, d\tau$ becomes more complicated, the same techniques from \cref{algorithm:extended_ramsey_protocol,algorithm:extended_rfe,algorithm:strong_extended_rfe} still apply. We do not include it here in this paper, and leave that for future work.
    
\end{remark}

\section{Hamiltonian Simulation under Stochastic Evolution}\label{section:Hamiltonian_simulation}

In this section, we consider the task of Hamiltonian simulation. Specifically, we analyse the qDRIFT protocol when the quantum system is evolved by an SEH. Our choice of task is motivated by qDRIFT being used to implement a subroutine (Hamiltonian reshaping) of the Hamiltonian learning protocol we shall see in \cref{section:Hamiltonian_learning} below. 

Given a Hamiltonian $\widetilde{H} = \sum_j w_j H_j$, its associated unitary channel is $\mathcal{U}_t(\cdot) \coloneq e^{-i\widetilde{H}t}(\cdot)e^{i\widetilde{H}t}$. The idea of qDRIFT \cite{campbell2018random} is to produce an approximating channel $\mathcal{E}_t(\cdot)$ through the following simple procedure. Assuming we have access to all the unitaries $e^{-iH_jt}$, simply sample from the set $\{e^{-iH_jt}\}_j$ according to the distribution $\{p_j\}$ where $p_i \coloneq w_i/\sum_j w_j$. This produces the channel
\begin{equation}\begin{aligned}
    \mathcal{E}_t(\cdot) = \sum_j p_j e^{-iH_jt}(\cdot)e^{iH_jt}.
\end{aligned}\end{equation}
Taylor-expanding both sides one sees that $\mathcal{E}_t$ is equal to $\mathcal{U}_t$ up to first order in $t$, and that
\begin{equation}\begin{aligned}
    d_\diamond(\mathcal{E}_t,\mathcal{U}_t) = O(t^2).
\end{aligned}\end{equation}
Thus, in the spirit of Trotterization, if the desired total evolution time is $t$, then constructing the approximating channels $\mathcal{E}_{t/N}$ for evolution time $t/N$ and composing these $N$ channels yields an approximation to $\mathcal{U}_t$ with error $O(t^2/N)$ in the diamond norm. The outline here will be stated in full detail below when we discuss its analysis for \textit{time-dependent} Hamiltonians in \cref{result:time-dependent_qDRIFT}.

We state the qDRIFT protocol formally in \cref{algorithm:qDRIFT}.
\begin{algorithm}
\caption{qDRIFT}
\label{algorithm:qDRIFT}
\begin{algorithmic}[1]

\Require 
\Statex - The classical description of the Hamiltonian $\widetilde{H} = \sum_j w_j H_j$, where for simplicity we set $\sum_j w_j = 1$.
\Statex - Black-box access to the unitaries $\{e^{-iH_jt}\}_j$.
\Statex - Evolution time $t$.
\Statex - Classical sampler $\textsc{sample}()$ that returns a value $j$ from the probability distribution $\{w_j\}_j$.
\Statex - Desired channel precision $\varepsilon>0$.

\Ensure 
\Statex - A channel $\mathcal{E}_t$ satisfying $d_\diamond(\mathcal{E}_t,\mathcal{U}_t) < \varepsilon$.

\Statex
\Pseudocode
\State Set $\mathcal{U}_t^{\text{app}} = \mathcal{I}$.
\For{$i=1,\dots,N=O(t^2/\varepsilon)$}
    \State $i \gets \textsc{sample}()$
    \State $\mathcal{U}_t^{\text{app}} \gets e^{-iH_it/N} \circ \mathcal{U}_t^{\text{app}}$
\EndFor
\State \Return $\mathcal{U}_t^{\text{app}}$ \Comment{$\mathcal{E}_t = \E_{\text{qDRIFT}}[\mathcal{U}_t^{\text{app}}]$, $\E_{\text{qDRIFT}} = \E_{\{w_j\}_j^{\otimes N}}$}

\end{algorithmic}
\end{algorithm}

When the Hamiltonian of interest is an SEH, we have to take into account both time-dependence and stochasticity in the qDRIFT protocol. We shall build things up one at a time, starting with time-dependence.

\subsection{Incorporating time-dependence}\label{subsection:ham_sim_incorporating_td}
When $\widetilde{H}(t)$ is time-dependent, the unitary channel of interest is
\begin{equation}\begin{aligned}
    \mathcal{U}_{(t,0)}(\rho) = \exp_\mathcal{T}\left(-i\int_0^t \widetilde{H}(\tau)\,d\tau \right) \rho \exp_\mathcal{T}^\dag \left(-i\int_0^t \widetilde{H}(\tau)\,d\tau \right).
\end{aligned}\end{equation}
Here the we use the subscript $(t,0)$ to indicate the time evolution interval.

At first glance one would consider existing time-dependent extensions of qDRIFT \cite{berry2020time} but it turns out that is not suitable, due to differences in the allowed primitives. The reason is this: in the usual Hamiltonian simulation context, whether via Trotterization methods \cite{huyghebaert1990product,poulin2011quantum,childs2021theory,bosse2025efficient} or qDRIFT \cite{campbell2018random,berry2020time}, the goal is to approximate the unitary generated by some complicated, possibly time-dependent Hamiltonian $\widetilde{H}$ using simple building blocks, most often taking the form $e^{-iH_jt}$ for the Hamiltonian components $H_j$ (which could be further broken down into simpler operations via Solovay-Kitaev). That is, the primitives are the $e^{-iH_jt}$'s. For instance in \cite{berry2020time} the simulating channel for $U(t) \coloneq \exp_\mathcal{T}\left(-i\int_0^t \widetilde{H}(\tau)\,d\tau \right)$, $\widetilde{H}(\tau) = \sum_j H_j(\tau)$ is 
\begin{equation}\begin{aligned}
    \mathcal{E}_{(t,0)}(\rho) = \sum_j \int_0^t p_j(\tau) e^{-i\frac{H_j(\tau)}{p_j(\tau)}} \rho e^{i\frac{H_j(\tau)}{p_j(\tau)}}
\end{aligned}\end{equation}
for some suitably chosen pdfs (over $j$ and $\tau$) $p_j(\tau)$. In \cite{berry2020time}, the authors made the assumption that they had the ability to exponentiate the components $H_j(\tau)$ for any time $\tau$, i.e. implement the unitary $U(\tau)=e^{-iH_j(\tau)/p_j(\tau)}$. On the other hand, for our purposes (motivated by Hamiltonian reshaping, see \cref{section:Hamiltonian_learning}) we assume that the primitives we have are the $\exp_\mathcal{T}\left( -i\int H_j(\tau)\,d\tau \right)$'s, i.e. the (time-ordered) unitaries generated by the constituent $H_j(t)$'s. This primitive in fact makes our task easier. As we show now, the original qDRIFT protocol already suffices, although the error analysis requires more care. For clarity, we restate the qDRIFT protocol for time-dependent Hamiltonians in \cref{algorithm:time-dependent_qDRIFT} and \cref{result:time-dependent_qDRIFT}, highlighting the changes in primitives and complexities in \textcolor{DeepPink3}{pink}.

\begin{algorithm}
\caption{qDRIFT for time-dependent Hamiltonians}
\label{algorithm:time-dependent_qDRIFT}
\begin{algorithmic}[1]

\Require 
\Statex - The classical description of the Hamiltonian $\widetilde{H}(t) = \sum_j w_j H_j(t)$, where $\sum_j w_j = 1$.
    \Statex - Black-box access to the unitaries \textcolor{DeepPink3}{$\left\{\exp_\mathcal{T} \left( -i\int_{t_1}^{t_2} H_j(\tau)\,d\tau \right)\right\}_j$, for arbitrary evolution times $t_1,t_2$}.
\Statex - Evolution time $t$.
\Statex - Classical sampler $\textsc{sample}()$ that returns a value $j$ from the probability distribution $\{w_j\}_j$.
\Statex - Desired channel precision $\varepsilon>0$.

\Ensure 
\Statex - A channel $\mathcal{E}_t$ satisfying $d_\diamond(\mathcal{E}_t,\mathcal{U}_t) < \varepsilon$.

\Statex
\Pseudocode
\State Set $\mathcal{U}_t^{\text{app}} = \mathcal{I}$.
\For{$i=1,\dots,\textcolor{DeepPink3}{N=O\left( \frac{t^2 \sup_{\tau \in [0,t]} \max_j \|H_j(\tau)\|^2_\infty}{\varepsilon} \right)}$}
    \State $i \gets \textsc{sample}()$
    \State $\mathcal{U}_t^{\text{app}} \gets \textcolor{DeepPink3}{\exp_\mathcal{T}\left(-i\int_{\frac{(i-1)t}{N}}^{\frac{it}{N}} H_j(\tau)\,d\tau \right)} \circ \mathcal{U}_t^{\text{app}}$
\EndFor
\State \Return $\mathcal{U}_t^{\text{app}}$ \Comment{$\mathcal{E}_t = \E_{\text{qDRIFT}}[\mathcal{U}_t^{\text{app}}]$, $\E_{\text{qDRIFT}} = \E_{\{w_j\}_j^{\otimes N}}$}

\end{algorithmic}
\end{algorithm}

\begin{result}[qDRIFT error bound for time-dependent Hamiltonians]\label{result:time-dependent_qDRIFT}
    Given the Hamiltonian $\widetilde{H}(\tau)$ written as a convex sum
    \begin{equation}\begin{aligned}
        \widetilde{H}(t) = \sum_j w_jH_j(t),
    \end{aligned}\end{equation}
    and assume we have the ability to implement the time-dependent evolutions $\exp_\mathcal{T} \left( -i\int_{t_1}^{t_2} H_j(\tau)\,d\tau \right)$ for any $j$ and evolution start and endtimes $t_1,t_2$. 
    
    Then for an $N$-step qDRIFT protocol \cref{algorithm:time-dependent_qDRIFT} we have
    \begin{equation}\begin{aligned}
        d_\diamond(\mathcal{E}^N_{(t,0)},\mathcal{U}_{(t,0)}) &\leq \frac{2t^2}{N^2} \sum_i \sup_{\tau \in [\frac{it}{N},\frac{(i-1)t}{N}]} \max_j \|H_j(\tau)\|^2_\infty\\
        &\leq \frac{2t^2}{N} \sup_{\tau \in [0,t]} \max_j \|H_j(\tau)\|^2_\infty.
    \end{aligned}\end{equation}
    Here $\mathcal{E}^N_{(t,0)}$ is the channel produced by qDRIFT (the superscript $N$ is to emphasize the dependence on the step number $N$), and $\mathcal{U}_{(t,0)}$ is the unitary channel generated by $\widetilde{H}(t)$. Therefore, with a desired precision $\varepsilon>0$ the number of sampling steps required is
    \begin{equation}\begin{aligned}
        N \gtrsim \frac{t^2 \sup_{\tau \in [0,t]} \max_j \|H_j(\tau)\|^2_\infty}{\varepsilon}.
    \end{aligned}\end{equation}

    For the Hamiltonian reshaping task in \cref{section:Hamiltonian_learning} it further holds that $H_j(t)=U_jH(t)U_j^\dag$ for the reshaping subject $H(t)$. In this case
    \begin{equation}\begin{aligned}
        d_\diamond(\mathcal{E}^N_{(t,0)},\mathcal{U}_{(t,0)}) \leq \frac{2t^2}{N} \sup_{\tau \in [0,t]} \|H(\tau)\|^2_\infty.
    \end{aligned}\end{equation}
\end{result}

\begin{remark}
    Before proving \cref{result:time-dependent_qDRIFT} we make a remark on the assumed primitives, which are
    \begin{equation}\begin{aligned}
        \mathcal{V}_{j}(t_2,t_1) \coloneq \exp_\mathcal{T} \left( -i\int_{t_1}^{t_2} H_j(\tau)\,d\tau \right)
    \end{aligned}\end{equation}
    for all $j$'s. This can be implemented in particular if we have access to $\mathcal{V}_{j}^\dag$, since $\mathcal{V}_{j}(t_2,t_1) = \mathcal{V}_{j}(t_2,0)\mathcal{V}_{j}^\dag(t_1,0)$.

    In the case of Hamiltonian reshaping, $H_{j}(t)$ takes the form $H_{j}(t) = U_jH(t)U_j^\dag$ for some `base' Hamiltonian $H(t)$, so
    \begin{equation}\begin{aligned}
        \mathcal{V}_{j}(t_2,t_1) &= \exp_\mathcal{T} \left( -i\int_{t_1}^{t_2} U_jH(\tau)U_j^\dag\,d\tau \right)\\
        &= U_j \exp_\mathcal{T} \left( -i\int_{t_1}^{t_2} H(\tau)\,d\tau \right) U_j^\dag.
    \end{aligned}\end{equation}
    Here we do not need to have access to $\mathcal{V}_{j}^\dag$. It suffices to have access to $\exp_\mathcal{T} \left( -i\int H(\tau)\,d\tau \right)$ and the ability to implement the unitaries $U_j, U_j^\dag$ instantaneously. This is the called the `discrete quantum control' model, one of the three control modalities considered in \cite{dutkiewicz2024advantage}.
\end{remark}

\begin{proof} 
The proof structure here is similar to the one in \cite{campbell2018random}. We want to simulate the time evolution
\begin{equation}\begin{aligned}
    \exp_\mathcal{T} \left( -i\int_0^t \widetilde{H}(\tau)\,d\tau \right)
\end{aligned}\end{equation}
using only queries to
\begin{equation}\begin{aligned}
    \exp_\mathcal{T} \left( -i\int H_j(\tau)\,d\tau \right).
\end{aligned}\end{equation}
First, consider arbitrary start and end times for evolution. Given the target channel $\mathcal{U}_{(t,t_0)}$ and simulating channel $\mathcal{E}_{(t,t_0)}$ both implementing evolution from initial time $t_0$ to final time $t$, we first write down their Dyson series in terms of the Liouvillian superoperators, see \cref{equation:dyson_series_unitary_channel}. For the target channel we have
\begin{equation}\begin{aligned}
    \mathcal{U}_{(t,t_0)} &= \exp_\mathcal{T}\left( \int_{t_0}^{t} \mathcal{L}(\tau)\,d\tau \right)\\
    &= \sum_{n=0}^\infty \frac{1}{n!} \int_{t_0}^t d\tau_1 \int_{t_0}^t d\tau_2 \;\dots \int_{t_0}^t d\tau_n\,\mathcal{T}\{\mathcal{L}(\tau_1)\mathcal{L}(\tau_2)\dots \mathcal{L}(\tau_n)\}
\end{aligned}\end{equation}
and similarly for the simulating channel we have
\begin{equation}\begin{aligned}
    \mathcal{E}_{(t,t_0)} &= \sum_j w_j \exp_\mathcal{T}\left( \int_{t_0}^{t} \mathcal{L}_j(\tau)\,d\tau \right)\\
    &= \sum_j w_j \sum_{n=0}^\infty \frac{1}{n!} \int_{t_0}^t d\tau_1 \int_{t_0}^t d\tau_2 \;\dots \int_{t_0}^t d\tau_n\, \mathcal{T}\{\mathcal{L}_j(\tau_1)\mathcal{L}_j(\tau_2)\dots \mathcal{L}_j(\tau_n)\}.
\end{aligned}\end{equation}
Here the Liouvillians are $\mathcal{L}(\tau)(X) = -i[\widetilde{H}(\tau),X]$ and $\mathcal{L}_j(\tau)(X) = -i[H_j(\tau),X]$.

We seek to upper bound $\|\mathcal{E}_{(t,t_0)} - \mathcal{U}_{(t,t_0)}\|_\diamond$. To do so we shall expand the Dyson series for the channels above, note the cancellation of terms up to first order in $t-t_0$, and bound the remainder. We have
\begin{equation}
\begin{aligned}
&\|\mathcal{U}_{(t,t_0)}-\mathcal{E}_{(t,t_0)}\|_\diamond\\
&=
\left\|
\textcolor{DeepPink3}{
\cancel{I + \int_{t_0}^t d\tau_1\,\mathcal{L}(\tau_1)} + \sum_{n=2}^{\infty}\frac{1}{n!} \int_{t_0}^t d\tau_1 \int_{t_0}^t d\tau_2 \;\dots \int_{t_0}^t d\tau_n\, \mathcal{T}\{\mathcal{L}(\tau_1)\dots\mathcal{L}(\tau_n)\}
}
\right.\\
&\quad
\left.
-\textcolor{Blue3}{\left(
\cancel{I + \sum_j w_j \int_{t_0}^t d\tau_1\,\mathcal{L}_j(\tau_1)} + \sum_j w_j \sum_{n=2}^\infty \frac{1}{n!} \int_{t_0}^t d\tau_1 \int_{t_0}^t d\tau_2 \;\dots \int_{t_0}^t d\tau_n\, \mathcal{T}\{\mathcal{L}_j(\tau_1)\mathcal{L}_j(\tau_2)\dots \mathcal{L}_j(\tau_n)\}
\right)}
\right\|_\diamond\\
&\leq \sum_{n=2}^{\infty}\frac{1}{n!} \int_{t_0}^t d\tau_1 \int_{t_0}^t d\tau_2 \;\dots \int_{t_0}^t d\tau_n\, \left\|\mathcal{T}\{\mathcal{L}(\tau_1)\dots\mathcal{L}(\tau_n)\} \right\|_\diamond\\
&\qquad+\, \sum_j w_j \sum_{n=2}^\infty \frac{1}{n!} \int_{t_0}^t d\tau_1 \int_{t_0}^t d\tau_2 \;\dots \int_{t_0}^t d\tau_n\, \left\|\mathcal{T}\{\mathcal{L}_j(\tau_1)\mathcal{L}_j(\tau_2)\dots \mathcal{L}_j(\tau_n)\}
\right\|_\diamond\\
&\leq \sum_{n=2}^{\infty}\frac{1}{n!} \int_{t_0}^t d\tau_1 \int_{t_0}^t d\tau_2 \;\dots \int_{t_0}^t d\tau_n\; 2^n \sup_{\tau \in [t_0,t]} \|\widetilde{H}(\tau)\|^n_\infty\\
&\qquad+\, \sum_j w_j \sum_{n=2}^\infty \frac{1}{n!} \int_{t_0}^t d\tau_1 \int_{t_0}^t d\tau_2 \;\dots \int_{t_0}^t d\tau_n\; 2^n \sup_{\tau \in [t_0,t]} \|H_j(\tau)\|^n_\infty\\
&= \sum_{n=2}^{\infty}\frac{2^n(t-t_0)^n}{n!} \sup_{\tau \in [t_0,t]} \|\widetilde{H}(\tau)\|^n_\infty +
\sum_j w_j \sum_{n=2}^\infty \frac{2^n(t-t_0)^n}{n!} \sup_{\tau \in [t_0,t]} \|H_j(\tau)\|^n_\infty\\
&= \sum_{n=2}^{\infty}\frac{2^n(t-t_0)^n}{n!} \left( \sup_{\tau \in [t_0,t]} \|\widetilde{H}(\tau)\|^n_\infty + \sum_j w_j \sup_{\tau \in [t_0,t]} \|H_j(\tau)\|^n_\infty \right).
\end{aligned}
\label{equation:channel_diff_first_order}
\end{equation}
In the second inequality we have invoked \cref{lemma:bound_for_L_tau} below and further bounded each $\|H_j(\tau_i)\|_\infty$ by $\sup_{\tau \in [t_0,t]} \|H_j(\tau)\|_\infty$ (similarly for the $\widetilde{H}(\tau_i)'s$), and in the second equality we have evaluated the simple volume integral. The sum within the parentheses in the last expression can be further simplified. For any $\tau$, $\|H_j(\tau)\|_\infty \leq \max_j \|H_j(\tau)\|_\infty$, from which we also get $\|\widetilde{H}(\tau)\|_\infty \leq \sum_j w_j \|H_j(\tau)\|_\infty \leq \max_j \|H_j(\tau)\|_\infty$. Thus 
\begin{equation}\begin{aligned}
    \sup_{\tau \in [t_0,t]} \|\widetilde{H}(\tau)\|^n_\infty + \sum_j w_j \sup_{\tau \in [t_0,t]} \|H_j(\tau)\|^n_\infty \leq 2 \sup_{\tau \in [t_0,t]} \max_j \|H_j(\tau)\|^n_\infty
\end{aligned}\end{equation}
and we get
\begin{equation}\begin{aligned}
    \|\mathcal{E}_{(t,t_0)}-\mathcal{U}_{(t,t_0)}\|_\diamond \leq 2\sum_{n=2}^{\infty}\frac{2^n(t-t_0)^n}{n!} \sup_{\tau \in [t_0,t]} \max_j \|H_j(\tau)\|^n_\infty
\end{aligned}\end{equation}
or equivalently
\begin{equation}\begin{aligned}
    d_\diamond(\mathcal{E}_{(t,t_0)},\mathcal{U}_{(t,t_0)}) \leq \sum_{n=2}^{\infty}\frac{2^n(t-t_0)^n}{n!} \sup_{\tau \in [t_0,t]} \max_j \|H_j(\tau)\|^n_\infty.
\end{aligned}\end{equation}
Next we make use of an exponential tail bound which states that $\sum_{n=2}^\infty \frac{x^n}{n!} \leq \frac{x^2}{2}e^x$ for all $x>0$. This result is a simple consequence of Taylor's remainder theorem (in the Lagrange form): for the exponential function $e^x$ we have that $e^x = 1+x+\frac{x^2}{2}e^\xi$ for some $0<\xi<x$, thus $\sum_{n=2}^\infty \frac{x^n}{n!} = e^x-1-x = \frac{x^2}{2}e^\xi < \frac{x^2}{2}e^x$. Plugging in $x = 2(t-t_0)\sup_{\tau \in [t_0,t]} \max_j \|H_j(\tau)\|_\infty$ yields our final bound
\begin{equation}\begin{aligned}
    d_\diamond(\mathcal{E}_{(t,t_0)},\mathcal{U}_{(t,t_0)}) \leq 2(t-t_0)^2\sup_{\tau \in [t_0,t]} \max_j \|H_j(\tau)\|^2_\infty \exp\left( 2(t-t_0)\sup_{\tau \in [t_0,t]} \max_j \|H_j(\tau)\|_\infty \right).
\end{aligned}\end{equation}
While we have let the evolution time be $t-t_0$ for generality, ultimately $t-t_0$ is intended to be small in the qDRIFT protocol. If we sample $N$ times from $\{H_j\}$, the evolution time for each sample has length $t/N \ll 1$. Thus letting $t-t_0 \gets t/N$ we arrive at 
\begin{equation}\begin{aligned}
    d_\diamond(\mathcal{E}_{(\frac{t}{N},0)},\mathcal{U}_{(\frac{t}{N},0)}) &\leq \frac{2t^2}{N^2}\sup_{\tau \in [0,\frac{t}{N}]} \max_j \|H_j(\tau)\|^2_\infty \exp\left( \frac{2t^2}{N^2}\sup_{\tau \in [0,\frac{t}{N}]} \max_j \|H_j(\tau)\|_\infty \right)\\
    &\approx \frac{2t^2}{N^2}\sup_{\tau \in [0,\frac{t}{N}]} \max_j \|H_j(\tau)\|^2_\infty.
\end{aligned}\end{equation}
Achieving a total simulation period $[0,t]$ amounts to simulating 
\begin{equation}\begin{aligned}
    \mathcal{U}_{(t,0)} = \prod_{1 \leq i \leq N}^\leftarrow \mathcal{U}_{(\frac{it}{N},\frac{(i-1)t}{N})} = \mathcal{U}_{(t,\frac{(N-1)t}{N})}\mathcal{U}_{(\frac{(N-1)t}{N},\frac{(N-2)t}{N})} \dots \mathcal{U}_{(\frac{2t}{N},\frac{t}{N})}\mathcal{U}_{(\frac{t}{N},0)}
\end{aligned}\end{equation}
with
\begin{equation}\begin{aligned}\label{equation:simulating_channel_decomposition}
    \mathcal{E}^N_{(t,0)} = \prod_{1 \leq i \leq N}^\leftarrow \mathcal{E}_{(\frac{it}{N},\frac{(i-1)t}{N})}.
\end{aligned}\end{equation}
Note that the $i$th constituent in the product in \cref{equation:simulating_channel_decomposition} is obtained by sampling $j$, then implementing
\begin{equation}\begin{aligned}
    \exp_\mathcal{T}\left(-i\int_{\frac{(i-1)t}{N}}^{\frac{it}{N}} H_j(\tau)\,d\tau \right).
\end{aligned}\end{equation}
Each constituent has an evolution time length $t/N$. From the compositional subadditivity of the diamond norm for quantum channels (see \cref{appendix:quantum_info} below) we get
\begin{equation}\begin{aligned}
    d_\diamond(\mathcal{E}^N_{(t,0)},\mathcal{U}_{(t,0)}) &\leq \frac{2t^2}{N^2} \sum_i \sup_{\tau \in [\frac{it}{N},\frac{(i-1)t}{N}]} \max_j \|H_j(\tau)\|^2_\infty\\
    &\leq \frac{2t^2}{N} \sup_{\tau \in [0,t]} \max_j \|H_j(\tau)\|^2_\infty.
\end{aligned}\end{equation}
Note that while the final bound is neater, the tighter penultimate bound may be useful if $H_j(\tau)$ increases rapidly throughout $[0,t]$.

Finally, in our case of interest it further holds that $H_j=U_jHU_j^\dag$, so $\|H_j\|_\infty$ is the same for all $j$ by the unitary invariance of the operator norm. Therefore in this case
\begin{equation}\begin{aligned}
    d_\diamond(\mathcal{E}^N_{(t,0)},\mathcal{U}_{(t,0)}) &\leq \frac{2t^2}{N^2} \sum_i \sup_{\tau \in [\frac{it}{N},\frac{(i-1)t}{N}]} \|H(\tau)\|^2_\infty\\
    &\leq \frac{2t^2}{N} \sup_{\tau \in [0,t]} \|H(\tau)\|^2_\infty.
\end{aligned}\end{equation}
This completes the proof.
\end{proof}

\begin{lemma}\label{lemma:bound_for_L_tau}
    For a Liouvillian superoperator $\mathcal{L}(\tau)(X) = -i[H(\tau),X]$, we have 
    \begin{equation}\begin{aligned}
        \|\mathcal{L}(\tau)\|_\diamond \leq 2\|H(\tau)\|_\infty.
    \end{aligned}\end{equation}
    As an immediate consequence we obtain
    \begin{equation}\begin{aligned}
        \left\|\prod_{i=1}^n \mathcal{L}(\tau_i)\right\|_\diamond \leq 2^n \prod_{i=1}^n \|H(\tau_i)\|_\infty.
    \end{aligned}\end{equation}
\end{lemma}
\begin{proof}
    For brevity we omit the time-dependence $\tau$. By definition $\|\Phi\|_\diamond = \sup_{X \neq 0} \frac{\|(\Phi \otimes I)(X)\|_1}{\|X\|_1}$, so it suffices to show that $\frac{\|(\mathcal{L} \otimes I)(X)\|_1}{\|X\|_1} \leq 2\|H\|_\infty$ for all $X \neq 0$. We have 
    \begin{equation}\begin{aligned}
        \|(\mathcal{L} \otimes I)(X)\|_1 &= \|-i[H \otimes I,X]\|_1\\
        &\leq \|(H \otimes I)X\|_1 + \|X(H \otimes I)\|_1\\
        &\leq \|H \otimes I\|_\infty \cdot \|X\|_1 + \|X\|_1 \cdot \|H \otimes I\|_\infty\\
        &= 2\|H\|_\infty\|X\|_1.
    \end{aligned}\end{equation}
    In the third line we have made use of the matrix H\"{o}lder inequality $\|AB\|_1 \leq \|A\|_1\|B\|_\infty$ (with $p,q=1,\infty$).
\end{proof}

\subsection{Incorporating randomness}

With an eye toward Hamiltonian reshaping, in this subsection we shall assume $H_j(t) = U_jH(t)U_j^\dag$ for some base $H(t)$ and unitaries $U_j$, so
\begin{equation}\begin{aligned}
    \widetilde{H}(t) = \sum_j w_j U_jH(t)U_j^\dag.
\end{aligned}\end{equation}
In our SEH setting we have not just time-dependence, but also randomness (that evolves with time). As discussed in \cref{section:SEHs}, this makes all the objects discussed above -- Hamiltonians, unitaries, channels, Liouvillians etc. random variables (also see \cref{remark:potential_technicality,remark:realizations_of_quantum_objects}). Now fix an arbitrary sample path $\omega$. Rerunning through the entire argument in \cref{subsection:ham_sim_incorporating_td} above with this fixed $\omega$, we have the pointwise bound
\begin{equation}\begin{aligned}
\label{equation:diamond_norm_rv_upperbound_pointwise}
    d_\diamond(\mathcal{E}^N_{(t,0)}(\omega),\mathcal{U}_{(t,0)}(\omega)) &\leq \frac{2t^2}{N} \sup_{\tau \in [0,t]} \|H(\tau,\omega)\|^2_\infty\\
    &\leq \frac{2t^2}{N} \left(\sup_{\tau \in [0,t]} \sum_i |\delta_i(\tau,\omega)| \right)^2.
\end{aligned}\end{equation}
Here the second inequality comes from $H(\tau) = \sum_i \delta_i(\tau)H_i$ where the $H_i$ are Pauli strings which have operator norm 1.

Recall now that we describe the evolution under SEHs by quantum channels, see \cref{section:SEHs}, \cref{equation:channel_description_SEH/SEU}. Thus the target channel and simulating channel are respectively $\E_\omega[\mathcal{U}_{(t,0)}(\omega)],\E_\omega[\mathcal{E}^N_{(t,0)}(\omega)]$ or simply $\E \mathcal{U}_{(t,0)},\E \mathcal{E}^N_{(t,0)}$ for brevity. By the triangle inequality and monotonicity of expectations,
\begin{equation}\begin{aligned}
\label{equation:diamond_norm_rv_upperbound_expectation}
    d_\diamond(\E \mathcal{E}^N_{(t,0)},\E \mathcal{U}_{(t,0)}) &\leq \E d_\diamond(\mathcal{E}^N_{(t,0)},\mathcal{U}_{(t,0)})\\
    &\leq \frac{2t^2}{N} \E \sup_{\tau \in [0,t]} \|H(\tau)\|^2_\infty\\
    &\leq \frac{2t^2}{N} \E \left(\sup_{\tau \in [0,t]} \sum_{i=1}^m |\delta_i(\tau)| \right)^2.
\end{aligned}\end{equation}

It remains to evaluate $\E \left(\sup_{\tau \in [0,t]} \sum_{i=1}^m |\delta_i(\tau)| \right)^2$, which we do in the proof of \cref{result:SEH_qDRIFT} below. We also re-present the qDRIFT protocol in \cref{algorithm:SEH_qDRIFT}. We emphasize that \cref{algorithm:SEH_qDRIFT} (and \cref{algorithm:time-dependent_qDRIFT} as mentioned) are \textit{operationally identical} to the original qDRIFT \cref{algorithm:qDRIFT}, but differ in the primitives used and thus the ensuing analysis. In \cref{algorithm:time-dependent_qDRIFT} the primitives are the unitaries generated by the constituent time-dependent Hamiltonians, while here in \cref{algorithm:SEH_qDRIFT} the primitives are the channels generated by the constituent Hamiltonians which are SEHs. We state all these separately for clarity for the reader. The notations involved below will be somewhat bloated due to multiple layers of randomness (one from SDEs and one from qDRIFT), please refer to \cref{remark:two_layers_of_randomness} for clarification.

\begin{algorithm}
\caption{qDRIFT for SEHs}
\label{algorithm:SEH_qDRIFT}
\begin{algorithmic}[1]

\Require 
\Statex - The classical description of the Hamiltonian $\widetilde{H}(t) = \sum_j w_j U_jH(t)U_j^\dag$, where $\sum_j w_j = 1$.
\Statex - The classical description of the SEH $H(t)= \sum_{i=1}^m \delta_i(t)H_i$, see description in \cref{result:SEH_qDRIFT}.
\Statex - Black-box access to the \textcolor{DeepPink3}{channels $\E_\omega \exp_\mathcal{T} \left( -i\int H(\tau,\omega)\,d\tau \right)$ and unitaries $\{U_j\}_j$}.
\Statex - Evolution time $t$.
\Statex - Classical sampler $\textsc{sample}()$ that returns a value $j$ from the probability distribution $\{w_j\}_j$.
\Statex - Desired channel precision $\varepsilon>0$.

\Ensure 
\Statex - A channel $\mathcal{E}_t$ satisfying $d_\diamond \left(\mathcal{E}_t, \textcolor{DeepPink3}{\E_\omega \exp_\mathcal{T} \left( -i\int_0^t \widetilde{H}(\tau,\omega)\,d\tau \right)}\right) < \varepsilon$.

\Statex
\Pseudocode
\State Set $\mathcal{U}_t^{\text{app}} = \mathcal{I}$.
\For{$i=1,\dots,\textcolor{DeepPink3}{N=O_{\delta_{\text{max}},G_{\text{max}}}\left(\frac{t^3m^2\chi^2}{\varepsilon}\right)}$}
    \State $i \gets \textsc{sample}()$
    \State $\mathcal{U}_t^{\text{app}} \gets \textcolor{DeepPink3}{U_j \circ \E_\omega \exp_\mathcal{T}\left(-i\int_{\frac{(i-1)t}{N}}^{\frac{it}{N}} H(\tau,\omega)\,d\tau \right) \circ U_j^\dag} \circ \mathcal{U}_t^{\text{app}}$
\EndFor
\State \Return $\mathcal{U}_t^{\text{app}}$ \Comment{$\mathcal{E}_t = \E_{\text{qDRIFT}}[\mathcal{U}_t^{\text{app}}]$, $\E_{\text{qDRIFT}} = \E_{\{w_j\}_j^{\otimes N}}$}

\end{algorithmic}
\end{algorithm}

\begin{result}[qDRIFT error bound for SEHs]\label{result:SEH_qDRIFT}
    Given the Hamiltonian $\widetilde{H}(\tau)$ written as a convex sum
    \begin{equation}\begin{aligned}
        \widetilde{H}(t) = \sum_j w_j U_jH(t)U_j^\dag,
    \end{aligned}\end{equation}
    where the $U_j$ are unitaries and $H(t)$ is an SEH given in \cref{problem:SEH_Hamiltonian_learning} but does not need to be low-intersection, i.e. $H(t)= \sum_{i=1}^m \delta_i(t)H_i$, where the local parameters $\delta_i(t)$ are governed by the SDEs 
    \begin{equation}\begin{aligned}
        d\delta_{i}(t) &= \sum_{j=1}^\chi G_{ij}\,dW_{j}(t) \quad \text{for } i=1,\dots,m\\
        \delta_{i}(0) &= \delta_{i0}.
    \end{aligned}\end{equation}
    It is also given that $|\delta_{i0}| \leq \delta_{\text{max}}$ for all $i$, and $0 \leq G_{ij} \leq G_{\text{max}}$.
    
    Assume we have the ability to implement the time-dependent evolutions $\E \exp_\mathcal{T} \left( -i\int H(\tau)\,d\tau \right)$, and the unitaries $U_j$ instantaneously. Then for an $N$-step qDRIFT protocol \cref{algorithm:SEH_qDRIFT} we have
    \begin{equation}\begin{aligned}
        d_\diamond(\E \mathcal{E}^N_{(t,0)},\E \mathcal{U}_{(t,0)}) &\leq \frac{2t^2m^2}{N} \left( \delta_{\text{max}} + 2G_{\text{max}}\chi\sqrt{t} \right)^2\\
        &= O_{\delta_{\text{max}},G_{\text{max}}}\left(\frac{t^3m^2\chi^2}{N}\right).
    \end{aligned}\end{equation}
    Here $\E\mathcal{E}^N_{(t,0)}$ is the channel produced by qDRIFT (which we simply denoted by $\mathcal{E}_t$ in \cref{algorithm:SEH_qDRIFT} for brevity) and $\E\mathcal{U}_{(t,0)}$ is the channel generated by $\widetilde{H}(t,\omega)$, where $\E = \E_\omega$, see \cref{remark:two_layers_of_randomness} below.
    
    Therefore, with a desired precision $\varepsilon>0$ the number of sampling steps required is
    \begin{equation}\begin{aligned}
        N \gtrsim \frac{t^2m^2}{\varepsilon} \left( \delta_{\text{max}} + 2G_{\text{max}}\chi\sqrt{t} \right)^2.
    \end{aligned}\end{equation}
\end{result}

\begin{remark}\label{remark:two_layers_of_randomness}
    Note that in the context of \cref{algorithm:SEH_qDRIFT} and \cref{result:SEH_qDRIFT}, there are \textit{two} layers of randomness involved. First, in qDRIFT each sampling step draws $j$ over the probability distribution $\{w_j\}_j$, thus over $N$ steps the sample string $(j_N,\dots,j_1)$ is drawn over the tensor product distribution $\{w_j\}_j^{\otimes N}$. We denote the expectation over $\{w_j\}_j^{\otimes N}$ by $\E_{\text{qDRIFT}} = \E_{\{w_j\}_j^{\otimes N}}$. Second, when SDEs enter the picture we denote the expectation over the underlying source of randomness by $\E_\omega$ -- this is exactly the same thing as $\E$ in \cref{equation:channel_description_SEH/SEU}, and $\E_\Omega$ in \cref{equation:channel_description_state_space}.

    For brevity sometimes the subscripts $\omega$/qDRIFT will be omitted from the expectations, but we have tried our best to make the $\E$ clear from the surrounding context, i.e. whether
    \begin{enumerate}[i.]
        \item $\E=\E_{\text{qDRIFT}}$: context is qDRIFT, see beginning of \cref{section:Hamiltonian_simulation};
        \item $\E=\E_\omega$: context is SDEs, SEHs and channels describing the evolution under SEHs, see \cref{section:SEHs};
        \item $\E=\E_{\text{qDRIFT},\omega} = \E_{\text{qDRIFT}}\E_\omega$: context is qDRIFT for SEHs, see \cref{algorithm:SEH_qDRIFT}, \cref{result:SEH_qDRIFT} and the proofs therein.
    \end{enumerate}
    
\end{remark}

\begin{proof}
    We continue where we left off from \cref{equation:diamond_norm_rv_upperbound_expectation} above. We remind the reader that the $\delta_i(t)$'s are governed by SDEs and in this relatively simple setting look like
    \begin{equation}\begin{aligned}
        \delta_i(\tau) = \delta_i(0) + \sum_{j=1}^\chi G_{ij}W_j(\tau) \sim \mathcal{N}\left(\delta_i(0),\tau\sum_{j=1}^\chi G_{ij}^2\right).
    \end{aligned}\end{equation}
    Vectorially,
    \begin{equation}\begin{aligned}
        \bm{\delta}(\tau) = \bm{\delta}(0) + \bm{G}\bm{W}(\tau) \sim \mathcal{N}(\bm{\delta}(0),\tau\bm{G}\bm{G}^T)
    \end{aligned}\end{equation}
    where
    \begin{equation}\begin{aligned}
        \bm{\delta}(\tau) \in \mathbb{R}^m, \quad \bm{G} \in \mathbb{R}^{m \times \chi}, \quad \bm{W}(\tau) \in \mathbb{R}^\chi.
    \end{aligned}\end{equation}
    
    For notational simplicity let us first define
    \begin{equation}\begin{aligned}
        Z_t \coloneq \sup_{\tau \in [0,t]} \sum_{i=1}^m |\delta_i(\tau)| = \sup_{\tau \in [0,t]} \|\bm{\delta}(\tau)\|_1.
    \end{aligned}\end{equation}
    By triangle we have
    \begin{equation}\begin{aligned}
        Z_t &\leq \|\bm{\delta}(0)\|_1 + \sup_{\tau \in [0,t]} \|\bm{G}\bm{W}(\tau)\|_1\\
        &\leq \|\bm{\delta}(0)\|_1 + \|\bm{G}\|_{2 \to 1}\sup_{\tau \in [0,t]} \|\bm{W}(\tau)\|_2
    \end{aligned}\end{equation}
    where $\|\bm{G}\|_{2 \to 1}$ is the $(2 \to 1)$-operator norm (see \cref{appendix:miscellanea})
    \begin{equation}\begin{aligned}
        \|\bm{G}\|_{2 \to 1} \coloneq \sup_{x \neq 0} \frac{\|\bm{G}x\|_1}{\|x\|_2}.
    \end{aligned}\end{equation}

    Since our goal is to upper-bound $\E \!Z_t^2$, let us recall the $L^2$-norm of random variables, i.e. $\|X\|_{L^2} \coloneq \sqrt{\E \!X^2}$. Being a norm, $\|\cdot\|_{L^2}$ satisfies the triangle inequality $\|X+Y\|_{L^2} \leq \|X\|_{L^2} + \|Y\|_{L^2}$. With the intention to apply this triangle inequality with $X \gets \|\bm{\delta}(0)\|_1$ and $Y \gets \|\bm{G}\|_{2 \to 1}\sup_{\tau \in [0,t]} \|\bm{W}(\tau)\|_2$, the remaining hurdle is then to evaluate $\E \sup_{\tau \in [0,t]} \|\bm{W}(\tau)\|_2^2$. We have
    \begin{equation}\begin{aligned}
        \E \sup_{\tau \in [0,t]} \|\bm{W}(\tau)\|_2^2 &= \E \left( \sup_{\tau \in [0,t]} \sum_{j=1}^\chi |W_j(\tau)|^2 \right)\\
        &\leq \E \left( \sum_{j=1}^\chi \sup_{\tau \in [0,t]} |W_j(\tau)|^2 \right)\\
        &= \sum_{j=1}^\chi \E \sup_{\tau \in [0,t]} |W_j(\tau)|^2\\
        &= \chi \cdot \E \sup_{\tau \in [0,t]} |W(\tau)|^2\\
        &\leq \chi \cdot 4t.
    \end{aligned}\end{equation}
    In the last equality we have invoked Doob's martingale inequality (\cref{theorem:Doob_martingale_ineq}) with $p=2$, applicable because Wiener processes are martingales. Putting things together gives us
    \begin{equation}\begin{aligned}
        \sqrt{\E \!Z_t^2} &\leq \sqrt{\E \left(\|\bm{\delta}(0)\|_1 + \|\bm{G}\|_{2 \to 1}\sup_{\tau \in [0,t]} \|\bm{W}(\tau)\|_2\right)^2}\\
        &\leq \sqrt{\E \left(\|\bm{\delta}(0)\|_1\right)^2} + \sqrt{\E \left(\|\bm{G}\|_{2 \to 1}\sup_{\tau \in [0,t]} \|\bm{W}(\tau)\|_2\right)^2}\\
        &= \|\bm{\delta}(0)\|_1 + \|\bm{G}\|_{2 \to 1} \cdot \sqrt{\E \sup_{\tau \in [0,t]} \|\bm{W}(\tau)\|_2^2}\\
        &= \|\bm{\delta}(0)\|_1 + \|\bm{G}\|_{2 \to 1} \cdot 2\sqrt{\chi t}\\
        & \leq m\delta_{\text{max}} + 2G_{\text{max}}m\chi\sqrt{t}
    \end{aligned}\end{equation}
    Here in the second inequality we have applied the triangle inequality for $L^2$-norms; in the second equality we have evaluated $\E \sup_{\tau \in [0,t]} \|\bm{W}(\tau)\|_2^2$; and in the last equality we have further evaluated $\|\bm{\delta}(0)\|_1 \leq m\delta_{\text{max}}$ and $\|\bm{G}\|_{2 \to 1} \leq \|G_{\text{max}}J_{m \times \chi}\|_{2 \to 1} = G_{\text{max}}m\sqrt{\chi}$, see \cref{lemma:monotonicity_2-1_norm} and \cref{example:examples_of_normA} for the latter derivation.
    
    Therefore we have
    \begin{equation}\begin{aligned}
        d_\diamond(\E \mathcal{E}^N_{(t,0)},\E \mathcal{U}_{(t,0)}) &\leq \frac{2t^2}{N} \E \left(\sup_{\tau \in [0,t]} \sum_{i=1}^m |\delta_i(\tau)| \right)^2\\
        &\leq \frac{2t^2m^2}{N} \left( \delta_{\text{max}} + 2G_{\text{max}}\chi\sqrt{t} \right)^2\\
        &= O_{\delta_{\text{max}},G_{\text{max}}}\left(\frac{t^3m^2\chi^2}{N}\right).
    \end{aligned}\end{equation}
\end{proof}

We end this section with one more result, which is closely related to \cref{result:SEH_qDRIFT}. Above we have evaluated the bound for $d_\diamond(\E_\omega \mathcal{E}^N_{(t,0)}(\omega),\E_\omega \mathcal{U}_{(t,0)}(\omega))$ in \cref{equation:diamond_norm_rv_upperbound_expectation}, but what about the \textit{pointwise} bound 
\begin{equation}\begin{aligned}
    d_\diamond(\mathcal{E}^N_{(t,0)}(\omega),\mathcal{U}_{(t,0)}(\omega)) \leq \frac{2t^2}{N} \left(\sup_{\tau \in [0,t]} \sum_i |\delta_i(\tau,\omega)| \right)^2
\end{aligned}\end{equation}
 in \cref{equation:diamond_norm_rv_upperbound_pointwise}?

Since the $\delta_i(t)$'s are Gaussian random variables (with variance increasing with time), they are supported on the entire real line $\mathbb{R}$, so we can no longer provide a finite upper bound with certainty. But as is characteristic of Gaussian variables, the probability of each $\delta_i(t)$ taking values far away from its mean drops exponentially, i.e. they exhibit strong concentration. We parlay this into concentration for $\sup_{\tau \in [0,t]} \sum_i |\delta_i(\tau)|$ and subsequently for $d_\diamond(\mathcal{U}_{(t,0)},\mathcal{E}^N_{(t,0)})$. The execution of this idea however relies on an astute choice of concentration inequality. For us, it is the Borell-Tsirelson-Ibragimov-Sudakov inequality (see \cref{theorem:Borell-TIS}). We have the following

\begin{result}[qDRIFT pointwise error bound for SEHs]\label{result:pointwise_SEH_qDRIFT}
Let $\mathcal{U}_{(t,0)}$ and $\mathcal{E}^N_{(t,0)}$ be the target and simulating channels as prescribed by the time-dependent qDRIFT protocol \cref{algorithm:SEH_qDRIFT}. Then for any $B>0$ we have 
    \begin{equation}\begin{aligned}
    \Pr\left[ d_\diamond(\mathcal{U}_{(t,0)},\mathcal{E}^N_{(t,0)}) \leq \frac{2t^2m^2}{N}(\delta_{\text{max}} + G_{\text{max}}\chi \sqrt{t}B)^2 \right] \geq 1 - \exp\left( \frac{-(B-2)^2\chi}{2} \right).
\end{aligned}\end{equation}
So simply choosing $B=O(1)$, say $B=10$, yields
\begin{equation}\begin{aligned}
    d_\diamond(\mathcal{U}_{(t,0)},\mathcal{E}^N_{(t,0)}) = O_{\delta_{\text{max}},G_{\text{max}}}\left(\frac{t^3m^2\chi^2}{N}\right)
\end{aligned}\end{equation}
with extremely high probability. Thus not only does the bound $O_{\delta_{\text{max}},G_{\text{max}}}\left(\frac{t^3m^2\chi^2}{N}\right)$ hold for the \textit{averaged} channels $\E_\omega \mathcal{E}^N_{(t,0)}(\omega),\E_\omega \mathcal{U}_{(t,0)}(\omega)$, it also holds for each realization $\mathcal{E}^N_{(t,0)}(\omega),\mathcal{U}_{(t,0)}(\omega)$, albeit with small failure probability.
\end{result}

\begin{proof}
The proof here is closely related and uses the same notation (naturally) in the proof of \cref{result:SEH_qDRIFT} above. Recall the notation $Z_t = \sup_{\tau \in [0,t]} \|\bm{\delta}(\tau)\|_1$, satisfying $Z_t \leq \|\bm{\delta}(0)\|_1 + \|\bm{G}\|_{2 \to 1}\sup_{\tau \in [0,t]} \|\bm{W}(\tau)\|_2$. As mentioned above, since $\bm{W}(\tau,\omega)$ is unbounded over $\omega \in \Omega$, no finite upper bound for $Z_t(\omega)$ exists which holds for all $\omega \in \Omega$. Our task now is to demonstrate that a bound exists but with small failure probability, i.e. a statement along the lines of $\Pr[Z_t \leq \dots] \geq 1-\exp(\dots)$.

From $Z_t \leq \|\bm{\delta}(0)\|_1 + \|\bm{G}\|_{2 \to 1}\sup_{\tau \in [0,t]} \|\bm{W}(\tau)\|_2$, our problem further reduces to obtaining concentration for the Euclidean norm of a Brownian vector $\sup_{\tau \in [0,t]} \|\bm{W}(\tau)\|_2$, because for any bound $B > 0$
\begin{equation}\begin{aligned}
    \Pr\left[ Z_t \geq \|\bm{\delta}(0)\|_1 + \|\bm{G}\|_{2 \to 1}B \right] \leq \Pr\left[ \sup_{\tau \in [0,t]} \|\bm{W}(\tau)\|_2 \geq B \right].
\end{aligned}\end{equation}
This we do in \cref{proposition:conc_bound_magnitude_Brownian_vector} below. Using the one-sided bound there, we obtain
\begin{equation}\begin{aligned}
    \Pr\left[ Z_t \geq \|\bm{\delta}(0)\|_1 + \|\bm{G}\|_{2 \to 1}\sqrt{\chi t}B \right] \leq \exp\left( \frac{-(B-2)^2\chi}{2} \right).
\end{aligned}\end{equation}
At this point the bound for $Z_t$ still contains terms we do not know, namely $\bm{\delta}(0)$ and $\bm{G}$. We make a further simplification, evaluating $\|\bm{\delta}(0)\|_1 \leq m\delta_{\text{max}}$ and $\|\bm{G}\|_{2 \to 1} \leq G_{\text{max}}m\sqrt{\chi}$ as we did above in the proof of \cref{result:SEH_qDRIFT}. The simplified and usable bound is then
\begin{equation}\begin{aligned}
    \Pr\left[ Z_t \geq m\delta_{\text{max}} + mG_{\text{max}}\chi \sqrt{t} B \right] \leq \exp\left( \frac{-(B-2)^2\chi}{2} \right).
\end{aligned}\end{equation}
Finally, let us convert this to concentration for the diamond norm, \cref{equation:diamond_norm_rv_upperbound_pointwise}. We have
\begin{equation}\begin{aligned}
    \Pr\left[ d_\diamond(\mathcal{U}_{(t,0)},\mathcal{E}^N_{(t,0)}) \geq \frac{2t^2}{N}(\textcolor{DeepPink3}{m\delta_{\text{max}} + mG_{\text{max}}\chi \sqrt{t}B})^2 \right] &\leq \Pr\left[ \frac{2t^2}{N}Z_t^2 \geq \frac{2t^2}{N}(\textcolor{DeepPink3}{m\delta_{\text{max}} + mG_{\text{max}}\chi \sqrt{t}B})^2 \right]\\
    &= \Pr\left[ Z_t \geq \textcolor{Maroon3}{m\delta_{\text{max}} + mG_{\text{max}}\chi \sqrt{t}B} \right]\\
    &\leq \exp\left( \frac{-(B-2)^2\chi}{2} \right).
\end{aligned}\end{equation}
Note that the RHS is independent of time and the system parameter $m$.
\end{proof}

\begin{proposition}\label{proposition:conc_bound_magnitude_Brownian_vector}
    Let $\bm{W}(\tau)\in \mathbb{R}^\chi $ be a Brownian vector. Then for any $\varepsilon > 0$,
    \begin{equation}\begin{aligned}
        \Pr\left[ \left|\sup_{\tau \in [0,t]} \|\bm{W}(\tau)\|_2 - \E\sup_{\tau \in [0,t]} \|\bm{W}(\tau)\|_2\right| \geq \varepsilon \right] \leq 2\exp\left( \frac{-\varepsilon^2}{2t} \right).
    \end{aligned}\end{equation}
    This gives the one-sided bound for any $B > 0$:
    \begin{equation}\begin{aligned}
        \Pr\left[ \sup_{\tau \in [0,t]} \|\bm{W}(\tau)\|_2 \geq B\sqrt{\chi t} \right] \leq \exp\left( \frac{-(B-2)^2\chi}{2} \right).
    \end{aligned}\end{equation}
\end{proposition}
\begin{proof}
    The presence of $\bm{W}(\tau)$ motivates us to consider the Borell-TIS inequality \cref{theorem:Borell-TIS}. However the quantity $\|\bm{W}(\tau)\|_2$ is not Gaussian, so we cannot yet apply Borell-TIS as is, with $T=[0,t]$. This is remedied however, by utilising the \href{https://en.wikipedia.org/wiki/Dual_norm}{dual norm} identity
    \begin{equation}\begin{aligned}
        \|w\|_2 = \sup_{\|v\|_2=1} \langle v,w \rangle.
    \end{aligned}\end{equation}
    With this we have
    \begin{equation}\begin{aligned}
        \sup_{\tau \in [0,t]} \|\bm{W}(\tau)\|_2 = \sup_{\tau \in [0,t],\|v\|_2=1} \langle v,\bm{W}(\tau) \rangle.
    \end{aligned}\end{equation}
    Importantly, for fixed $(\tau,v)$, $\langle v,\bm{W}(\tau) \rangle$ being a linear combination of centered Gaussians is also a centered Gaussian. We are almost ready to apply \cref{theorem:Borell-TIS} with $T= [0,t] \times S^{\chi-1}$ and $X_{\tau,v} = \langle v,\bm{W}(\tau) \rangle$ (here $S^{\chi-1} \coloneq \{v \in \mathbb{R}^\chi: \|v\|_2=1\}$), but first let us verify that the precondition is satisfied, and also obtain $\sigma_T^2$.
    \begin{enumerate}[i.]
        \item (Precondition) We use the precondition $\E\sup_{t \in T} X_t < \infty$, with $T=[0,t] \times S^{\chi-1}$ and $X_{\tau,v} = \langle v,\bm{W}(\tau) \rangle$, see \cref{fact:equivalence_BTIS_preconditions}. We have
        \begin{equation}\begin{aligned}
            \E \sup_{(\tau,v) \in [0,t] \times S^{\chi-1}} \langle v,\bm{W}(\tau) \rangle &= \E \sup_{\tau \in [0,t]} \|\bm{W}(\tau)\|_2\\
            &\leq \E \left( \sum_{j=1}^\chi \sup_{\tau \in [0,t]} |W_j(\tau)|^2 \right)^{1/2}\\
            &\leq \left( \sum_{j=1}^\chi \E \sup_{\tau \in [0,t]} |W_j(\tau)|^2 \right)^{1/2}\\
            &= \left( \chi \cdot \E \sup_{\tau \in [0,t]} |W(\tau)|^2 \right)^{1/2}\\
            &\leq (4\chi t)^{1/2} < \infty.
        \end{aligned}\end{equation}
        Here in the second inequality we have used Jensen with the concavity of $(\cdot)^{1/2}$, and in the final step invoked Doob's martingale inequality (\cref{theorem:Doob_martingale_ineq}) with $p=2$. Precondition satisfied.

        \item (Variance) Now we derive $\sigma_{[0,t] \times S^{\chi-1}}^2$. This is straightforward:
        \begin{equation}\begin{aligned}
            \sigma_T^2 &= \sup_{(\tau,v) \in [0,t] \times S^{\chi-1}} \vari\, \langle v,\bm{W}(\tau) \rangle\\
            &= \sup_{(\tau,v) \in [0,t] \times S^{\chi-1}} \sum_{j=1}^\chi v_j^2 \vari W_j(\tau)\\ 
            &= \sup_{\tau \in [0,t]} \tau = t.
        \end{aligned}\end{equation}
    \end{enumerate}
    Finally, Borell-TIS (\cref{theorem:Borell-TIS}) gives
    \begin{equation}\begin{aligned}
        \Pr\left[ \left|\sup_{\tau \in [0,t]} \|\bm{W}(\tau)\|_2 - \E\sup_{\tau \in [0,t]} \|\bm{W}(\tau)\|_2\right| \geq \varepsilon \right] \leq 2\exp\left( \frac{-\varepsilon^2}{2t} \right).
    \end{aligned}\end{equation}
    Furthermore using $\E\sup_{\tau \in [0,t]} \|\bm{W}(\tau)\|_2 \leq 2\sqrt{\chi t}$ and letting $\varepsilon \gets (B-2)\sqrt{\chi t}$ in the one-sided variant gives us
    \begin{equation}\begin{aligned}
        \Pr\left[ \sup_{\tau \in [0,t]} \|\bm{W}(\tau)\|_2 \geq B\sqrt{\chi t} \right] \leq \exp\left( \frac{-(B-2)^2\chi}{2} \right).
    \end{aligned}\end{equation}
    This completes the proof.
\end{proof}

\begin{remark}
    We note that the $\chi^2$-dependence of the upper bounds in \cref{result:SEH_qDRIFT,result:pointwise_SEH_qDRIFT} can most likely be improved, if further knowledge on the \textit{structure} of $\bm{G}$ is given. For instance if $\bm{G}$ is sparse it should be possible to exploit this for better dependence on $\chi$, although we do not further pursue this aspect here.
\end{remark}

\section{Hamiltonian Learning under Stochastic Evolution}\label{section:Hamiltonian_learning}

In this section, we investigate the task of Hamiltonian learning for SEHs. The setting we shall consider is that in \cite{huangtong2023learning}, where the authors focused on the so-called low-intersection (local) Hamiltonians, and constructed the first Heisenberg-limited parameter-learning protocol for this class of Hamiltonians. We investigate the effects of stochastic evolution on this protocol.

Recall the notion of low-intersection Hamiltonians \cite{haah2022optimal,huangtong2023learning}:
\begin{definition}
    A \textbf{low-intersection} local Hamiltonian is a local Hamiltonian (see \cref{equation:local_Hamiltonian}) $H = \sum_{i\in E} \delta_iH_i$ which has the additional property
    \begin{equation}\begin{aligned}
        \#\{j \neq i: \supp(H_j) \cap \supp(H_i) \neq \emptyset\} = \Theta(1) \quad\text{ for every } i .
    \end{aligned}\end{equation}
    That is, besides $k$-locality which states that $|\supp(H_i)| \leq k = \Theta(1)$ for every $i$, we also have that each interaction term overlaps with at most a constant number of other interaction terms. Equivalently, this also means each qubit participates in at most $O(1)$ interaction terms $H_i$ (thus, geometrically local Hamiltonians are instances of low-intersection Hamiltonians). For low-intersection Hamiltonians, $|E| \coloneq m=\Theta(n)$.  
\end{definition}

With this, we formally state our Hamiltonian learning problem:
\begin{problem}\label{problem:SEH_Hamiltonian_learning}
    Given the low-intersection SEH $H(t)= \sum_{i=1}^m \delta_i(t)H_i$, where wlog we assume the $H_i$ are ($k$-local) Pauli strings, and the local parameters $\delta_i(t)$ are governed by the SDE
    \begin{equation}\begin{aligned}
        d\bm{\delta}(t) &= \bm{G}\,d\bm{W}(t)\\
        \bm{\delta}(0) &= \bm{\delta}_0.
    \end{aligned}\end{equation}
    Equivalently, in component form this is
    \begin{equation}\begin{aligned}
        d\delta_{i}(t) &= \sum_{j=1}^\chi G_{ij}\,dW_{j}(t) \quad \text{for } i=1,\dots,m\\
        \delta_{i}(0) &= \delta_{i0}.
    \end{aligned}\end{equation}
    Here $\bm{\delta}_0$ and $\bm{G}$ are unknown, but it is given that $|\delta_{i0}| \leq \delta_{\text{max}}$ for all $i$, and $0 \leq G_{ij} \leq G_{\text{max}}$.
    
   \textbf{Goal.} Learn $\bm{\delta}_0$.
\end{problem}

The SDE we consider in \cref{problem:SEH_Hamiltonian_learning} is a particularly simple one, since the drift matrix $\bm{G}$ is constant (but whose values are unknown) and does not depend on $\bm{\delta}(t)$ nor $t$. Also note that the parameters $\delta_i(t)$ are in general correlated (via $\bm{G}$), since they could be driven by overlapping Wiener processes $W_j(t)$'s.

The first protocol enabling many-body Hamiltonian (parameter-)learning with Heisenberg-limited scaling was achieved in \cite{huangtong2023learning}, at least for the class of low-intersection Hamiltonians. The more complicated structure-learning variant was subsequently achieved in \cite{bakshi2024structure,hu2025ansatz}. In this article we shall attempt to replicate the overall strategy in \cite{huangtong2023learning}. Their protocol is modular and puts together various subroutines, most notably from Hamiltonian simulation and single-qubit quantum sensing (thus motivating the choice of tasks in \cref{section:quantum_sensing,section:Hamiltonian_simulation} above). In our SEH model the time-dependent stochasticity forms a further source of complication. This affects each subroutine of the noiseless protocol in various ways, and in each step innovations have to be made to account for it, which we have mostly done in \cref{section:quantum_sensing,section:Hamiltonian_simulation} above. In this final section we put things together, yielding a Hamiltonian learning protocol with the same high-level strategy as \cite{huangtong2023learning}, but with key subroutines replaced by the ones we developed above.

First, we review the learning protocol of Huang-Tong-Fang-Su (HTFS). We give a reasonably detailed, but not complete overview, for which the reader can refer to \cite{huangtong2023learning}.

\subsection{Huang-Tong-Fang-Su protocol recap}\label{subsection:HTFS_recap}
When $\bm{G}=\bm{0}$ we are in precisely the setting considered by HTFS. The central idea in their protocol is to `reshape' the given Hamiltonian into a new Hamiltonian having a desirable form. This `reshaping' is implemented by Hamiltonian simulation techniques. Hamiltonian reshaping is used twice. First, to annihilate a few local interaction terms so the remaining terms can be grouped into $k$-sized disjoint patches (where $k=O(1)$ is the locality of the Hamiltonian), so the learning of each patch Hamiltonian can be done in parallel. Then within each patch, Hamiltonian reshaping is utilised once again to remove more terms, that reduces the patch-Hamiltonian learning task to essentially that of single-qubit quantum sensing. We explain each step in detail below.

\paragraph{Hamiltonian reshaping.} Consider the following problem. Given a local Hamiltonian $H = \sum_i \delta_iH_i$, suppose we are interested in the new Hamiltonian $\widetilde{H}$ taking the form
\begin{equation}
\begin{aligned}
    \widetilde{H} = \sum_{k=1}^K w_kU_kHU_k^\dag
\end{aligned}
\end{equation}
where $\sum_k w_k =1$ (thus the $w_k$ can be viewed as probabilities) and the $U_k$'s are unitaries. The value/form of $w_k$ and $U_k$ will be determined by the problem at hand. `Hamiltonian reshaping' simply refers to the Hamiltonian simulation task of implementing $U(\widetilde{H})$ given only access to $U(H)$ and the $U_k$'s. Here $U(\widetilde{H})=e^{-i\widetilde{H}t}$ and $U(H)=e^{-iHt}$ are the unitaries generated respectively by the desired Hamiltonian $\widetilde{H}$ and the original Hamiltonian $H$.

In \cite{huangtong2023learning} and also in this manuscript, to implement Hamiltonian reshaping we utilise qDRIFT \cite{campbell2018random}, a randomized Hamiltonian simulation protocol whose idea we already described in \cref{section:Hamiltonian_simulation} above, see \cref{algorithm:qDRIFT}. Making use of the fact that
\begin{equation}
\begin{aligned}
    e^{-iU_kHU_k^\dag t} = U_ke^{-iHt}U_k^\dag,
\end{aligned}
\end{equation}
according to qDRIFT to implement $e^{-i\widetilde{H}t}$ approximately (viewed as a quantum channel) we simply have to repeat the following $N$ times for large $N$: sample $k$ from the probability distribution $\{w_k\}_k$ and implement $U_ke^{-iH(t/N)}U_k^\dag$. For the error analysis see \cref{algorithm:qDRIFT} again.

\paragraph{Step 1: Decoupling Hamiltonian into disjoint patches.} This is best illustrated by a simple example. Following \cite{huangtong2023learning}, let us consider the inhomogeneous Heisenberg model (but without the external field for simplicity):
\begin{equation}\label{equation:Heisenberg_model}
\begin{aligned}
    H = \sum_{i=1}^{n-1} (\delta_x^{i,i+1}X_iX_{i+1} + \delta_y^{i,i+1}Y_iY_{i+1} + \delta_z^{i,i+1}Z_iZ_{i+1}).
\end{aligned}
\end{equation}
At this stage it is obvious that we cannot group the interaction terms into disjoint patches, such that the interaction terms belonging to different patches do not have overlapping support. To be able to make such a grouping, it is necessary to kill off a few interaction terms. This is where Hamiltonian reshaping comes into play.

Consider a Pauli string $P \in \{I,X,Y,Z\}^{\otimes n}$. From the Pauli commutation relations it is easy to see that if $P$ acts nontrivially on the $i$th qubit, then
\begin{equation}
\begin{aligned}
    \frac{1}{4}(P+X_iPX_i+Y_iPY_i+Z_iPZ_i) = 0,
\end{aligned}
\end{equation}
while if
$P$ acts as the identity on the $i$th qubit, then
\begin{equation}
\begin{aligned}
    \frac{1}{4}(P+X_iPX_i+Y_iPY_i+Z_iPZ_i) = P.
\end{aligned}
\end{equation}
Consider then the Hamiltonian
\begin{equation}
\begin{aligned}
    \widetilde{H} = \frac{1}{4}(H+X_3HX_3+Y_3HY_3+Z_3HZ_3),
\end{aligned}
\end{equation}
i.e. $w_k=1/4$ for $k=1,2,3,4$ and $U_1=I,$ $U_2=X_3$, $U_3=Y_3$ and $U_4=Z_3$, where for concreteness we consider $i=3$. We have
\begin{equation}
\begin{aligned}
    \widetilde{H} = &\delta_x^{1,2}X_1X_{2} + \delta_y^{1,2}Y_1Y_{2} + \delta_z^{1,2}Z_1Z_{2}\\
    &+ \sum_{i=4}^{n-1} (\delta_x^{i,i+1}X_iX_{i+1} + \delta_y^{i,i+1}Y_iY_{i+1} + \delta_z^{i,i+1}Z_iZ_{i+1}).
\end{aligned}
\end{equation}
In the reshaped Hamiltonian $\widetilde{H}$, we now have two disjoint patches, namely the small patch supported on qubits $\{1,2\}$ and the rest of the chain supported on qubits $\{4,\dots,n\}$.

The same idea can then be utilised to achieve $k$-sized disjoint patches in the reshaped Hamiltonian. In this case it suffices to let $w_k=1/4$ for $k=1,2,3,4$ and $U_1=I,$ $U_2=X_3X_6X_9\dots$, $U_3=Y_3Y_6Y_9\dots$ and $U_4=Z_3Z_6Z_9\dots$. The reshaped Hamiltonian is then
\begin{equation}
\begin{aligned}
    \widetilde{H} = \widetilde{H}_{12} + \widetilde{H}_{45} + \widetilde{H}_{78} + \dots
\end{aligned}
\end{equation}
where $\widetilde{H}_{i,i+1} = \delta_x^{i,i+1}X_iX_{i+1} + \delta_y^{i,i+1}Y_iY_{i+1} + \delta_z^{i,i+1}Z_iZ_{i+1}$ is supported on the qubit patch $\{i,i+1\}$.
To conclude, using Hamiltonian reshaping we have obtained a new Hamiltonian $\widetilde{H}$ whose interacting terms can be grouped into $k$-sized disjoint patches. These patches evolve independently from one another, and thus enables the parallel estimation of the patches' parameters.

Note that the terms $\widetilde{H}_{23}$, $\widetilde{H}_{34}$ etc. have been annihilated. By choosing different $\{U_k\}_k$'s, one can similarly construct the reshaped Hamiltonian $\widetilde{H} = \widetilde{H}_{23} + \widetilde{H}_{56} + \widetilde{H}_{89} + \dots$ and $\widetilde{H} = \widetilde{H}_{34} + \widetilde{H}_{67} + \widetilde{H}_{9,10} + \dots$, thus covering all the original terms in $H$.

The Heisenberg Hamiltonian is geometrically local, and as such is low-intersection. For this Hamiltonian it was rather easy to decide how to choose the $w_k$'s and $U_k$'s required for Hamiltonian reshaping. For an arbitrary low-intersection Hamiltonian, with a given interacting edge set $E$ indexing the set of local interaction terms (see \cref{equation:local_Hamiltonian}), \cite{huangtong2023learning} gave a prescription for the construction of the $w_k$'s and $U_k$'s (see Appendices A.2 and B.1 in their paper).

\paragraph{Step 2: Isolating the diagonal Hamiltonian within each patch.} Having obtained $k$-sized disjoint patches, let us now focus on learning the parameters in each patch (this can be done in parallel over the patches). Let $H_C$ denote the Hamiltonian on the patch/cluster $C$, i.e. $H_C$ is the sum of all and only the interacting terms contained in the patch $C$. To build up intuition, let us first focus on the single-qubit case. Let
\begin{equation}\label{equation:general_single_qubit_Hamiltonian}
\begin{aligned}
    H = \delta_x X + \delta_y Y + \delta_z Z.
\end{aligned}
\end{equation}
We want to estimate the parameters $\delta_{x,y,z}$ up to precision $\varepsilon>0$. If there were only an interaction term, say wlog $\delta_x = \delta_y = 0$, then we find ourselves in the setting of \cref{section:quantum_sensing} above, with various protocols (Ramsey, \cref{algorithm:ramsey_protocol} or RFE, \cref{algorithm:rfe}) available. Crucially, both these protocols relies on prior knowledge on the eigenstates of $H$, which were simply $\ket{z,\pm1}$ for $H=\delta_z Z$. In the general case of \cref{equation:general_single_qubit_Hamiltonian}, we do not know the values of $\delta_{x,y,z}$ beforehand, and as such do not know the eigenstates of $H$. To overcome this, we again use Hamiltonian reshaping to remove the $X$ and $Y$ terms. In this case let us define $w_1=w_2=1/2$ and $U_1=I, U_2=Z$, so that
\begin{equation}
\begin{aligned}
    \widetilde{H} = \frac{1}{2}I \cdot H \cdot I + \frac{1}{2}Z \cdot H \cdot Z = \delta_z Z 
\end{aligned}
\end{equation}
takes on a desirable form, amenable to parameter estimation using the tools in \cref{section:quantum_sensing}. Likewise, choosing $U_2= X/Y$ yields $\widetilde{H}=\delta_x X/\delta_y Y$ respectively, and we progressively estimate all the parameters.

Let us next consider two-qubit Hamiltonians, i.e. most generally,
\begin{equation}
\begin{aligned}
    H = \sum_{P,Q \in \{I,X,Y,Z\}} \delta_{PQ} P_1Q_2.
\end{aligned}
\end{equation}
For the same reasons mentioned above, absence of knowledge on the $\delta_{PQ}$'s means we do not know the eigenstates of $H$. So we reshape it into a form whose eigenstates we know. Choose, for example, $w_1=w_2=w_3=w_4=1/4$, and $U_1=I,U_2=X_1,U_3=Z_2, U_4=X_1Z_2$. Then
\begin{equation}
\begin{aligned}
    \widetilde{H} &= \sum_{i=1}^4 w_kU_kHU_k^\dag\\ 
    &= \delta_{II}I_1I_2 + \delta_{XI}X_1I_2 + \delta_{IZ}I_1Z_2 + \delta_{XZ}X_1Z_2. 
\end{aligned}
\end{equation}
While not as simple as the single-qubit case, we do have knowledge on the eigenstates of $H$, because each qubit is only acted on by one (non-identity) Pauli operator. In this case, the eigenstates are
\begin{equation}
\begin{aligned}
    \left\{ \ket{x,\pm1} \otimes \ket{z,\pm1} \right\} = \left\{ \ket{+,0}, \ket{+,1}, \ket{-,0}, \ket{-,1} \right\} 
\end{aligned}
\end{equation}
Armed with the knowledge of these eigenstates, one can now proceed to estimate the associated parameters $\delta_{II},\delta_{XI},\delta_{IZ},\delta_{XI}$ using the method discussed in Step 3 below. Subsequently, for a different choice of reshaping unitaries $\{U_k\}_k$, a different set of parameters could be estimated analogously as above. In this way, we progressively estimate all $\delta_{PQ}$, $P,Q \in \{I,X,Y,Z\}$, thus fully learning the parameters of $H$. Notably, since the $C$'s are disjoint we can run this protocol on \textit{all} the patches $C$ in parallel, and thereby learn all $H_C$'s in parallel.

Following the terminology in \cite{huangtong2023learning}, if $\{\ket{v_i}\}_i$ is a basis of the Hilbert space, we say that an operator $O$ is \textbf{diagonal} relative to $\{\ket{v_i}\}_i$ if $\{\ket{v_i}\}_i$ is an eigenbasis of $O$. With this terminology, we see that $\widetilde{H}$ is a diagonal Hamiltonian with respect to $\left\{ \ket{+,0}, \ket{+,1}, \ket{-,0}, \ket{-,1} \right\}$. Note that the parameters of the diagonal Hamiltonian $\widetilde{H}$ is always a subset of the parameters of the full Hamiltonian $H$, and depends on the choice of the reshaping unitaries $\{U_k\}_k$.

The method discussed above is generalizable to many-qubit Hamiltonians. For a general $k$-sized patch $C$, \cite{huangtong2023learning} gave a prescription on how to obtain the diagonal Hamiltonian $\widetilde{H}_C$ for a fixed choice of $\{U_k\}_k$. Different choices of reshaping unitaries $\{U_k\}_k$ yield different diagonal Hamiltonians, and thus different subsets of parameters to be estimated. Thus one simply runs through all suitable choices of $\{U_k\}_k$ to ensure the full set of parameters in the original $H_C$ are estimated. The details are given in Appendix B.2 of \cite{huangtong2023learning}.

\paragraph{Step 3: Estimating the parameters of the diagonal Hamiltonian.}
Now that we have extracted the diagonal Hamiltonian $\widetilde{H}$ w.r.t. an eigenbasis, let us proceed to estimating the parameters of $\widetilde{H}$. Below we show how the estimation protocol essentially reduces to the case of single-qubit quantum sensing.

Consider the action of $\widetilde{H}$ on the eigenbasis above. They are given by
\begin{equation}
\label{equation:example_two_qubit_diagonal_Hamiltonian}
\begin{aligned}
    \widetilde{H}\ket{+,0} &= \left(\delta_{II} + \delta_{XI} + \delta_{IZ} + \delta_{XZ}\right)\ket{+,0} \eqcolon \kappa_{+0}\ket{+,0}\\
    \widetilde{H}\ket{+,1} &= \left(\delta_{II} + \delta_{XI} - \delta_{IZ} - \delta_{XZ}\right)\ket{+,1} \eqcolon \kappa_{+1}\ket{+,1}\\
    \widetilde{H}\ket{-,0} &= \left(\delta_{II} - \delta_{XI} + \delta_{IZ} - \delta_{XZ}\right)\ket{-,0} \eqcolon \kappa_{-0}\ket{-,0}\\
    \widetilde{H}\ket{-,1} &= \left(\delta_{II} - \delta_{XI} - \delta_{IZ} + \delta_{XZ}\right)\ket{-,1} \eqcolon \kappa_{-1}\ket{-,1}.
\end{aligned}
\end{equation}
If we can estimate the $\kappa$'s, which are the eigenvalues of $\widetilde{H}$, then we can invert
\begin{equation}
\begin{aligned}
    \begin{bmatrix}
        \kappa_{+0}\\
        \kappa_{+1}\\
        \kappa_{-0}\\
        \kappa_{-1}
    \end{bmatrix}
    =
    \begin{bmatrix}
        1 & 1 & 1 & 1\\
        1 & 1 & -1 & -1\\
        1 & -1 & 1 & -1\\
        1 & -1 & -1 & 1
    \end{bmatrix}
    \begin{bmatrix}
        \delta_{II}\\
        \delta_{XI}\\
        \delta_{IZ}\\
        \delta_{XZ}
    \end{bmatrix}
\end{aligned}
\end{equation}
to obtain the $\delta$'s. To do this, we apply the same techniques in \cref{algorithm:ramsey_protocol} or \cref{algorithm:rfe}. We first initiate a superposition of eigenstates, which in this situation must differ only by Hamming distance one, say $\ket{+,0}$ and $\ket{+,1}$, or $\ket{+,0}$ and $\ket{-,0}$, but not $\ket{+,0}$ and $\ket{-,1}$. Consider $\ket{+,0}$ and $\ket{+,1}$ for instance. We have
\begin{equation}
\begin{aligned}
    e^{-i\widetilde{H}t}\left( \frac{\ket{+,0}+\ket{+,1}}{\sqrt{2}} \right) &= \frac{1}{\sqrt{2}}\left( e^{-i\kappa_{+0}t}\ket{+,0} + e^{-i\kappa_{+1}t}\ket{+,1}  \right)\\
    &\sim \frac{1}{\sqrt{2}}\left( \ket{+,0} + e^{-i(\kappa_{+1}-\kappa_{+0})t})\ket{+,1}  \right)
\end{aligned}
\end{equation}
This is now essentially a single-qubit sensing problem. Measuring the observable $O$ whose eigenstates are $\frac{1}{\sqrt{2}}(\ket{+,0} \pm \ket{+,1})$, we get
\begin{equation}
\begin{aligned}
    \tr(Oe^{-i\widetilde{H}t}\rho(0)) = \cos (\kappa_{+1}-\kappa_{+0})t,
\end{aligned}
\end{equation}
where the initialized state $\rho(0)$ is the density operator of $\frac{\ket{+,0}+\ket{+,1}}{\sqrt{2}}$. Similarly, if we initiate the phase-Hadamard superposition (instead of just Hadamard above), we get 
\begin{equation}
\begin{aligned}
    \tr(Oe^{-i\widetilde{H}t}\rho(0)) = \sin (\kappa_{+1}-\kappa_{+0})t,
\end{aligned}
\end{equation}
where now the initialized state $\rho(0)$ is the density operator of $\frac{\ket{+,0}+i\ket{+,1}}{\sqrt{2}}$. Then by the RFE protocol, \cref{algorithm:rfe}, we can estimate $\kappa_{+1}-\kappa_{+0}$ (with Heisenberg-limited scaling). The roles of the observables whose eigenstates are $\frac{1}{\sqrt{2}}(\ket{+,0} \pm \ket{+,1})$ and $\frac{1}{\sqrt{2}}(\ket{+,0} \pm i\ket{+,1})$ are analogous to the roles of the Pauli observables $X,Y$ in \cref{algorithm:rfe}.

In a nutshell, we estimate the \textit{differences} of the $\kappa$'s which differ by Hamming distance one. We start from assuming $\kappa_{+0}=0$, which is wlog because that amounts to fixing a global (undetectable) phase, then obtain $\kappa_{+1}$ from the estimation of $\kappa_{+1}-\kappa_{+0}$. Having obtained $\kappa_{+1}$, we proceed and obtain for instance $\kappa_{-1}$ from the estimation of $\kappa_{-1}-\kappa_{+1}$. And continue to do so until all the $\kappa$'s are obtained. Then perform matrix inversion to get the $\delta$'s.

For the general case of $k$-sized patches, the philosophy is identical, but designing the order of $\kappa$-differences to be estimated is more involved. \cite{huangtong2023learning} demonstrated a procedure for determining this order, making use of a data structure called a `shortest path tree'. We refer the reader to Appendix C.2 of their paper.

\paragraph{Putting things together.} To recap, given a low-intersection Hamiltonian $H$ whose local parameters we want to estimate all within precision $\varepsilon>0$, the HTFS protocol is as follows. First reshape the Hamiltonian, via qDRIFT \cref{algorithm:qDRIFT} with the appropriate $w_k$'s and $U_k$'s, such that it can be grouped into $k$-sized disjoint patches. Then, in parallel, reshape each patch-Hamiltonian to extract a diagonal Hamiltonian on that patch. For these diagonal patch-Hamiltonians, the RFE protocol \cref{algorithm:rfe} and some postprocessing yields estimates of their local parameters. These estimated parameters form only a subset of all the parameters of the original Hamiltonian to be learned. By choosing different $U_k$'s, we run through different diagonal Hamiltonians on each patch (within the second Hamiltonian reshaping loop), then different patch groupings (within the first Hamiltonian reshaping loop). Collectively, these cover all the parameters $\delta_i$ of $H$.

The full algorithm pseudocode and its analysis can be found in Algorithm 2 and Appendix C.3 of \cite{huangtong2023learning}. The HTFS protocol takes total evolution time (omitting the standard failure probability factor $\log(1/\delta)$)
\begin{equation}
\begin{aligned}
    T_{\text{HTFS}} = O\left( \frac{1}{\varepsilon} \right)
\end{aligned}
\end{equation}
where $\varepsilon$ is the desired precision of the estimated parameters, i.e. $\sup_{i} |\hat{\delta}_i - \delta_i| < \varepsilon$.

\subsection{Modifications for SEHs}
When we are in the setting of \cref{problem:SEH_Hamiltonian_learning}, the time-dependent stochasticity of $\bm{\delta}(t)$ (captured by $\bm{G}$) introduces additional complications in the subroutines of HTFS protocol stated above. Let us see how to account for these complications. This subsection will actually be quite short, because we have done most of the heavy lifting in \cref{section:quantum_sensing,section:Hamiltonian_simulation} above!

\paragraph{Hamiltonian reshaping, Step 1 and Step 2.} 
The high-level strategy of Hamiltonian reshaping, i.e. picking the suitable $w_k$'s and $U_k$'s to annihilate certain interaction terms, remains unchanged. However, recall the implementation of Hamiltonian reshaping is via Hamiltonian simulation, in this case qDRIFT, and there is a necessary trade-off between the desired simulation precision $\varepsilon_{\text{qDRIFT}}$ and the number of sampling steps in qDRIFT (in this case, also the number of layers of interleaving unitaries $U_k$), $N_{\text{qDRIFT}}$. Given a fixed evolution time $t$ and desired simulation precision $\varepsilon_{\text{qDRIFT}}$, we see from \cref{algorithm:SEH_qDRIFT} that for $\bm{G}=\bm{0}$, we need
\begin{equation}
\begin{aligned}
    N_{\text{qDRIFT}} \gtrsim \frac{t^2m^2\delta_{\text{max}}^2}{\varepsilon_{\text{qDRIFT}}} = \Omega_{\delta_{\text{max}}} \left( \frac{t^2m^2}{\varepsilon_{\text{qDRIFT}}} \right),
\end{aligned}
\end{equation}
while if $\bm{G} \neq \bm{0}$ we need
\begin{equation}\begin{aligned}
    N_{\text{qDRIFT}} \gtrsim \frac{t^2m^2}{\varepsilon_{\text{qDRIFT}}} \left( \delta_{\text{max}} + 2G_{\text{max}}\chi\sqrt{t} \right)^2 = \Omega_{\delta_{\text{max}},G_{\text{max}}}\left(\frac{t^3m^2\chi^2}{\varepsilon_{\text{qDRIFT}}}\right).
\end{aligned}\end{equation}
That is, for SEHs there is an overhead factor of $\Theta(t\chi^2)$ for $N_{\text{qDRIFT}}$.

\paragraph{Step 3.}
For simplicity and concreteness let us return to the example in \cref{equation:example_two_qubit_diagonal_Hamiltonian}. Consider $\kappa_{+0}$ and $\kappa_{+1}$ again. Under the influence of SDEs, these are now time-dependent processes. We have
\begin{equation}
\begin{aligned}
    \kappa_{+0}(t) &= \delta_{II}(t) + \delta_{XI}(t) + \delta_{IZ}(t) + \delta_{XZ}(t)\\
    &= \underbrace{\delta_{II}(0) + \delta_{XI}(0) + \delta_{IZ}(0) + \delta_{XZ}(0)}_{\kappa_{+0}(0)} + \sum_{i = II,XI,IZ,XZ} \sum_{j=1}^\chi G_{ij}\, dW_j(t)\\
    &\sim \mathcal{N}\left( \kappa_{+0}(0), t \sum_{i = II,XI,IZ,XZ} \sum_{j=1}^\chi G_{ij}^2 \right)
\end{aligned}
\end{equation}
Similarly,
\begin{equation}
\begin{aligned}
    \kappa_{+1}(t) &= \delta_{II}(t) + \delta_{XI}(t) - \delta_{IZ}(t) - \delta_{XZ}(t)\\
    &= \underbrace{\delta_{II}(0) + \delta_{XI}(0) - \delta_{IZ}(0) - \delta_{XZ}(0)}_{\kappa_{+1}(0)} + \sum_{i = II,XI,IZ,XZ} \sum_{j=1}^\chi G_{ij}\, dW_j(t)\\
    &\sim \mathcal{N}\left( \kappa_{+1}(0), t \sum_{i = II,XI,IZ,XZ} \sum_{j=1}^\chi G_{ij}^2 \right),
\end{aligned}
\end{equation}
so
\begin{equation}
\begin{aligned}
    \kappa_{+1}(t)-\kappa_{+0}(t) \sim \mathcal{N}\left( \kappa_{+1}(0)-\kappa_{+0}(0), 2t \sum_{i = II,XI,IZ,XZ} \sum_{j=1}^\chi G_{ij}^2 \right).
\end{aligned}
\end{equation}
As discussed in Step 3 above in \cref{subsection:HTFS_recap}, $\kappa_{+1}(t)-\kappa_{+0}(t)$ is now essentially a single-qubit parameter to be estimated. Furthermore, since $2\sum_{i = II,XI,IZ,XZ} \sum_{j=1}^\chi G_{ij}^2 \leq 8\chi G_{\text{max}}^2$, the `effective' $G_{\text{max}}$ is
\begin{equation}
\begin{aligned}
    G_{\text{max}}^{\text{eff}} = \sqrt{8\chi} G_{\text{max}}.
\end{aligned}
\end{equation}
The extended protocol \cref{algorithm:extended_rfe} (or \cref{algorithm:strong_extended_rfe}) then yields an estimate to $\kappa_{+1}(0)-\kappa_{+0}(0)$, with complexity given in \cref{table:quantum_sensing_algorithms} (with the effective $G_{\text{max}}^{\text{eff}}$ in place of $G_{\text{max}}$). Furthermore, the Heisenberg limit is achieved when
\begin{equation}
\begin{aligned}
    G_{\text{max}}^{\text{eff}} \lesssim \varepsilon^{3/2}
\end{aligned}
\end{equation}
or equivalently,
\begin{equation}
\begin{aligned}
    G_{\text{max}} \lesssim \varepsilon^{3/2}/\sqrt{8\chi}.
\end{aligned}
\end{equation}
This generalizes completely to general $k$-sized patches $C$, with
\begin{equation}
\begin{aligned}
    G_{\text{max}} \lesssim \varepsilon^{3/2}/\sqrt{2^{k+1}\chi}.
\end{aligned}
\end{equation}
Recall that because each patch is $k$-sized, the diagonal Hamiltonian corresponding to this patch has $2^{|C|}=2^k$ interacting terms at most, because each qubit is to be acted on only by either the identity and at most a non-identity Pauli operator. In the two-qubit case of \cref{equation:example_two_qubit_diagonal_Hamiltonian} we saw that qubit 1 was acted on by $\{I,X\}$ and qubit 2 was acted on by $\{I,Z\}$, so the diagonal Hamiltonian had $2^2=4$ interacting terms at most.

We formally state this as
\begin{result}\label{result:SEH_Hamiltonian_learning}
    Given the SEH in \cref{problem:SEH_Hamiltonian_learning}, the modified HTFS protocol, with Hamiltonian reshaping implemented by \cref{algorithm:SEH_qDRIFT} and single-qubit quantum sensing implemented by \cref{algorithm:extended_rfe}, has a total evolution time complexity of
    \begin{equation}
    \begin{aligned}
        T_{\text{HTFS}}^{\text{SEH}} = O\left( \frac{1}{\varepsilon} \cdot \exp\left(\frac{4\cdot 2^{k+1}\chi G_{\text{max}}^2}{3\varepsilon^3} \right) \right)
    \end{aligned}
    \end{equation}
\end{result}
where $\varepsilon$ is the desired precision of the estimated parameters, i.e. $\sup_{i} |\hat{\delta}_{i0} - \delta_{i0}| < \varepsilon$. In particular, the Heisenberg limit is attained if
\begin{equation}
\begin{aligned}
    G_{\text{max}} \lesssim \left(\frac{\varepsilon^3}{2^{k}\chi}\right)^{1/2}.
\end{aligned}
\end{equation}


\newpage
\section*{Acknowledgments}
This work is supported by the National Research Foundation, Singapore through the National Quantum Office, hosted in A$^*$STAR, under its Centre for Quantum Technologies Funding Initiative. The authors warmly thank Itai Arad, Shao Hen Chiew, My Duy Hoang Long, Serge Massar and Zhan Yu for invaluable discussions. We credit ChatGPT for basic tasks such as evaluating Gaussian integrals and producing all the code generating the figures in this manuscript, and more nontrivially,
\begin{enumerate}[i.]
    \item Providing \cref{lemma:convexity_lemma};
    \item Analysing the resource complexity in \cref{algorithm:strong_extended_rfe};
    \item Suggesting Borell-TIS as an appropriate concentration inequality to use in the proof of \cref{result:pointwise_SEH_qDRIFT},
\end{enumerate}
all in an interactive manner.

\addcontentsline{toc}{section}{References}
\bibliographystyle{alpha}
\bibliography{main}


\newpage
\appendix

\section{Miscellanea}\label[appendix]{appendix:miscellanea}

\paragraph{Notation.}
We use the notation $f \sim g$ to mean $f = \Theta(g)$, and similarly $f \lesssim g$ for $f = O(g)$ and $f \gtrsim g$ for $f = \Omega(g)$. As usual, $\widetilde{O}$ hides polylog factors.

For a sequence of operators $O_1, \dots, O_n$ we write
\begin{equation}\begin{aligned}
    \prod_{1 \leq i \leq n}^\leftarrow O_i = O_n\dots O_1, \qquad \prod_{1 \leq i \leq n}^\rightarrow O_i = O_1\dots O_n. 
\end{aligned}\end{equation}

\paragraph{Operator norms.}
Given a linear operator $A: V \to W$ between the normed vector spaces $V,W$, its ($V \to W$)-operator norm is defined as
\begin{equation}\begin{aligned}
    \|A\|_{V \to W} \coloneq \sup_{x \neq 0} \frac{\|Ax\|_W}{\|x\|_V}
\end{aligned}\end{equation}
where $\|\cdot\|_V, \|\cdot\|_W$ are the respective norms on $V,W$. In particular, if $V,W$ are the standard Euclidean spaces with respective $l_p$-, $l_q$-norms, we write
\begin{equation}\begin{aligned}
    \|A\|_{p \to q} \coloneq \sup_{x \neq 0} \frac{\|Ax\|_q}{\|x\|_p}.
\end{aligned}\end{equation}

\begin{lemma}[Monotonicity of $\|\cdot\|_{2 \to 1}$ under entrywise domination]\label{lemma:monotonicity_2-1_norm}
    If $0 \leq A \leq B$, then $\|A\|_{2 \to 1} \leq \|B\|_{2 \to 1}$.
\end{lemma}
\begin{proof}
    For the purposes of just this proof let us use the notation $|x|\coloneq(|x_1|,\dots,|x_n|)$ for the vector $x=(x_1,\dots,x_n)$. Note that $\||x|\|_2 = \|x\|_2$.
    
    By definition $\|A\|_{2 \to 1} = \sup_{\|x\|_2 = 1} \|Ax\|_1$. For any arbitrary $x$ with $\|x\|_2=1$ we have $|(Ax)_i| \leq \sum_j A_{ij}|x_j| \leq \sum_j B_{ij}|x_j| = (B|x|)_i$, so summing over $i$ yields $\|Ax\|_1 \leq \|B|x|\|_1$. Since $\||x|\|_2=1$ too, we get $\|B|x|\|_1 \leq \sup_{\|y\|_2 = 1} \|Bx\|_1 = \|B\|_{2 \to 1}$. We have deduced that for every unit $x$, $\|Ax\|_1 \leq \|B\|_{2\to 1}$, so taking the supremum over $x$ yields the desired result.
\end{proof}

\begin{example}\label{example:examples_of_normA}
    Here we give two simple examples of $\|A\|_{2\to 1}$.
    \begin{enumerate}
        \item $A = I_m$: $\|I_m\|_{2\to 1} = \sqrt{m}$.
        \item $A = J_{m\times \chi}$: $\|J_{m\times \chi}\|_{2\to 1} = m\sqrt{\chi}$.
    \end{enumerate}
\end{example}

\begin{proof}\hfill
    \begin{enumerate}
        \item Here $\chi = m$ and $I_m$ is the identity matrix. We have $\|I_m\|_{2\to 1} = \sup_{\|x\|_2 = 1} \|x\|_1 \leq \sqrt{m}$, using $\|x\|_1 \leq \sqrt{m}\|x\|_2$. Equality is attained for $x = \frac{1}{\sqrt{m}}(1,\dots,1)$.

        \item Here $J_{m\times \chi}$ is the all-ones matrix. For any $x$ we have $\|J_{m \times \chi}x\|_1 = m|\sum_{j=1}^\chi x_j| \leq m\sqrt{\chi}\|x\|_2$, so $\|J_{m \times \chi}\|_{2\to 1} \leq m\sqrt{\chi}$, with equality again attained for $x = \frac{1}{\sqrt{\chi}}(1,\dots,1)$. Here in the last inequality we have used Cauchy-Schwarz to get $|\sum_{j=1}^\chi x_j| \leq \sqrt{\chi}\|x\|_2$.        
    \end{enumerate}
\end{proof}

\section{Quantum Information}\label[appendix]{appendix:quantum_info}

\paragraph{Paulis.} 
Pauli matrices are denoted $\sigma_{x,y,z}$ or $X/Y/Z$. Their corresponding eigenstates are denoted by
\begin{equation}
\begin{aligned}
    &\ket{z,1} = \ket{0}, \quad &&\ket{z,-1} = \ket{1}\\
    &\ket{x,1} = \ket{+}, \quad &&\ket{x,-1} = \ket{-}\\
    &\ket{y,1} = \frac{1}{\sqrt{2}}(\ket{0} + i\ket{1}), \quad &&\ket{y,-1} = \frac{1}{\sqrt{2}}(\ket{0} - i\ket{1}).
\end{aligned}
\end{equation}

\paragraph{Quantum channels.}
For a state $\ket{\psi}$ let $\|\psi\|_2$ denote its standard Euclidean norm. For a linear operator $A$ its trace norm is $\|A\|_1 \coloneq \tr \sqrt{A^\dag A}$ and its spectral/operator norm is $\|A\|_\infty \coloneq \sup_{\|\psi\|_2=1} \|A\ket{\psi}\|_2$. The trace norm induces the trace distance $d_1(A,B)\coloneq \frac{1}{2}\|A-B\|_1$. For density operators $\rho,\sigma$ we have $d_1(\rho,\sigma) \in [0,1]$.

For a linear superoperator $\Phi$ acting on a Hilbert space $\mathcal{H}$, its diamond norm is
\begin{equation}\begin{aligned}
    \|\Phi\|_\diamond \coloneq \sup_{X \neq 0} \frac{\|(\Phi \otimes I_{\dim \mathcal{H}})(X)\|_1}{\|X\|_1} =  \sup_{\|X\|_1=1} \|(\Phi \otimes I_{\dim \mathcal{H}})(X)\|_1.
\end{aligned}\end{equation}
The diamond norm induces the diamond distance
\begin{equation}\begin{aligned}
    d_\diamond(\Phi,\Psi) \coloneq \frac{1}{2} \|\Phi-\Psi\|_\diamond.
\end{aligned}\end{equation}
For quantum channels $\mathcal{E}, \mathcal{F}$ we have $\|\mathcal{E}\|_\diamond=1$ and $d_\diamond(\mathcal{E}, \mathcal{F}) \in [0,1]$.

We need the following properties of the diamond norm for quantum channels:
\begin{enumerate}[i.]
    \item Triangle inequality: $\|\Phi+\Psi\|_\diamond \leq \|\Phi\|_\diamond + \|\Psi\|_\diamond$.
    
    \item Submultiplicativity under composition: $\|\Phi\Psi\|_\diamond \leq \|\Phi\|_\diamond\|\Psi\|_\diamond$.

    \item Subadditivity under composition for quantum channels: $\|\mathcal{E}_n\dots\mathcal{E}_1 - \mathcal{F}_n\dots\mathcal{F}_1\|_\diamond \leq \sum_{i=1}^n \|\mathcal{E}_i - \mathcal{F}_i\|_\diamond$.
\end{enumerate}

\paragraph{Time-dependent evolution.} Say a quantum system is governed by a time-dependent Hamiltonian $H(t)$. Then the Schr\"{o}dinger equation reads
\begin{equation}\begin{aligned}
    \frac{d\ket{\psi(t)}}{dt} = -iH(t)\ket{\psi(t)}
\end{aligned}\end{equation}
or equivalently
\begin{equation}
\begin{aligned}
    \frac{dU(t)}{dt} = -iH(t)U(t)
\end{aligned}
\end{equation}
where $\ket{\psi(t)} = U(t)\ket{\psi(0)}$. The corresponding integral form is
\begin{equation}\label{equation:integral_form_Sch_eq_U}
\begin{aligned}
    U(t) = U(0) - i\int_0^t d\tau\, H(\tau)U(\tau).
\end{aligned}
\end{equation}
Substituting in the RHS of the integral equation \cref{equation:integral_form_Sch_eq_U} for $U(\tau)$ in the integrand and iterating gives the Dyson series
\begin{equation}\begin{aligned}
    U(t) &= \sum_{n=0}^\infty (-i)^n \int_{0 \leq \tau_n \leq \dots \leq \tau_1 \leq t} d\tau_1 \dots d\tau_n\,H(\tau_1)\dots H(\tau_n).\\
    &=: \exp_\mathcal{T}\left( -i\int_0^t H(\tau)\,d\tau \right).
\end{aligned}\end{equation}
Here the integral $\int_{0 \leq \tau_n \leq \dots \leq \tau_1 \leq t} d\tau_1 \dots d\tau_n$ is shorthand for $\int_0^t d\tau_1 \int_0^{\tau_1} d\tau_2 \dots \int_0^{\tau_{n-1}} d\tau_n$. Furthermore note that
\begin{equation}\begin{aligned}
    \int_0^t d\tau_1 \int_0^{\tau_1} d\tau_2 \dots \int_0^{\tau_{n-1}} d\tau_n\,H(\tau_1)\dots H(\tau_n) = \frac{1}{n!} \int_0^t d\tau_1 \int_0^t d\tau_2 \dots \int_0^t d\tau_n\,\mathcal{T}\{H(\tau_1)\dots H(\tau_n)\}.
\end{aligned}\end{equation}
While the $n$-simplex variant (LHS) is cleaner conceptually, the $n$-hypercube variant (RHS) is occasionally useful for certain computations.

For a density operator the the Schr\"{o}dinger equation is
\begin{equation}\begin{aligned}
    \frac{d\rho(t)}{dt} = -i[H(t),\rho(t)].
\end{aligned}\end{equation}
Defining the Liouvillian (super)operator
\begin{equation}\begin{aligned}
    \mathcal{L}(t)(\rho) \coloneq -i[H(t),\rho],
\end{aligned}\end{equation}
we can rewrite Schr\"{o}dinger as
\begin{equation}\begin{aligned}
    \frac{d\rho(t)}{dt} = \mathcal{L}(t)\rho(t).
\end{aligned}\end{equation}
Substituting in the corresponding integral equation and iterating as in the pure state case gives the density operator Dyson series
\begin{equation}\begin{aligned}
    \rho(t) &= \exp_\mathcal{T}\left( \int_0^t \mathcal{L}(\tau)\,d\tau \right)\rho(0)\\
    &= \sum_{n=0}^\infty \int_{0 \leq \tau_n \leq \dots \leq \tau_1 \leq t} d\tau_1 \dots d\tau_n\,\mathcal{L}(\tau_1)\dots \mathcal{L}(\tau_n)\,\rho(0)\\
    &= \sum_{n=0}^\infty (-i)^n \int_{0 \leq \tau_n \leq \dots \leq \tau_1 \leq t} d\tau_1 \dots d\tau_n\,
    [H(\tau_1),[H(\tau_2),\dots[H(\tau_n),\rho(0)]\dots]].
\end{aligned}\end{equation}

If we write $\rho(t)$ in terms of the unitary channel $\mathcal{U}(t)$, i.e. $\rho(t)=\mathcal{U}(t)\rho(0)$, then the Dyson series for $\mathcal{U}(t)$ is
\begin{equation}\label{equation:dyson_series_unitary_channel}
\begin{aligned}
    \mathcal{U}(t) &= \exp_\mathcal{T}\left( \int_0^t \mathcal{L}(\tau)\,d\tau \right)\\
    &= \sum_{n=0}^\infty \int_{0 \leq \tau_n \leq \dots \leq \tau_1 \leq t} d\tau_1 \dots d\tau_n\,\mathcal{L}(\tau_1)\dots \mathcal{L}(\tau_n)\\
    &= \sum_{n=0}^\infty (-i)^n \int_{0 \leq \tau_n \leq \dots \leq \tau_1 \leq t} d\tau_1 \dots d\tau_n\,
    [H(\tau_1),[H(\tau_2),\dots[H(\tau_n),\,\cdot\,]\dots]].
\end{aligned}
\end{equation}

\section{Concentration Inequalities}\label[appendix]{appendix:conc_ineq}

\begin{theorem}[Hoeffding's Inequality]\label{theorem:Hoeffding}
    Let $(X_1,\dots,X_n)$ be independent bounded random variables, where $X_i \in [a_i,b_i]$. Then for every $\varepsilon>0$,
    \[
        \Pr\left[\left|\frac{1}{N}\sum_i X_i - \E\left( \frac{1}{N}\sum_i X_i \right)\right| \geq \varepsilon \right] \leq 2\exp\left( \frac{-2N^2\varepsilon^2}{\sum_{i=1}^N (b_i-a_i)^2} \right).
    \]
\end{theorem}


The following concentration inequality is also commonly called Borell-TIS. We refer the reader to Theorem 2.1.1 in \cite{adler2007random} or Theorem 5.8 in \cite{boucheron2013concentration} for more context and details on Borell-TIS and other related concentration inequalities.
\begin{theorem}[Borell-Tsirelson-Ibragimov-Sudakov Inequality]\label{theorem:Borell-TIS}
    Let $T$ be an indexing set and $(X_t)_{t \in T}$ be a centered Gaussian process almost surely bounded on $T$, i.e.
    \begin{equation}\begin{aligned}
        \Pr\left[ \sup_{t \in T} |X_t| < \infty \right] = 1.
    \end{aligned}\end{equation}
    Define $\sigma_T^2 \coloneq \sup_{t \in T} \vari(X_t)$. Then for every $\varepsilon > 0$, we have
    \begin{equation}\begin{aligned}
        \Pr\left[ \sup_{t \in T} X_t - \E\sup_{t \in T} X_t \geq \varepsilon \right] \leq \exp\left( \frac{-\varepsilon^2}{2\sigma_T^2} \right).
    \end{aligned}\end{equation}
    By symmetry the two-sided version is
    \begin{equation}\begin{aligned}
        \Pr\left[ \left|\sup_{t \in T} X_t - \E\sup_{t \in T} X_t\right| \geq \varepsilon \right] \leq 2\exp\left( \frac{-\varepsilon^2}{2\sigma_T^2} \right).
    \end{aligned}\end{equation}
\end{theorem}
The above result holds for \textit{arbitrary} indexing sets $T$, as long as the precondition is satisfied. Oftentimes $T$ is `nice' (e.g. compact), and one can show stronger conditions that implies the original a.s. boundedness precondition. For our purposes in this paper, we shall use one such stronger condition, namely $\E\sup_{t \in T} X_t < \infty$. This is justified by the following fact, which can be found in Theorem 2.1.2 of \cite{adler2007random}:
\begin{fact}\label{fact:equivalence_BTIS_preconditions}
    The following chain of implications involving the Borell-TIS precondition holds:
    \begin{equation}\begin{aligned}
        \E\sup_{t \in T} X_t < \infty &\iff \E\sup_{t \in T} |X_t| < \infty\\
        &\implies \Pr\left[ \sup_{t \in T} |X_t| < \infty \right] = 1.
    \end{aligned}\end{equation}
\end{fact}

\section{Stochastic Processes}\label[appendix]{appendix:stochastic_processes}
The goal of this section is such that by the end of it, the reader would hopefully be able to make sense of the equation
\begin{equation}\begin{aligned}
    d\bm{X}(t) = \bm{F}(\bm{X}(t),t)\,dt + \bm{G}(\bm{X}(t),t)\,d\bm{W}(t).
\end{aligned}\end{equation}
Intuitively this is not so hard: at any point in time $t$ the change in $\bm{X}(t)$ is simply the sum of a deterministic component $F(\bm{X}(t),t)\,dt$ plus a stochastic component $G(\bm{X}(t),t)\,d\bm{W}(t)$. A deeper understanding of SDEs requires more machinery to set up however, so here we provide a concise review for completeness and convenience for the reader.

There are many good resources on stochastic processes and SDEs, such as \cite{baldi2017stochastic,van2007stochastic,cohen2015stochastic,schilling2014brownian,bass2011stochastic,evans2012introduction,oksendal2013stochastic}. The definitions and theorems presented here are kept to a minimum and are all relevant -- either they are invoked directly in the results above, or provide intuition and enhance understanding on the concepts used in this paper. For pedagogical purposes we also gloss over certain technicalities at the expense of full precision. For instance, many statements here hold only \textit{almost surely}, i.e. up to sets of measure zero, which we neglect to mention. Also we shall take the liberty to move the time index $t$ into the subscript, as in $\bm{X}_t$, or into the argument, as in $\bm{X}(t)$. Both notations have their place, depending on the surrounding context.

\subsection{Basic notions}
Recall that a \textbf{probability space} is a triple $(\Omega,\mathcal{F},\P)$ where $\Omega$ is the \textbf{sample space}, $\mathcal{F}$ is a \textbf{$\bm{\sigma}$-algebra} of subsets of $\Omega$ and $\P$ is a \textbf{probability measure} on $\mathcal{F}$. The elements $\omega \in \Omega$ are called \textbf{outcomes}, and the elements $F \in \mathcal{F}$ are called \textbf{events}. A \textbf{random variable} $X: \Omega \to S$ is an $\mathcal{F}$-\textbf{measurable}\footnote{It is quite sufficient to simply think of measurable functions as continuous, or piecewise continuous functions in this paper. Measurability is preserved by the standard operations such as sums, products, compositions, supremums, infimums and limits.} map from the original probability space $(\Omega,\mathcal{F},\P)$ to the \textbf{state space} $(S,\mathcal{F}_S)$, i.e. $X^{-1}(A) \in \mathcal{F}$ for all $A \in \mathcal{F}_S$. The \textbf{$\bm{\sigma}$-algebra generated by $\bm{X}$}, $\sigma(X) \coloneq \{X^{-1}(A): A \in \mathcal{F}_S\}$, is the smallest sub-$\sigma$-algebra of $\mathcal{F}$ with respect to which $X$ is measurable; we think of it as containing all and only the information on $X$ -- there is no information extraneous to $X$. Without further specification we often assume $\mathcal{F}=\sigma(X)$. For concrete applications $S$ is almost always a metric space and $\mathcal{F}_S$ its Borel $\sigma$-algebra, the archetypal example being $(\mathbb{R}^d,\mathcal{B}(\mathbb{R}^d))$.

While $X$ is formally defined as a map, it is common to write things like $X = x$, or more generally $X \in A$, where $x \in S$ and $A \in \mathcal{F}_S$. What this is intended to convey is that an outcome $\omega \in \Omega$ for which $X(\omega)=x$, or more generally $X(\omega) \in A$, has occurred in the experiment. The random variable $X$ induces a probability measure $\P_X$ on the state space, defined by 
\begin{equation}
\begin{aligned}
    \P_X(A)\coloneq\P(X^{-1}(A)) \overset{\text{abbr.}}{=} \P(X \in A)
\end{aligned}
\end{equation}
for the events $A \in \mathcal{F}_S$. We call $\P_X$ the \textbf{distribution} of $X$, and write $X \sim \P_X$. The \textbf{expectation} of $X$ is 
\begin{equation}
\begin{aligned}
    \E \!X \coloneq \int_\Omega X(\omega) \P(d\omega) = \int_S x \P_X(dx).
\end{aligned}
\end{equation}
More generally, if $f$ is yet another measurable map from $(S,\mathcal{F}_S)$ to some other state space, we have
\begin{equation}\label{equation:expectations_over_state_space}
\begin{aligned}
    \E \!f(X) \coloneq \int_\Omega f(X(\omega)) \P(d\omega) = \int_S f(x) \P_X(dx).
\end{aligned}
\end{equation}
The second equality simply maintains the equivalence of integrating over two different domains -- the base space and the state space. If a \textbf{probability density function} (pdf) further exists for $\P_X$ (this does not always hold, most notably for the Wiener measure on path space we shall encounter below), we denote it by $p_X$ and also write $\underset{X \sim p_X}{\E}[f(X)] = \int_S f(x) \P_X(dx) = \int_S f(x)p_X(x)\, dx$.

Now we formally define the notion of a stochastic process, which is simply a collection of random variables indexed by time. This enables us to mathematically model processes whose evolution in time are random.
\begin{definition}\label{definition:stochastic_processes}
    A \textbf{stochastic process} is a tuple
    \begin{equation}
    \begin{aligned}
        (\Omega,\mathcal{F},(\mathcal{F}_t)_{t \in T},(X_t)_{t \in T},\P)
    \end{aligned}
    \end{equation}
    where
    \begin{enumerate}[i.]
        \item $(\Omega,\mathcal{F},\P)$ is a probability space;
        \item $T \subseteq \mathbb{R}_{\geq 0}$ indexes the times and is most often an interval;
        \item $(\mathcal{F}_t)_{t \in T}$ is a filtration, i.e. an increasing family of sub-$\sigma$-algebras of $\mathcal{F}$: $\mathcal{F}_s \subseteq \mathcal{F}_t$ whenever $s \leq t$;
        \item $X=(X_t)_{t \in T}$ is a family of random variables taking values in the state space $(S,\mathcal{F}_S)$, such that $(X_t)_t$ is adapted to the filtration $(\mathcal{F}_t)_t$, i.e. $X_t$ is $\mathcal{F}_t$-measurable for every $t$.
    \end{enumerate}
\end{definition}

Intuitively, since each $\mathcal{F}_t$ captures the information available up to time $t$, i.e. the events to which we know how to assign probabilities up to time $t$, $X=(X_t)_t$ being adapted to $(\mathcal{F}_t)_t$ means we are progressively gaining more knowledge on the process $X$. Without further specification we assume $\mathcal{F}_t = \sigma(X_\tau: \tau \leq t)$ is the minimal $\sigma$-algebra that contains the information on the stochastic process $X$ up to time $t$. This filtration $(\sigma(X_\tau: \tau \leq t))_{t \in T}$ is called the \textbf{history} of $X$.

It is customary to simply refer to $X=(X_t)_t$ as the stochastic process itself. For a fixed $\omega \in \Omega$, the set $\{X_t(\omega): t \in T\}$ is the \textbf{sample path} of $X$ associated with $\omega$. In this spirit, it is helpful to think of $X$ as a path-valued map:
\begin{equation}
\label{equation:path-valued_map}
\begin{aligned}
    X: \Omega &\to S^T\\
    \omega &\mapsto \{t \mapsto X_t(\omega): t \in T\}.
\end{aligned}
\end{equation}
Yet another useful and \href{https://en.wikipedia.org/wiki/Currying}{equivalent} way to think of $X$ is as the map $X: T \times \Omega \to S$ given by $X(t,\omega)\coloneq X_t(\omega)$. A process $X$ is \textbf{continuous} if every sample path $X(\omega)$ is continuous.

Let us now consider $\mathbb{R}^d$-valued continuous stochastic processes. Viewing $X$ as a path-valued map as in \cref{equation:path-valued_map}, we have that $X(\omega)$ takes values in $C(T,\mathbb{R}^d)$, the space of continuous functions $T \to \mathbb{R}^d$. If $T = [0,t]$ is a finite interval then we equip $C([0,t],\mathbb{R}^d)$ with the uniform metric $d(\gamma,\gamma')\coloneq \sup_{\tau \in [0,t]} \|\gamma(\tau)-\gamma'(\tau)\|$.
This metric then generates the corresponding Borel $\sigma$-algebra on $C(T,\mathbb{R}^d)$, which we denote by $\mathcal{F}_{C(T,\mathbb{R}^d)}$. With this, we can think of a stochastic process $X$ as a \textit{path-valued random variable}:
\begin{equation}\label{equation:path-valued_random_variable}
\begin{aligned}
    X: (\Omega,\mathcal{F},\P) \to (C(T,\mathbb{R}^d),\mathcal{F}_{C(T,\mathbb{R}^d)},\P_X)
\end{aligned}
\end{equation}
where $\P_X$ is the distribution of $X$ induced by $\P$.

Let us next discuss two important classes of stochastic processes, namely Markov processes and martingales. To do so it is necessary to first discuss the notion of conditional expectations. In introductory probability classes, the \textit{conditional} probability of event $A$, given event $B$ has occurred is $\P(A|B) \coloneq \frac{\P(A \cap B)}{\P(B)}$ provided $\P(B)>0$. Likewise, the conditional expectation of a random variable $X$, given $B$ has occurred, is $\E(X|B) \coloneq \frac{1}{\P(B)}\int_B X(\omega) \P(d\omega)$. We can `lift' this to asking: what is the expected value of a random variable $X$, given another random variable $Y$? That is, if by chance $\omega \in \Omega$ has occured, and all we know about $\omega$ is the value $Y(\omega)$, what is our best guess for the value $X(\omega)$?

We present the canonical motivating example, which highlights the essence of the above question, without getting bogged down by mathematical difficulties. Consider a simple random variable $Y$ on $(\Omega,\mathcal{F},\P)$ defined by $Y \coloneq \sum_{i=1}^n a_i\1_{A_i}$, where $A_i \in \mathcal{F}$\footnote{We are running out of symbols here -- sometimes $A$ will mean an event in $\Omega$, sometimes it will mean an event in the \textit{state space} $S$. Which one it is will be clear from the surrounding context.} and $\{A_i\}_{i=1,\dots,n}$ is a partition of $\Omega$. That is, $Y$ is piecewise constant on the subsets $A_i$. In this case, if $Y(\omega)$ is known (but not $\omega$), then we can tell which event $A_1,\dots,A_n$ has occurred. Only this being known, our best estimate for $X$ should be its expected value over the event that has occurred, i.e. the $A_i$ containing $Y(\omega)$. That is, we define
\begin{equation}
\begin{aligned}
    \E(X\,|\,Y) \coloneq \sum_{i=1}^n \E(X\,|\,A_i)\1_{A_i}.
\end{aligned}
\end{equation}
Note that the values of $Y$ didn't really play a role here, we would have obtained the same definition $\E(X\,|\,Y') = \sum_{i=1}^n \E(X\,|\,A_i)\1_{A_i}$ for $Y' \coloneq \sum_{i=1}^n a_i'\1_{A_i}$. The crucial ingredient was the $\sigma$-algebra $\sigma(A_1,\dots,A_n)$. Thus, we define 
\begin{equation}
\begin{aligned}
    \E(X\,|\,\sigma(A_1,\dots,A_n)) \coloneq \sum_{i=1}^n \E(X\,|\,A_i)\1_{A_i}.
\end{aligned}
\end{equation}
Note that 
\begin{enumerate}
    \item $\E(X\,|\,\sigma(A_1,\dots,A_n))$ is a random variable, and not simply a number, like $\E(X\,|\,A_i)$;
    \item $\E(X\,|\,\sigma(A_1,\dots,A_n))$ is $\sigma(A_1,\dots,A_n)$-measurable;
    \item For any $A \in \sigma(A_1,\dots,A_n)$, which is generally given by a union (need not be the full union) $A=\cup_i A_i$,
    \begin{equation}
    \begin{aligned}
        \int_A \E(X(\omega)\,|\,\sigma(A_1,\dots,A_n)) \P(d\omega) = \int_A X(\omega)\P(d\omega).
    \end{aligned}
    \end{equation}
\end{enumerate}
This motivates the most general definition of conditional expectation:
\begin{definition}\label{definition:conditional_expectation}
    Let $(\Omega,\mathcal{F},\P)$ be a probability space and let $\mathcal{A} \subseteq \mathcal{F}$ b a sub-$\sigma$-algebra. If $X$ is a random variable on $\Omega$, then its \textbf{conditional expectation} w.r.t. $\mathcal{A}$ is the random variable
    \begin{equation}
    \begin{aligned}
        \E(X\,|\,\mathcal{A})
    \end{aligned}
    \end{equation}
    which satisfies
    \begin{enumerate}[i.]
        \item $\E(X\,|\,\mathcal{A})$ is $\mathcal{A}$-measurable;
        \item $\int_A \E(X(\omega)\,|\,\mathcal{A}) \P(d\omega) = \int_A X(\omega)\P(d\omega)$ for all $A \in \mathcal{A}$.
    \end{enumerate}
    Consequently, given a random variable $Y$, 
    \begin{equation}
    \begin{aligned}
        \E(X\,|\,Y) \coloneq \E(X\,|\,\sigma(Y)).
    \end{aligned}
    \end{equation}
\end{definition}
It can be shown that $\E(X\,|\,\mathcal{A})$ exists and is unique. The \textbf{conditional probability} of a general event $B$ (which need not be $\mathcal{A}$-measurable) w.r.t. $\mathcal{A}$ is then a special case of conditional expectations:
\begin{equation}
\begin{aligned}
    \P(B\,|\,\mathcal{A}) \coloneq \E(\1_B\,|\,\mathcal{A}).
\end{aligned}
\end{equation}

Back to stochastic processes. A stochastic process $(X_t)_{t \in T}$ (taking values in the state space $(S,\mathcal{F}_S)$) is a \textbf{Markov process} if for each $t \in T$ and $s \leq t$, it holds that
\begin{equation}
\begin{aligned}
    \P(X_t \in A\,|\,\mathcal{F}_s) = \P(X_t \in A\,|\,X_s) \qquad \text{for every } A \in \mathcal{F}_S.
\end{aligned}
\end{equation}
This is equivalent to having
\begin{equation}
\begin{aligned}
    \E( f(X_t) \,|\,\mathcal{F}_s) = \E( f(X_t) \,|\,X_s) 
\end{aligned}
\end{equation}
for every measurable function $f$ on $S$. Recall that $\mathcal{F}_s = \sigma(X_\tau: \tau \leq s)$ is the available information on the process $X$ up to the present time $s$. Informally, this means the state of \textit{any aspect} of $X$ (captured by the function $f$) at the future time $t$ is only dependent on information available at the present time $s$, i.e. $\sigma(X_s)$, and all the previous information $\mathcal{F}_s$ is irrelevant.

A stochastic process $(X_t)_{t \in T}$ is a \textbf{martingale} if for each $t \in T$ and $s \leq t$, it holds that
\begin{equation}
\begin{aligned}
    \E( X_t \,|\,\mathcal{F}_s) = X_s.
\end{aligned}
\end{equation}
Note that while the defining characteristics of Markov processes and martingales look similar superficially, they are not the same. There exist Markov processes which are not martingales, and vice versa. As a simple example of the former, simply consider Markov processes with deterministic drifts. The presence of a drift kills off the martingale property. For the latter, while a martingale's conditional mean $\E( X_t \,|\,\mathcal{F}_s)$ depends only on the present value $X_s$, it could very well be that some other aspect of $X$ at a future time $t$, say the (conditional) variance $\E( \vari X_t \,|\,\mathcal{F}_s)$, depends on the entire history $\mathcal{F}_s$ and not just $\sigma(X_s)$.

We shall need the following inequality for martingales in one of the results above:
\begin{theorem}[Doob's Martingale Inequality]\label{theorem:Doob_martingale_ineq}
    Let $(X_t)_{t \geq 0}$ be a martingale with continuous sample paths. Then for every $t>0$ and every $1<p<\infty$ it holds that
    \begin{equation}\begin{aligned}
        \E\sup_{\tau \in [0,t]} |X_\tau|^p \leq \left(\frac{p}{p-1}\right)^p \E|X_t|^p.
    \end{aligned}\end{equation}
\end{theorem}
For most context on this inequality we direct the reader to Theorem 5.12 in \cite{baldi2017stochastic} or Theorem 3.6 in \cite{bass2011stochastic}

\subsection{The Wiener process}\label{subsection:Wiener_process}
The Wiener process, or Brownian motion, is perhaps the most recognizable continuous-time random process amongst mathematicians, physicists and engineers. It is also the archetypal example of continuous-time martingales, Markov processes, and Gaussian processes. Before stating the formal definition of Brownian motion, we first try to motivate it from physical considerations. The short treatment here is largely based on the excellent introduction in \cite{van2007stochastic}.

Brownian motion takes its name from the Scottish botanist Robert Brown, who in 1827 observed under a microscope the apparently erratic motion of pollen grains suspended in water. Brown first considered whether this motion might be biological in origin, but rejected that hypothesis after observing the same behaviour in glass powder and other inorganic materials. A physical explanation of Brown's observation was provided by Albert Einstein in 1905.\footnote{It is worth noting that a very similar analysis to Einstein's had appeared just a few years earlier, but in an entirely different context, due to French mathematician Louis Bachelier. Bachelier had proposed a stochastic model for French government bond prices, whose mathematical structure was essentially that of Brownian motion.} Einstein's argument is grounded in the atomic hypothesis, which in this context puts forth that the water comprises an enormous number of rapidly moving molecules, whose individual velocities fluctuate randomly because of thermal motion. Kinetic theory suggests that these velocities are independent and randomly distributed (with zero mean, since the total fluid has zero net velocity). A suspended microscopic particle (which is still large relative to the water molecules) is continually struck by these molecules, thus giving it net random displacements, which cumulatively give rise to the irregular trajectory observed by Brown.

How do we model this phenomenon? In what follows we reproduce Einstein's perhaps crude, yet physically quite revealing analysis. We suppose that the pollen particle is bombarded by $N$ water molecules per unit time, where each molecule imparts an i.i.d., zero-mean displacement $\xi_n$ to the particle. Zeroing the initial position of the pollen particle wlog, its position $x_t(N)$ at time $t$ is given by
\begin{equation}
\begin{aligned}
    x_t(N) = \sum_{i=1}^{Nt} \xi_i. 
\end{aligned}
\end{equation}
Since the pollen grain is large relative to a single water molecule, we are interested in the regime where the number of collisions $N$ is large, but the individual displacements $\xi_i$ are tiny. We capture the latter by assuming $\vari(\xi_i)$ takes the form $\vari(\xi_i) = G/N$. Note that
\begin{equation}
\begin{aligned}
    \E x_1(N)^2 = \vari\left( \sum_{i=1}^N \xi_i \right) = N \vari(\xi_i) = G,
\end{aligned}
\end{equation}
i.e. the mean-square displacement per unit time of the pollen is fixed to $G$ as $N$ increases. With an eye toward the central limit theorem (CLT), we rewrite the pollen position as
\begin{equation}
\begin{aligned}
    x_t(N) = \sqrt{Gt}\frac{\sum_{i=1}^{Nt} \Xi_i}{\sqrt{Nt}},
\end{aligned}
\end{equation}
where $\Xi_i = \xi_i\sqrt{N/G}$ are now standardized to unit variance. By the CLT, which states that for an i.i.d. sequence $(X_i)_i$ with $\E \!X_i=\mu$, $\vari X_i=\sigma^2$,
\begin{equation}
\begin{aligned}
    \sqrt{N}\left( \frac{X_1+ \dots +X_N}{N}-\mu \right) \overset{d}{\to} \mathcal{N}(0,\sigma^2),
\end{aligned}
\end{equation}
we see that in the limit $N \to \infty$, $\frac{\sum_{i=1}^{Nt} \Xi_i}{\sqrt{Nt}} \overset{d}{\to} \mathcal{N}(0,1)$. Thus $x_t(N)$ converges (in distribution) to a Gaussian of variance $Gt$. This was the crux of Einstein's analysis. Also observe that for $0 \leq s_1 < s_2 \leq t_1 < t_2$,
\begin{equation}
\begin{aligned}
    x_{t_2}(N) - x_{t_1}(N) = \sum_{i=Nt_1+1}^{Nt_2} \xi_i = \sum_{i=1}^{N(t_2-t_1)} \xi_i \overset{d}{\to} \mathcal{N}(0,G(t_2-t_1)),
\end{aligned}
\end{equation}
likewise for $x_{s_2}(N) - x_{s_1}(N)$. Furthermore, $x_{t_2}(N) - x_{t_1}(N)$ and $x_{s_2}(N) - x_{s_1}(N)$ are clearly independent since their associated increments $\{\xi_i\}_i$ are mutually disjoint. These features are reflected in \cref{definition:wiener_process} below.

While we have shown convergence in distribution, a fully rigorous mathematical analysis is still fraught with hurdles. Does the limit of the \textit{process} $t \mapsto x_t(N)$ as $N \to \infty$ even exist, in any suitable sense? This was first resolved by Norbert Wiener, whose formalization of Brownian motion we now present:
\begin{definition}\label{definition:wiener_process}
    A real-valued stochastic process $(\Omega,\mathcal{F},(\mathcal{F}_t)_{t \geq 0},(W_t)_{t \geq 0},\P)$ is a standard \textbf{Wiener process} if
    \begin{enumerate}[i.]
        \item $W_0 = 0$;
        \item $W_t-W_s \sim \mathcal{N}(0,t-s)$ for $0 \leq s \leq t$;
        \item Independent increments: $W_t-W_s$ is independent of $\mathcal{F}_s$ for $0 \leq s \leq t$.
    \end{enumerate}
    The generalization to higher dimensions is straightforward: an $\mathbb{R}^d$-valued process is a $d$-dimensional Wiener process if it again starts from $0$, has independent increments, and $W_t-W_s \sim \mathcal{N}(0,(t-s)I)$. The distribution of $W$ on path space $C([0,\infty),\mathbb{R}^d)$, $\P_W$, is called the \textbf{Wiener measure}.
\end{definition}
Properties ii,iii. imply that a Wiener process is a Gaussian process, i.e. the joint distribution of $(W_{t_1},\dots,W_{t_n})$ is Gaussian for any $n$ chosen times. Property iii. intuitively means the increments of a Wiener process after time $s$ are independent of its path up to time $s$. Furthermore, for times $s,t \geq 0$,
\begin{equation}
\begin{aligned}
    \cov(W_s,W_t) = \E W_sW_t = s \wedge t.
\end{aligned}
\end{equation}

A priori it is not obvious that a Wiener process exists, i.e. \cref{definition:wiener_process} can be satisfied. Of course, judging by its name it is clear that Wiener had also provided a construction. Many other constructions were subsequently developed; detailed treatments can be found in \cite{schilling2014brownian}. To help the reader get a feel of the Wiener process, we state a few distinguishing properties of its sample paths:
\begin{fact}\label{fact:path_properties_Wiener_process}
    The sample paths of the Wiener process $W(\omega)$
    \begin{enumerate}
        \item Are continuous. This is a consequence of Kolmogorov's continuity theorem.
        \item Are nowhere differentiable.
        \item Have infinite total variation on any finite interval $[a,b]$: $\text{TV}_{[a,b]}(W(\omega)) = \infty$.
        \item Have nonzero \textbf{quadratic variation}: for any sequence of partitions of $[a,b]$, $(\pi_n=\{a=t^n_0 < t^n_1 < \dots < t^n_{m_n} = b\})_n$ satisfying $|\pi_n| \to 0$,
        \begin{equation}
        \begin{aligned}
            \sum_{i=0}^{m_n-1} (W_{t^n_{i+1}}(\omega)-W_{t^n_i}(\omega))^2 \to b-a.\footnotemark
        \end{aligned}
        \end{equation}
        \footnotetext{Strictly speaking this convergence is in the $L^2$ sense, it is only after passing to an appropriately chosen subsequence that we get pointwise convergence.}
        Heuristically, we write this as
        \begin{equation}
        \begin{aligned}
            dW(\omega) = (dt)^{1/2}.
        \end{aligned}
        \end{equation}
    \end{enumerate}
\end{fact}

\begin{figure}[t]
    \centering
    \includegraphics[width=0.7\linewidth]{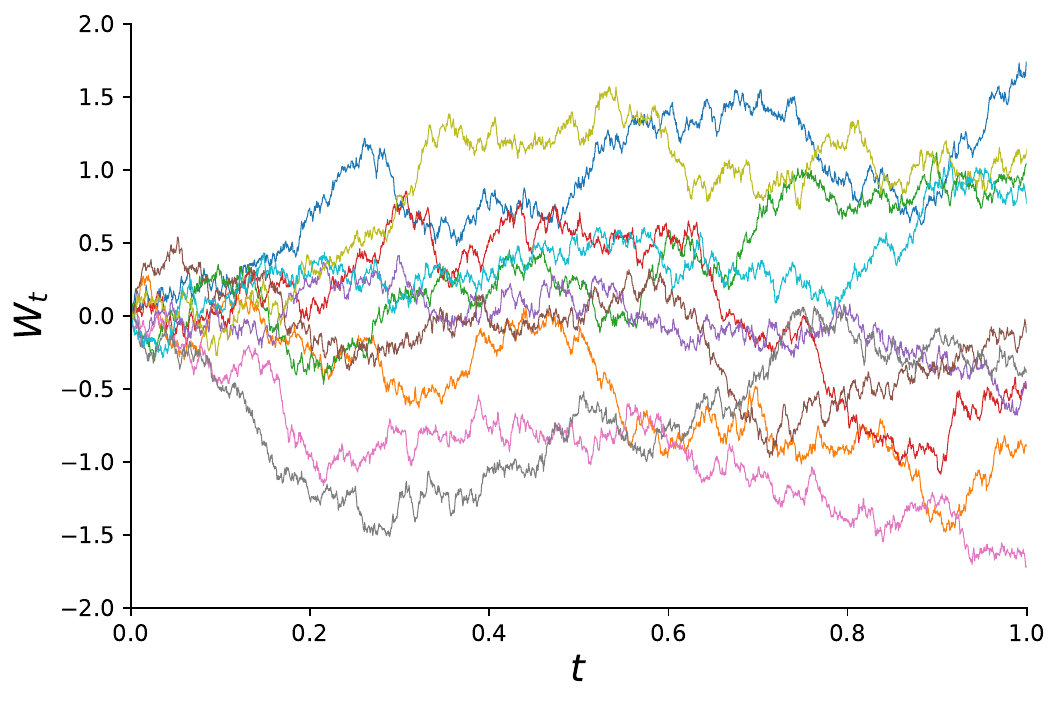}
    \caption{Ten sample paths of the 1D Wiener process, for $t \in [0,1]$. The paths are simulated using the discrete approximation $W_{t_{k+1}}=W_{t_k}+\sqrt{\Delta t}\,Z_k$ for $\Delta t = 1/2000$, where $Z_k \sim \mathcal{N}(0,1)$.
    }
    \label{figure:sample_paths_Brownian_motion}
\end{figure}

In short, while the paths of a Wiener process are continuous, they are \textit{extremely} irregular, see \cref{figure:sample_paths_Brownian_motion}. This is expressed in varying degrees of precision through properties 2-4. The source of this irregularity really comes from its quadratic variation. Intuitively, small changes in time $dt$ produces changes $dW(\omega)$ on the order of $\sqrt{dt} \gg dt$. Identifying as physicists for the moment, let us see how we can recover 2-4. heuristically from $dW(\omega) = (dt)^{1/2}$:
\begin{enumerate}
    \item[4.] Writing $dt \approx \frac{b-a}{n}$ for large $m_n \gets n$, $\sum_{i=0}^{n-1} (W_{t^n_{i+1}}(\omega)-W_{t^n_i}(\omega))^2 \approx n\cdot  (dW(\omega))^2 = n\,dt \approx b-a$.
    \item[3.] $\sum_{i=0}^{n-1} |dW(\omega)| \approx n \sqrt{dt} \approx \sqrt{n(b-a)} \to \infty$ as $n \to \infty$.
    \item[2.] $dW(\omega)/dt = \pm 1/\sqrt{dt} = \pm\infty$. Alternatively, if $W(\omega)$ were differentiable then it would have finite total variation on a finite interval, contradicting 3.
\end{enumerate}
For more precise statements on the regularity of Brownian paths and further interesting properties, see \cite{schilling2014brownian}. Furthermore, the Wiener process is both Markov and a martingale. The martingale property admits a very simple proof: for fixed $s\leq t$, we have
\begin{equation}
\begin{aligned}
    \E(W_t \,|\, \mathcal{F}_s) &= \E(W_s + (W_t-W_s) \,|\, \mathcal{F}_s)\\
    &= \E(W_s \,|\, \mathcal{F}_s) + \E(W_t-W_s \,|\, \mathcal{F}_s)\\
    &= W_s + 0
\end{aligned}
\end{equation}
where in the last equality we have used the independent increments property of Wiener processes. The formal proof of the Markov property is more involved, and we direct the reader to any of the excellent texts mentioned above. But intuitively, the independent increments property is again the underlying reason: writing again
\begin{equation}
\begin{aligned}
    W_t = \underbrace{W_s}_{\text{current position}} + \underbrace{W_t-W_s}_{\text{increment}},
\end{aligned}
\end{equation}
we see that $W_t$ depends only on the current position, and an increment that is independent of the entire history up to time $s$, $\mathcal{F}_s$.

\subsection{Stochastic integrals and calculus}
To make sense of SDEs, we need to first make sense of objects like
\begin{equation}
\begin{aligned}
    \int_0^t X_\tau \,dW_\tau,
\end{aligned}
\end{equation}
or, with $\tau$ explicitly as an argument,
\begin{equation}
\begin{aligned}
    \int_0^t X(\tau)\,dW(\tau).
\end{aligned}
\end{equation}

In elementary calculus, the Riemann-Stieltjes integral $\int_0^t f(\tau)\,dg(\tau)$ is defined as follows. We first construct a partition, or mesh, of the interval $[0,t]$, i.e. write $[0,t]=\bigcup_{i=1}^n [\tau_{i-1},\tau_i]$ with $\tau_0=0$ and $\tau_n=t$. For this mesh the approximating Riemann-Stieltjes sum is
\begin{equation}
\begin{aligned}
    S_n(f,g) \coloneq \sum_{i=1}^n f(c_i)(g(\tau_i)-g(\tau_{i-1}))
\end{aligned}
\end{equation}
where $c_i \in [\tau_{i-1},\tau_i]$ can be \textit{any} point in $[\tau_{i-1},\tau_i]$. Taking the limit as the mesh size goes to zero, we have
\begin{equation}
\begin{aligned}
    \int_0^t f(\tau)\,dg(\tau) \coloneq \lim_{n \to \infty} S_n(f,g).
\end{aligned}
\end{equation}
This limit, if it exists, is independent of the sequence of meshes used and the choice of $c_i$'s. In elementary calculus, most often $f$ is continuous and $g$ is of bounded variation, and these conditions are sufficient for the existence of $\int_0^t f(\tau)\,dg(\tau)$.

We would like to replicate this approach in our attempt to define stochastic integrals $\int_0^t X(\tau)\,dW(\tau)$. First note that $X(\tau)$ and $W(\tau)$ are random variables, so the stochastic integral is also a random variable. Naively, one attempts to define the stochastic integral \textit{pointwise}, as
\begin{equation}
\begin{aligned}
    \left(\int_0^t X(\tau)\,dW(\tau)\right)(\omega) &\coloneq \int_0^t X(\tau,\omega)\,dW(\tau,\omega)\\
    &\coloneq \lim_{n \to \infty} \sum_{i=1}^n X(c_i,\omega)(W(\tau_i,\omega)-W(\tau_{i-1},\omega))
\end{aligned}
\end{equation}
for each fixed outcome $\omega \in \Omega$. But doing so one quickly runs into difficulties. The culprit is the infinite variation property of $W(\omega)$, see \cref{fact:path_properties_Wiener_process}. As $n \to \infty$, the mesh size approaches zero and the RHS of the equation above blows up. To overcome this our forebears proceeded as follows.

\paragraph{Step 1.} Consider first \textit{simple} processes, which are processes taking the form
\begin{equation}
\begin{aligned}
    X(t) = \sum_{i=1}^n X_i\1_{[\tau_i,\tau_{i+1})}(t).
\end{aligned}
\end{equation}
That is, $X(t)$ looks like a step function (for each fixed $\omega$ -- note that the $X_i$'s are random variables). For such $X(t)$, define
\begin{equation}
\begin{aligned}
    \int_0^t X(\tau)\,dW(\tau) \coloneq \sum_{i=1}^n X_i ( W_{t_{i+1}}-W_{t_i}).
\end{aligned}
\end{equation}

\paragraph{Step 2.} Next we attempt to define the stochastic integral for general processes $X(t)$. Let the time interval $[0,t]$ be given. We first define a sequence of meshes $[0,t]=\bigcup_{i=1}^n [\tau_{i-1},\tau_i]$ (where the mesh size $\sup_{i=1,\dots,n} |\tau_{i+1}-\tau_i| \to 0$ as $n \to \infty$). It is a fact (which we shall not prove here) that we have the $L^2$-convergence
\begin{equation}\label{equation:L2_approximation_X(t)}
\begin{aligned}
    \E \left( \int_0^t |X(\tau)-X_n(\tau)|^2\,d\tau \right) \xrightarrow{n \to \infty} 0
\end{aligned}
\end{equation}
where the
\begin{equation}
\begin{aligned}
    X_n(s) = \sum_{i=1}^n X(c_i)\1_{[\tau_i,\tau_{i+1})}(s)
\end{aligned}
\end{equation}
are simple processes corresponding to the $n$-th mesh, which holds for \textit{any} choice of $c_i \in [\tau_{i-1},\tau_i)$. So far this has nothing to do with any stochastic integral of $X$ whatsoever. It is at this point where we make the \textit{convention} to choose $c_i = \tau_i$ for the approximating simple processes $X_n$. This convention is the \textit{It\^{o}} convention, and it is not without good reason as we shall see below. The other main convention (out of infinitely many) is the \textit{Stratonovich} convention, in which we choose $c_i = \frac{\tau_i+\tau_{i+1}}{2}$.

With the It\^{o} convention in mind, let us consider simple processes $X_{\text{simp}}(t) = \sum_{i=1}^n X_i\1_{[\tau_i,\tau_{i+1})}(t)$ where the $X_i$ are $\mathcal{F}_{t_i}$-measurable. For such simple processes we have $\E \left(\int_0^t X_{\text{simp}}(\tau)\,dW(\tau) \right) = 0$, because
\begin{equation}
\begin{aligned}
    \E \left(\int_0^t X_{\text{simp}}(\tau)\,dW(\tau) \right) &= \sum_{i=1}^n \E\left( X_i( W_{t_{i+1}}-W_{t_i})\right)\\
    &= \sum_{i=1}^n \E(X_i) \underbrace{\E (W_{t_{i+1}}-W_{t_i})}_{=0}.
\end{aligned}
\end{equation}
We have made use of the independence of $X_i$ and $W_{t_{i+1}}-W_{t_i}$, which is a consequence of the $\mathcal{F}_{t_i}$-measurability of $X_i$. With this, we can further proceed to show (calculations are tedious but straightforward) the \textbf{It\^{o} isometry}
\begin{equation}\label{equation:ito_isometry}
\begin{aligned}
    \E \left( \int_0^t X_{\text{simp}}(\tau)\,dW(\tau) \right)^2 = \E \left( \int_0^t X_{\text{simp}}(\tau)^2\,d\tau \right).
\end{aligned}
\end{equation}

\paragraph{Step 3.} The plan now is to leverage this It\^{o} isometry and the $L^2$-approximation fact \cref{equation:L2_approximation_X(t)} to define the stochastic integral (which we also call the It\^{o} integral due to the convention made) for a general stochastic process $X(t)$. This is straightforward: for the approximating simple processes $X_n, X_m$ of $X$, clearly the difference $X_n-X_m$ is also a simple process (a `finer' one, so to speak). Then by the It\^{o} isometry \cref{equation:ito_isometry} we have
\begin{equation}
\begin{aligned}
    \E \left( \int_0^t X_n(\tau)-X_m(\tau) \,dW(\tau) \right)^2 = \E \left( \int_0^t (X_n(\tau)-X_m(\tau))^2\,d\tau \right) \xrightarrow{n,m \to \infty} 0,
\end{aligned}
\end{equation}
i.e. the sequence 
\begin{equation}
\begin{aligned}
    \left(\int_0^t X_n(\tau)\,dW(\tau)\right)_n
\end{aligned}
\end{equation}
is Cauchy in the metric space $L^2(\Omega)$ (equipped with the $L^2$-norm). Since $L^2(\Omega)$ is complete, we have
\begin{equation}
\begin{aligned}
    \int_0^t X(\tau)\,dW(\tau) \coloneq L^2\text{-}\lim_{n \to \infty} \int_0^t X_n(\tau)\,dW(\tau).
\end{aligned}
\end{equation}
This is how we define the It\^{o} stochastic integral. We emphasize that the limit is \textit{in the $L^2$-sense}. One can also show that this limit does not depend on the choice of sequence of meshes used. The stochastic integral satisfies the following properties:
\begin{enumerate}
    \item Linearity, $\int_0^t aX(\tau)+bY(\tau) \,dW(\tau) = a\int_0^t X(\tau)\,dW(\tau) + b\int_0^t Y(\tau) \,dW(\tau)$.
    \item $\E\left( \int_0^t X(\tau)\,dW(\tau) \right)=0$.
    \item It\^{o} isometry, $\E \left( \int_0^t X(\tau)\,dW(\tau) \right)^2 = \E \left( \int_0^t X(\tau)^2\,d\tau \right)$.
\end{enumerate}
In this paper, by `stochastic integral' we mean the It\^{o} stochastic integral. 

We now remark briefly on It\^{o} vs Stratonovich. As mentioned above, the key difference between these two is that in the approximating sequence of simple processes for $X$, It\^{o} uses the sampling point $c_i=\tau_i$, while Stratonovich uses $c_i=\frac{\tau_i+\tau_{i+1}}{2}$ (the choice of $c_i$ is irrelevant in Riemann-Stieltjes). What does this mean, physically? The It\^{o} convention is based on the idea that since $t$ represents time, and at time $\tau_i$ we do not yet know what $W(s)$ will do on $[\tau_i,\tau_{i+1}]$, it is least presumptuous to use the known value of $X(\tau_i)$ in the approximation. As a consequence, the It\^{o} convention leads to many mathematically nice properties of the stochastic integral, perhaps most importantly that the process $\left( \int_0^t X(\tau)\,dW(\tau) \right)_t$ is a martingale. The Stratonovich convention is not without its strengths too, but we shall not pursue it here, instead directing the reader to the textbooks, and in particular the note \cite{smith2022itovsstrat}.

The It\^{o} convention is nice, true, but it also leads to an infamous side effect, namely
\begin{theorem}[Ito's lemma]\label{theorem:Ito_lemma}
    The following result is known as \textbf{It\^{o}'s chain rule}, the chain rule for stochastic calculus. Given
    \begin{align*}
        dX(t) = F(X(t),t)\,dt + G(X(t),t)\,dW(t),
    \end{align*}
    which is just shorthand (coming soon in the next subsection) for 
    \begin{align*}
        X(t) = \int_0^t F(X(\tau),\tau)\,d\tau + \int_0^t G(X(\tau),\tau)\,dW(\tau).
    \end{align*}
    Let $f = f(x,t)$ be a `nice' scalar function, i.e. assume smoothness for simplicity.
    Then
    \begin{align*}
        df(X,t) = \left(f_t + f_xF + \frac{1}{2}f_{xx}G^2\right)\,dt + f_xG\,dW(t).
    \end{align*}
\end{theorem}
Contrast this to the classical chain rule, which would have that $df(X,t) = \left(f_t + f_xF\right)\,dt + f_xG\,dW(t)$. The extra $\frac{1}{2}f_{xx}G^2\,dt$ term is called the It\^{o} correction term, and ultimately is an artefact of It\^{o}'s convention and the quadratic variation property of $W(t)$, see \cref{fact:path_properties_Wiener_process}. We shall not discuss this aspect further, since it is not used in this paper. We will comment however that using the Stratonovich convention, there is no correction term and the Stratonovich chain rule is as normal.

The extension of the It\^{o} stochastic integral to higher dimensional processes is straightforward, and there are corresponding high-dimensional extensions of It\^{o}'s lemma. We leave that for future editions of this paper.

\subsection{Stochastic differential equations}\label{subsection:SDEs}
We have finally assembled the necessary tools to unpack the meaning of
\begin{equation}\begin{aligned}
    dX(t) = F(X(t),t)\,dt + G(X(t),t)\,dW(t).
\end{aligned}\end{equation}
As mentioned above, this intuitively means that at time $t$ the change $dX(t)$ comes from a deterministic component $F(X(t),t)\,dt$ and a stochastic component $G(X(t),t)\,dW(t)$. Without the stochastic component we simply have an ordinary ODE $dX(t) = F(X(t),t)\,dt$, so $X(t)$ is actually deterministic and its entire trajectory is determined by the ODE and the initial value $X(0)$, i.e. $X(t) = X(0) + \int_0^t F(X(\tau),\tau)\,d\tau$. With the noise component $dW(t)$, $X(t)$ becomes a bona fide random variable, since for each run its increments are random. Formally,

\begin{definition}\label{definition:SDEs}
    Let $\bm{F}(x,t) = (F_i(x,t))_{i \in [m]}$ and $\bm{G}(x,t) = (G_{ij}(x,t))_{i \in [m], j \in [\chi]}$ be measurable functions defined on $\mathbb{R}^m \times T$. We say the stochastic process $X=(X(t))_{t \in T} = (X_{1}(t),\dots,X_{m}(t))_{t \in T}$ is governed by/a solution of the \textbf{stochastic differential equation (SDE)}
    \begin{equation}
    \label{equation:SDE_def1}
    \begin{aligned}
        d\bm{X}(t) &= \bm{F}(\bm{X}(t),t)\,dt + \bm{G}(\bm{X}(t),t)\,d\bm{W}(t)\\
        \bm{X}(0) &= \bm{X}_0 \in \mathbb{R}^m
    \end{aligned}
    \end{equation}
    if for every $t \in T$ it holds that
    \begin{equation}
    \label{equation:SDE_def2}
    \begin{aligned}
        \bm{X}(t) = \bm{X}_0 + \int_0^t \bm{F}(\bm{X}(\tau),\tau)\,d\tau + \int_0^t \bm{G}(\bm{X}(\tau),\tau)\,d\bm{W}(\tau).
    \end{aligned}
    \end{equation}
    Equivalently, in component form this reads
    \begin{equation}\begin{aligned}
        dX_{i}(t) &= F_i(\bm{X}(t),t)\,dt + \sum_{j=1}^\chi G_{ij}(\bm{X}(t),t)\,dW_{j}(t)\\
        X_{i}(t) &= X_{i,0} + \int_0^t F_i(\bm{X}(\tau),\tau)\,d\tau + \sum_{j=1}^\chi \int_0^t G_{ij}(\bm{X}(\tau),\tau)\,dW_{j}(\tau).
    \end{aligned}\end{equation}
    That is, \cref{equation:SDE_def1} is really just shorthand for \cref{equation:SDE_def2}. Here $\bm{W}$ is a $\chi$-dimensional Wiener process, the time-integral is an ordinary Riemann/Lebesgue integral, and the $\bm{W}(\tau)$-integral is the It\^{o} stochastic integral defined above. $\bm{F},\bm{G}$ are called the drift and diffusion (coefficients) respectively.

    Note that the components $X_i$ are in general not independent, since they could be driven by overlapping Wiener processes $W_j$'s.
\end{definition}
Here we shall simply assume the existence and uniqueness of solutions to a given SDE, which generally hold if the drift $\bm{F}$ and diffusion $\bm{G}$ terms are `sufficiently nice'. We refer the reader to the textbooks mentioned above for the required nice properties and the proofs, based on a mix of probability and PDE theory. Here are a few examples of (single-component) SDEs:
\begin{example}\hfill
\label{example:examples_of_SDEs}
\begin{enumerate}
    \item Brownian motion with drift: $dX(t) = F\,dt + G\,dW(t)$. Here $F,G$ are constants.
    
    \item Ornstein-Uhlenbeck process: $dX(t) = \gamma(X_0-X(t))\,dt + G\,dW(t)$. Here $\gamma>0$ and $G$ are constants.

    \item Geometric Brownian motion: $dX(t) = FX(t)\,dt + GX(t)\,dW(t)$. Here $F,G$ are constants.
\end{enumerate}
\end{example}

Stochastic processes which are governed by SDEs satisfy many nice properties:
\begin{fact}\label{fact:properties_SDE_processes}
Let $X$ be a stochastic processes given in \cref{definition:SDEs}. Then
\begin{enumerate}
    \item The sample paths $X(\omega)$ are continuous.
    
    \item Drift-martingale decomposition:
    \begin{equation}
    \begin{aligned}
        \bm{X}(t) = \bm{X}_0 + \underbrace{\int_0^t \bm{F}(\bm{X}(\tau),\tau)\,d\tau}_{\text{drift}} + \underbrace{\int_0^t \bm{G}(\bm{X}(\tau),\tau)\,d\bm{W}(\tau)}_{\text{martingale}}.
    \end{aligned}
    \end{equation}

    \item $X$ is a Markov process.
\end{enumerate}
That $X$ is Markov can perhaps be intuitively understood by its increments $dX(t)$ depending only on the current state $X(t)$ through $F(X(t),t)$ and $G(X(t),t)$. The fact that the sample paths $X(\omega)$ are continuous means that SDEs driven by Wiener processes are not suitable to model processes that have jumps. These must be modelled using SDEs driven by other processes, say the L\'{e}vy process.
\end{fact}

Finally, since stochastic processes governed by SDEs are Markov, this makes them amenable to study using tools from Markov process theory. One particularly important aspect of a process $\bm{X}$ is its associated pdf, $p(x,t)$. It can be shown that the time evolution of $p(x,t)$ is governed by the partial differential equation known as the \textbf{Fokker-Planck equation} (also known as the Kolmogorov forward equation):
\begin{equation}
\label{equation:Fokker-Planck}
\begin{aligned}
    \frac{\partial p(x,t)}{\partial t} = -\sum_{i=1}^m \frac{\partial}{\partial x_i} \left[F_i(x,t)p(x,t)\right] + \frac{1}{2}\sum_{i,j=1}^m \frac{\partial^2}{\partial x_i \partial x_j} [(\bm{G}\bm{G}^T)_{ij}(x,t)p(x,t)].
\end{aligned}
\end{equation}
In the special case $m=\chi=1$ we are reduced to
\begin{equation}
\begin{aligned}
    \frac{\partial p_X(x,t)}{\partial t} = -\frac{\partial}{\partial x} \left[F(x,t)p_X(x,t)\right] + \frac{1}{2}\frac{\partial^2}{\partial x^2} [G(x,t)^2p_X(x,t)].
\end{aligned}
\end{equation}
More generally, the multi-time joint pdf
\begin{equation}
\begin{aligned}
    p(x_n,\tau_n; \dots ;x_1,\tau_1) = \prod_{i=2}^n p(x_i,\tau_i|x_{i-1},\tau_{i-1})\cdot p(x_1,\tau_1).
\end{aligned}
\end{equation}
can be obtained from Fokker-Planck because the transition pdfs $p(x_i,\tau_i|x_{i-1},\tau_{i-1})$ can all be obtained from \cref{equation:Fokker-Planck}, with respective initial conditions $p(x_i,\tau_{i-1}|x_{i-1},\tau_{i-1})=\delta(x_i-x_{i-1})$.

\end{document}